\documentclass[11pt]{article}
\usepackage[left=1 in, right=1 in, top=1 in, bottom = 1 in]{geometry}

\usepackage[utf8]{inputenc} % allow utf-8 input
\usepackage[T1]{fontenc}    % use 8-bit T1 fonts
\usepackage{url}            % simple URL typesetting
\usepackage{booktabs}       % professional-quality tables
\usepackage{amsfonts}       % blackboard math symbols
\usepackage{nicefrac}       % compact symbols for 1/2, etc.
\usepackage{microtype}      % microtypography
\usepackage{xcolor}         % colors
\usepackage{makecell}
\usepackage{multirow}
\usepackage{graphicx}
\usepackage{subfigure}
\usepackage{booktabs, tabularx}
\usepackage{tabularray}
\usepackage{authblk}
\usepackage[numbers]{natbib}
\usepackage{multicol}
\usepackage{enumitem}
 \usepackage{soul}
 
\usepackage{amsmath, bm}
\usepackage{amssymb}
\usepackage{mathtools}
\usepackage{amsthm}
\usepackage{algorithm}
\usepackage{algorithmic}
\usepackage{color-edits}
\usepackage{array}
\usepackage{makecell}

\theoremstyle{plain}
\newtheorem*{theorem*}{Theorem}
\newtheorem{theorem}{Theorem}[section]
\newtheorem{proposition}[theorem]{Proposition}
\newtheorem{lemma}[theorem]{Lemma}

\theoremstyle{definition}

\theoremstyle{remark}

\usepackage{thmtools}
\usepackage{thm-restate}

\declaretheorem[name=Lemma, sibling = theorem]{lem}

\definecolor{DarkGreen}{rgb}{0.1,0.5,0.1} 

\usepackage{hyperref}
\hypersetup{colorlinks=true,linkcolor=DarkGreen,citecolor=DarkGreen,urlcolor=DarkGreen}

\newcommand{\N}{\mathcal{N}}
\newcommand{\J}{\mathcal{J}}
\newcommand{\predJ}{\hat{\mathcal{J}}}

\newcommand{\predT}{\texttt{OPT}(\predJ)}
\newcommand{\trueT}{\texttt{OPT}(\J)}

\newcommand{\A}{\mathcal{A}}

\newcommand{\OPT}{\texttt{OPT}}

\newcommand{\tpe}{\textsc{TPE}}
\newcommand{\tpes}{\textsc{TPE-S}}

\newcommand{\tl}{t_{\lambda}}

\newcommand{\cst}{\text{cost}}

\newcommand{\K}{\mathcal{K}}
\newcommand{\Js}{\J^{\text{shift}}_{\geq \tl}}
\newcommand{\Ss}{S^{\text{shift}}}
\newcommand{\sft}{\frac{\eta^{\text{shift}}}{\beta(\predJ)} \cdot\frac{\texttt{OPT}(\predJ)}{|\predJ|}}

\begin{document}

\title{Energy-Efficient Scheduling with Predictions}

\author[1]{Eric Balkanski\thanks{ {\tt eb3224@columbia.edu}.}}
\author[1]{Noemie Perivier\thanks{ {\tt np2708@columbia.edu}.}}
\author[1]{Clifford Stein\thanks{ {\tt cliff@ieor.columbia.edu}.}}
\author[1]{Hao-Ting Wei\thanks{ {\tt hw2738@columbia.edu}.}}
\affil[1]{Department of Industrial Engineering and Operations Research, Columbia University}

% \author{%
%   Eric Balkanski \\
%   Columbia University \\
%   \texttt{eb3224@columbia.edu} \\
%   % examples of more authors
%   \And
%   Noemie Perivier \\
%   Columbia University \\
%   \texttt{np2708@columbia.edu} \\
%     \AND
%   Clifford Stein \\
%   Columbia University \\
%   \texttt{cliff@ieor.columbia.edu} \\
%     \And
%   Hao-Ting Wei \\
%   Columbia University \\
%   \texttt{hw2738@columbia.edu} \\
% }
\date{}
\maketitle

\begin{abstract}
An important goal of modern scheduling systems is to efficiently manage power usage. In energy-efficient scheduling,   the operating system controls the speed at which a machine is processing jobs with the dual objective of minimizing energy consumption and optimizing the quality of service cost of the resulting schedule. Since machine-learned predictions about future requests can often be learned from historical data, a recent line of work on learning-augmented algorithms aims to achieve improved performance guarantees by leveraging predictions.   In particular, for energy-efficient scheduling, \citet{BamasMRS20} and \citet{antoniadis2021novel} designed algorithms with predictions for the energy minimization with deadlines problem and achieved an improved competitive ratio when the prediction error is small while also maintaining worst-case bounds even when the prediction error is arbitrarily large.

In this paper, we consider a general setting for energy-efficient scheduling and provide a flexible learning-augmented algorithmic framework that takes as input an offline and an online algorithm for the desired energy-efficient scheduling problem. We show that, when the prediction error is small, this framework gives improved competitive ratios for many different energy-efficient scheduling problems, including energy minimization with deadlines, while also maintaining a bounded competitive ratio regardless of the prediction error. Finally, we empirically demonstrate that this framework achieves an improved performance on real and synthetic datasets.
\end{abstract}

\newpage

\section{Introduction}
\vspace{-.05cm}
Large data centers and machine learning models are important contributors to the growing impact that computing systems have on climate change. An important goal is thus to efficiently manage power usage in order to not only complete computing tasks in a timely manner but to also minimize energy consumption. In many operating systems, this tradeoff can be controlled by carefully scaling the speed at which jobs run. An extensive area of scheduling has studied such online (and offline) speed scaling problems (see, e.g., \cite{albers2010energy}). Since the speed of many processors is approximately the cube root of their power~\cite{mudge2001power,brooks2000power}, these works assume that the power of a processor is equal to speed to some power $\alpha \geq 1$, where $\alpha$ is thought of as being approximately $3$~\cite{yao1995scheduling,bansal2007speed} and the total energy consumption is power integrated over time.

Online energy-efficient scheduling algorithms have mostly been evaluated using competitive analysis, which provides robust guarantees that hold for any instance. However, since competitive analysis evaluates algorithms over worst-case instances, it can often be pessimistic. In particular, it ignores the fact that, in the context of scheduling,  future computation requests to computing systems can often be estimated from historical data. A recent line of work on algorithms with predictions aims to address this limitation by assuming that the algorithm designer is given access to machine-learned predictions about the input. In the context of online algorithms, where this line of work has been particularly active, the predictions are about future requests and the goal is to achieve improved competitive ratios when the predictions are accurate (consistency), while also maintaining the benefits of worst-case analysis with guarantees that hold even when the predictions are arbitrarily wrong (robustness).   

In this framework with predictions, \citet{BamasMRS20}  and \citet{antoniadis2021novel}  recently studied the energy minimization with deadlines problem, which is a classical setting for energy-efficient scheduling  (see, e.g., \cite{yao1995scheduling,bansal2007speed}). However, there are many scenarios where the jobs do not have strict deadlines and the goal is instead to minimize the job response time. In energy plus flow time minimization problems, which are another family of energy-efficient scheduling problems that have been extensively studied in the setting without predictions, the objective is to minimize a combination of the energy consumption and the flow time of the jobs, which is the difference between their release date and completion time (see, e.g., \cite{albers2007energy,bansal2010speed,bansal2013speed,andrew2009optimal}).

In this paper, we study a general energy-efficient scheduling problem that we augment with predictions.  This general problem includes both energy minimization with deadlines, which has been previously studied with predictions, and energy plus flow time minimization, which has not been previously studied with predictions, as well as many other variants and generalizations. In particular, the flow time problem with predictions introduces challenges that require novel learning-augmented scheduling algorithms (see Section~\ref{sec:framework} for additional discussion).

\vspace{-.1cm}
\subsection{Our results} 

\vspace{-.05cm}

An instance of the General Energy-efficient Scheduling (GES) problem is described by a collection $\J$ of $n$ jobs and an arbitrary quality of service cost function $F$. Each job  $(j, r_j, p_j) \in \J$ consists of a release time $r_j$, a processing time $p_j$, and an identifier $j$ (and potentially other parameters such as weights $v_j$ or deadlines $d_j$). A schedule $S$ is specified by the speeds $s_j(t)$ at which job $j$ is processed by the machine at time $t$. The goal is to find a schedule of minimum cost $E(S) + F(S, \J)$, where the energy consumption of a schedule is  $E(S) = \int_{t \geq 0} (\sum_j s_j(t))^{\alpha} \text{dt}$, for some constant $\alpha > 0$.  %Such quality costs capture both the energy minimization with deadlines problem, where each job $j$ has a deadline $d_j$ and $F(S, \J) = \infty$ if a job $j$ is not completed by $d_j$, and the energy plus flow time minimization problem, where $F(S, \J) = \sum_{j \in \J} c_j - r_j$ where $c_j$ is the completion time of job $j$, as well as many of their extensions. 
In the general energy-efficient scheduling with predictions (GESP) problem, the algorithm is given at time $t=0$  a collection $\hat{\J}$ of $\hat{n}$ predicted jobs $(j, \hat{r}_j, \hat{p}_j)$, which is a similar prediction model as in \citep{BamasMRS20}. For all our results, we assume that the quality cost function $F$  is monotone and subadditive, which are two mild conditions that are satisfied for the problems with flow times and with deadlines.

\paragraph{Near-optimal consistency and bounded robustness.}
Our first goal is to design an algorithm for the GESP problem that achieves a good tradeoff between its consistency  (competitive ratio when the predictions are exactly correct) and robustness (competitive ratio when the predictions are arbitrarily wrong). %Ideally, for an instance $I$ of GES, an algorithm with predictions  would achieve a consistency of $1$ and a robustness equal to the best competitive ratio known for $I$ (without predictions).  
%Since \citet{BamasMRS20} shows that there is no algorithm with predictions for energy minimization with deadlines that has consistency $1$ and bounded robustness, this impossibility result also applies to GESP. In Appendix~\ref{app:consistency_robustness}, we also show a necessary trade-off between consistency and robustness for energy plus flow time minimization with predictions. 
Our first main result is that for any instance of the GES problem for which there exists a constant competitive algorithm and an optimal offline algorithm, there is an algorithm with predictions that is  $1 + \epsilon$ consistent and $O(1)$ robust for any constant $\epsilon \in (0,1]$ (Corollary~\ref{cor:consistency}). Since problems with the flow time and the problem with deadlines admit constant-competitive algorithms,  we achieve a consistency that is arbitrarily close to optimal while also maintaining constant robustness for these problems (see Table~\ref{tab:results} for a summary of problem-specific upper bounds). We complement this result by showing that there is a necessary trade-off between
consistency and robustness for the flow time problem: for any $\lambda > 0$, there is no $1+ \lambda$-consistent algorithm that is $o(\sqrt{1 + 1/\lambda})$-robust (Appendix~\ref{app:lowerbound-tradeoff}).

\begin{table*}[t]
\centering
\setlength{\belowcaptionskip}{-8pt}
\resizebox{\textwidth}{!}{
\begin{tabular}{|c|c|c|c|c|}
\hline
\multirow{2}{*}{Problem} & \multicolumn{2}{c|}{Previous results}  & \multicolumn{2}{c|}{Our results with predictions}\\ \cline{2-5} 
 &  \makecell{without \\ predictions} & \makecell{with \\ predictions}   & Consistency & Robustness\\ \hline \hline
\makecell{Flow time}  & 2 \cite{andrew2009optimal} & \multirow{3}{*}{None} &  \multirow{2}{*}{($1 + (2\lambda)^{\frac{1}{\alpha}})^{\alpha}$} & \multirow{2}{*}{$\frac{1 + 2^{2\alpha + 1}\lambda^{\frac{1}{\alpha}}}{\lambda}$ }  \\ \cline{1-2} 
\makecell{Fractional weighted \\flow time}  & 2 \cite{bansal2013speed}  & &  &\\ \cline{1-2} \cline{4-5}
\makecell{Integral weighted \\flow time}  &  $O((\frac{\alpha}{\log \alpha})^2)$~\cite{bansal2010speed} & & ($1 + ((\frac{\alpha}{\log \alpha})^2\lambda)^{\frac{1}{\alpha}})^{\alpha}$  & $\frac{1+(\frac{\alpha}{\log \alpha})^2 2^{2\alpha}\lambda^{\frac{1}{\alpha}}}{\lambda}$ \\ \hline
\hline
\makecell{Deadlines} & $e^{\alpha}$~\cite{bansal2007speed} & \cite{BamasMRS20, antoniadis2021novel} & $1+\lambda$ & $O(\frac{4^{\alpha^{2}}}{\lambda^{\alpha-1}})$ \\ \hline 

\end{tabular}}
\caption{\footnotesize{The best-known competitive ratios for $4$ energy-efficient scheduling problems, previous work studying these problems in the algorithms with predictions framework, and our consistency and robustness results, for any $\lambda \in (0,1]$. Note that when $\lambda$ is sufficiently small, the consistency improves over the best-known competitive ratios, while also maintaining bounded robustness. A detailed comparison with the results of \cite{BamasMRS20, antoniadis2021novel} for  deadlines is provided in Section~\ref{sec:related-work}.}}
\label{tab:results}
\end{table*}
%\npedit{We present in Table~\ref{tab:results} the consistency and robustness trade-offs obtained when applying our algorithm to different GESP problems, in comparison with the competitive ratio of the best known online algorithm without predictions (see Section~\ref{sec:various_objectives} for details). When setting the confidence parameter $\lambda$ small enough, we obtain a better competitive ratio than the best online algorithm for these problems in the case the predictions are correct, while maintaining a bounded competitive ratio in the case the predictions are arbitrarily wrong.}

\paragraph{The competitive ratio as a function of the prediction error.} The second main result is that our algorithm achieves a competitive ratio that smoothly interpolates from the $1+\epsilon$ consistency to the constant robustness as a function of the prediction error (Theorem~\ref{thm:competitive_ratio}).
%Our second goal is to obtain a competitive ratio that smoothly interpolates between the consistency and robustness as a function of some measure of the prediction error. %A non-trivial challenge here is to design an appropriate measure for the prediction error.
To define the prediction error, we denote by $\J^+ = \J \cap \hat{\J}$  the jobs that are correctly predicted. %Note that there are two types of prediction errors: the jobs $\J\setminus \J^+$ that arrived but were not predicted to arrive and the jobs $\predJ \setminus \J^+$ that were predicted to arrive but did not arrive. 
We define the  prediction error $\eta = \frac{1}{\texttt{OPT}(\predJ)}\max\{\texttt{OPT}(\J\setminus \J^+), \texttt{OPT}(\predJ \setminus \J^+)\},$
which is the maximum between the optimal cost of scheduling the jobs $\J\setminus \J^+$ that arrived but were not predicted to arrive and the cost of the jobs $\predJ \setminus \J^+$ that were predicted to arrive but did not arrive.
%incurred by the two types of errors, normalized by the optimal cost over the predicted instance. %We note that $\eta = 0$ when the predictions are exactly correct).  
This prediction error is upper bounded by the prediction error in \cite{BamasMRS20} for the problem with uniform deadlines.
%\npedit{(we refer the reader to Section~\ref{sec:preliminaries} for additional discussion about this choice of error metric)}.
 %In particular, this competitive ratio is arbitrarily close to $1+ \epsilon$ as  $\eta$ approaches $0$.

\paragraph{Extension to jobs that are approximately predicted correctly.} We generalize our algorithm and the previous result to allow the correctly predicted jobs $\J^+$ to include jobs that are approximately predicted correctly, where the tolerable approximation is parameterized by a parameter chosen by the algorithm designer. %We define the generalized prediction error 
%$\eta^g = \frac{1}{\texttt{OPT}(\predJ)}\cdot\max\{\texttt{OPT}(\J \setminus \J^{\text{shift}})), \texttt{OPT}(\predJ \setminus \predJ^{\text{shift}})\}.$ 
The result for this extension requires an additional smoothness condition on the quality cost $F(S, \J)$ of a schedule. This condition is satisfied for the flow time problem, but not by the one with deadlines.\footnote{We note that \citet{BamasMRS20} give an alternate approach to transform an arbitrary  algorithm with predictions for the problem with uniform deadlines to an algorithm that allows small deviations in the release time of the jobs. This approach can be applied to our algorithm for the problem with uniform deadlines.} 
%(we note that for the energy minimization with uniform deadlines problem, \cite{BamasMRS20} shows how to transform any learning-augmented algorithm into an algorithm that tolerates small shifts in the prediction).}  

%, which achieves a competitive ratio that smoothly interpolates from the $(1+\epsilon)$ consistency to the constant robustness as a function of both $\eta^g$ and $\eta^{\text{shift}}$ (Theorem~\ref{thm:UBshiftm}).

% \npedit{A summary of our results in comparison to previous works for the specific objectives or energy plus deadlines and energy plus flow time is presented in Table~\ref{tab:comparison}.}

\paragraph{Experiments.} In Section~\ref{sec:experiments}, we show that when the prediction error is small, our algorithm empirically outperforms on both real and synthetic datasets the online algorithm that achieves the optimal competitive ratio for energy plus flow time minimization without predictions.

\subsection{Related work}
\label{sec:related-work}

\paragraph{Energy-efficient scheduling.}
Energy-efficient scheduling was initiated by~\citet{yao1995scheduling}, who studied the energy minimization with deadlines problem in both  offline and online settings. %where each job has a strict deadline and the objective is to determine the optimal speed of the machine while completing all jobs on time. This problem, called \textit{energy minimization with deadlines,} is a special case of our setting. 
These offline and online algorithms were later improved in  \cite{bansal2011average, bansal2007speed}. Over the last two decades, energy-efficient scheduling has been extended to several other objective functions. In particular, \citet{albers2007energy}  proposed the problem of energy plus flow time minimization, which has been studied extensively (see, e.g.,  \cite{andrew2009optimal,bansal2009weighted,bansal2013speed,bansal2010speed, lam2008speed, bansal2008scheduling, BansalP14}). %Among the variants, we would like to highlight  \citep{BansalP14}, which introduces objective functions of the form $\sum_{j=1}^{n}{w_j(c_j)}$, where $w_j(c_j)$ is a general function of the completion time $c_j$ of job $j$. This type of general objective function is captured by our setting as well.
%\ebcomment{Do we have results that apply for the setting from \citep{BansalP14}  where $w_j(c_j)$ is a general function of the completion time $c_j$ of job $j$?} \npcomment{In theory, it works in this setting as well. I think the main issue was that we are not aware of a blackbox online algorithm for this problem. Maybe it is still worth mentioning it? To justify the relevance of the general type of objective function we consider? We could also add one bullet point in Section 3.3 to precise that.}

\paragraph{Learning-augmented algorithms.} Algorithms with predictions is a recent and extremely active area, especially in online algorithms, where it was initiated in \cite{DBLP:journals/jacm/LykourisV21, KPZ18}. Many different scheduling problems have been studied with predictions (see, e.g., ~\cite{LLMV20,MM20,im2021non, lindermayr2022permutation,balkanski2022scheduling, balkanski2022strategyproof,cho2022scheduling,lindermayr2023speed}).

\paragraph{Learning-augmented energy-efficient scheduling.}
Energy-efficient scheduling with predictions has been studied by \citet{antoniadis2021novel} and \citet{BamasMRS20}, who focus on the problem with deadlines, which is a special case of our setting. The prediction model in~\citet{BamasMRS20} is the closest to ours. For the problem with deadlines, the algorithm in
\cite{BamasMRS20} achieves a better consistency-robustness tradeoff than our algorithm, but their algorithm and prediction model do not extend to more general energy-efficient scheduling problems such as the flow time problem. In addition, 
the competitive ratio as a function of the prediction error is only obtained in \cite{BamasMRS20} in the case of uniform deadlines where the difference between the deadline and release date of a job is equal for all jobs (the authors mention that defining algorithms for general deadlines becomes complex and notationally heavy when aiming for bounds as a function of the prediction error).
Thanks to our algorithmic framework and definition of prediction error, our bound generalizes to the non-uniform deadlines without complicating our algorithm.
%present an algorithm that achieves the current-best tradeoff between consistency and robustness. Their algorithm exhibits a smooth competitive ratio for the case of uniform deadlines, as well as robustness and consistency for general deadlines. They acknowledge that defining smooth algorithms for general deadlines becomes complex and notationally heavy when considering the prediction model and measure of error.  In contrast, our approach maintains simplicity in both the measure of error and the algorithm for the general deadlines case, while still obtaining smooth guarantees.
 \citet{antoniadis2021novel} propose a significantly different prediction model that requires an equal number of jobs in both the prediction $\predJ$ and true set of jobs $\J$. %Furthermore, their model's error is proportional to the maximum shift in parameters. 
 Consequently, their results are incomparable to ours and those by~\citet{BamasMRS20}.
% We provide a more detailed comparison with these works in Appendix~\ref{appendix_discussion}.

% which they qualify as 'orthogonal' to the one in \cite{BamasMRS20}.  Since our prediction model generalizes the one in \cite{BamasMRS20}, it is also incomparable to the one in \cite{antoniadis2021novel} (see the formal definition of the model in Section~\ref{sec:preliminaries} for further discussion). 

Finally, we note that \citet{lee2021online} also study energy scheduling with predictions, but with the different challenge of deciding if demand should be covered by local generators or the external grid.

\vspace{-.05cm}
\section{Preliminaries}
\label{sec:preliminaries}
\vspace{-.05cm}

In the \textit{General Energy-Efficient Scheduling (GES)}  problem, an instance is described by a collection $\J$ of $n$ jobs and a real-valued cost function $F(S, \J)$ that takes as input the instance $\J$ and a schedule $S$ for $\J$, and returns some quality evaluation of the schedule. Each job  $(j, r_j, p_j) \in \J$ consists of a release time $r_j$, a processing time $p_j$, and an identifier $j$ (and potentially other parameters such as weights $v_j$ and deadlines $d_j$). We often abuse notation and write $j\in \J$ instead of $(j, r_j, p_j) \in \J$.   For any time interval $I$, we let $\J_{I} = \{j\in \J: r_j \in I\}$ be the subset of jobs of $\J$ with release time in  $I$. For intervals $I = [0,t]$ or $I = [t,\infty]$, we write $\J_{\leq t}$ and $\J_{\geq t}$.

 A \textit{feasible schedule} for a set of jobs $\J$ is specified by $S = \{s_j(t)\}_{t\geq 0, j\in \J_{\leq t}}$, where $s(t) := \sum_{j\in \J_{\leq t}} s_j(t)$ is the speed at which the  machine runs at time $t$. Thus, $s_j(t)/s(t)$ is the fraction of the processing power of the machine allocated to job $j$ at time $t$.\footnote{For ease of notation, we allow the machine to split its processing power at every time step $t$ over multiple jobs. In practice, this is equivalent to partitioning time into arbitrarily small time periods and splitting each time period into smaller subperiods such that the machine is processing one job during each subperiod.} 
During a time interval $I$, there are $\int_{I} s_j(t) \text{dt}$ units of work for job $j$ that are completed and we let $S_{I}$ be the sub-schedule $\{s_j(t)\}_{t\in I, j\in \J_{I}}$. The cost function we consider is a combination of energy consumption and quality cost for the output schedule. The energy consumption incurred by a schedule is $E(S) = \int_{t\geq 0} s(t)^{\alpha}\text{dt}$, where $\alpha>1$ is a problem-dependent constant,
chosen so that the power at time $t$ is $s(t)^{\alpha}$.
To define the quality of a schedule, we introduce the \textit{work profile}  $W^S_j := \{w^S_j(t)\}_{t\geq r_j}$ of schedule $S$ for job $j$, where  $w^S_j(t) :=p_j - \int_{r_j}^t s_j(u) \text{du}$ is the amount of work for $j$ remaining at time $t$. %Note that the completion time $c_j^S$ of job $j$ while following $S$ is the first time $t\geq r_j$ such that $w^S_j(t) = 0$. 

We consider general objective functions of the form 
$\text{cost}(S, \J) =  E(S) + F(S, \J)$ and the goal is to compute a feasible schedule of minimum cost. $F(S, \J) = f((W^S_1, j_1), \ldots, (W^S_n, j_n))$ is an arbitrary 
 \textit{quality cost} function that is a function of the work profiles and the jobs' parameters. In the energy minimization with deadlines problem, $F(S, \J) = \infty$ if there is a job $j$ with completion time $c_j^S$  such that $c_j^S > d_j$, and $F(S, \J) = 0$ otherwise. In the energy plus flow time minimization problem, we have $F(S, \J) = \sum_{j \in \J} c_j^S - r_j$ (see Section~\ref{sec:various_objectives} for additional functions $F$).
%In particular, we consider in Section~\ref{sec:various_objectives} the special cases where $F$ is the total integral flow time (sum of all completion times minus release times), the integral weighted flow time and the fractional weighted flow time. 
% In these cases, the objective function is the classically studied total energy plus flow time (see for instance \cite{bansal2009weighted,bansal2010speed, bansal2013speed}).
A function $F(S, \J)$ is subadditive  if for all sets of jobs $\J_1$ and $\J_2$, we have $F(S, \J_1 \cup \J_2) \leq F(S, \J_1) + F(S, \J_2)$.  $F$ is monotone if for all sets of jobs $\J$ and schedules $S$ and $S'$ such that $w^S_j(t)\leq w^{S'}_j(t)$ for all $j \in \J$ and $t \geq r_j$, we have that $F(S,\J)\leq F(S',\J)$. We assume throughout the paper that $F$ is monotone subadditive, which holds for the deadlines and flow time problems. We let $S^*(\J)$ and $\OPT(\J) := \text{cost}(S^*(\J), \J)$  be an optimal offline schedule and the optimal objective value. 

% Given an algorithm $\A$ that outputs a schedule $S$ for a set of jobs $\J$, the performance metric we use is the competitive ratio:
\vspace{-.05cm}
\paragraph{The general energy-efficient scheduling with predictions problem.} We augment the GES problem with predictions
regarding future job arrivals and call this problem the General Energy-Efficient Scheduling with Predictions problem
(GESP). In this problem,  the algorithm is  given at time $t= 0$ a prediction $\predJ = \{(j, \hat{r}_j, \hat{p}_j)\}$ regarding the jobs  $\J = \{(j, r_j, p_j)\}$ that arrive online. An important feature of our prediction model is that the number of predicted jobs $|\predJ |$ can differ from the number of true jobs $|\J|$.

Next, we define a measure for the prediction error which generalizes the prediction error in~\cite{BamasMRS20} for the problem with uniform deadlines to any GES problem.
With $\J^+ = \J \cap \hat{\J}$ being the correctly predicted jobs, we define the prediction error as 
%\npedit{that can be used by the algorithm to schedule the real sequence of jobs $\J$. However, the predictions may not be completely correct. In particular, there may be additional jobs in the realization as compared to the prediction and some of the predicted jobs may not arrive. Ideally, the competitive ratio of a learning-augmented algorithm should achieve a graceful dependence in the number of mistakes made by the prediction: the more jobs correctly predicted, the better the competitive ratio should be. To account for this, a first intuitive notion of error could be to compare the 'probability density functions' of the predicted and realized sequences of jobs, by computing the difference of jobs arrived until each time step $t$. However, we show in Appendix~\ref{app:consistency_robustness} some impossibility results for such simple notion of error. The main take-away is that in order to obtain some smooth guarantees on the competitive ratio, one must analyse more precisely the total impact of the mistakes on the objective function (which heavily depends on how they are concentrated), and not the mere number of mistakes.
%Hence, we define the prediction error as the optimal cost for all jobs that were not correctly predicted (i.e., extras or missing jobs in the realization as compared to the prediction). } More formally, for any realization $\J$ of the jobs and any sequence $\predJ$ of predicted jobs, we define the error as
\vspace{-.05cm}
\[\eta(\J, \predJ) = \frac{1}{\texttt{OPT}(\predJ)}\max\{\texttt{OPT}(\J\setminus \J^+), \texttt{OPT}(\predJ \setminus \J^+)\},\]
\vspace{-.05cm}
% \hwedit{The proof that our error measure satisfying the \emph{Lipschitzness} property is presented in Appendix~\ref{app:Lip}.}
% \hwedit{which, we believe, better captures the essence of the problem, since it takes into accounts the cost of incorrectly predicted jobs. 
% Moreover, we would like to mention that the error is highly related to the learnability of predictions (see Appendix~\ref{app:learnability} for further discussion of the learnability).
% }
where $\texttt{OPT}(\J\setminus \J^+)$ is the optimal cost of scheduling the true jobs $(j, r_j, p_j)$ such that either the prediction for $j$ was wrong or there was no prediction for $j$ and that $\texttt{OPT}(\predJ \setminus \J^+)$ is the optimal cost of scheduling the predicted jobs $(j, \hat{r}_j, \hat{p}_j)$ such that either the prediction for $j$ was wrong or  $j$ never arrived. The prediction error $\eta(\J, \predJ)$ is then the maximum of these costs, normalized by the optimal cost $\texttt{OPT}(\predJ)$ of scheduling the predicted jobs. 
We assume that $\predJ \neq \emptyset$ to ensure that $\eta(\J, \predJ)$ is well-defined. 
% \hwedit{
% Here, we give a high-level idea behind the choice of this prediction error.
% Other types of prediction errors may also be considered. However, some common error measures may not capture the essence of the problem. For instance, the $\ell_1$ norm on workloads is unable to distinguish between a small error occurring in each time step and a large error happening in a single time step. Similarly, the $\ell_\infty$ norm cannot differentiate between a single large error and multiple large errors in prediction. Furthermore, the total cost of energy plus flow time minimization is highly sensitive to the density of arrivals within a given interval. Hence, we decided to use this particular error measure and we believe that it captures the essence of the problem.}
This prediction error is upper bounded by the prediction error $||w^{\text{true}} - w^{\text{pred}}||_{\alpha}^{\alpha}$ considered in \cite{BamasMRS20} for the problem with uniform deadlines, which we prove in Appendix~\ref{appe:discussion_model}. Here $w^{\text{true}}$ and $w^{\text{pred}}$ are the true and predicted workload at each time step $t$, i.e., the sum of the processing times of the jobs that arrive at $t$.

%\npedit{We note that for some specific cases of the problem, the prediction error may reduce to a simpler expression. For instance, for energy minimization with uniform deadlines, letting $w^{\text{real}}$ and $w^{\text{pred}}$ denote the real and predicted workloads, the error $\eta(\J, \predJ)$ can be upper bounded, up to a constant factor, by the simpler error $||w^{\text{real}} - w^{\text{pred}}||_{\alpha}^{\alpha}$ considered  .} 

We note that in the above error model, a job $j$ is in the set of correctly predicted jobs $\J^+$ only if all the parameters of $j$ have been predicted exactly correctly. To overcome this limitation, we introduce in Section~\ref{sec:extension} a more general error model where some small deviations between the true and predicted parameters of a job $j$ are allowed for the correctly predicted jobs $\J^+$. In Appendix~\ref{appe:discussion_model}, we provide further discussion of this prediction model in comparison with \cite{BamasMRS20, antoniadis2021novel}.

% \hwedit{Overall, this paper proposes two prediction error measures that address key challenges as follows. The first measure focuses on mispredicted jobs, particularly addressing the presence of additional jobs. The second measure tackles shifts in job parameters. However, we acknowledge that there are other error metrics that could be explored for future research. For instance, the second measure requires determining the shift tolerance before observing the input, it would be worthwhile to investigate alternative error measures that can dynamically adapt the tolerance throughout the input sequence.}

\paragraph{Performance metrics.} The standard evaluation metrics for an online algorithm with predictions are its consistency, robustness, and competitive ratio as a function of the prediction error~\cite{mahdian2007allocating, DBLP:journals/jacm/LykourisV21}. The competitive ratio of an algorithm \textsc{ALG} as a function of a prediction error $\eta$ is  $$c(\eta) = \max_{\J, \predJ:\; \eta(\J, \predJ)\leq \eta} \frac{\cst_{\textsc{ALG}}(\J,\predJ)}{\OPT(\J)}.$$ $\textsc{ALG}$ is $\rho$-robust if for all $\eta \geq 0$, $c(\eta)\leq \rho$ (competitive ratio when the error is arbitrarily large) and  $\mu$-consistent if $c(0) \leq \mu$ (competitive ratio when the prediction is exactly correct). The competitive ratio of \textsc{ALG} is called smooth if it smoothly degrades from $\mu$ to $\rho$ as the prediction error $\eta$ grows.

\section{The Algorithm}
\label{sec:framework}

% In this section, we describe an algorithm with predictions for the GESP problem and analyze its consistency, robustness, and competitive ratio as a function of the prediction error.
In this section, we develop a simple and general algorithmic framework for GESP and analyze the resulting consistency, robustness, and competitive ratio as a function of the prediction error. We first note that the algorithm with predictions from \cite{BamasMRS20} for the problem with deadlines does not easily generalize to some of the other problems that we consider, including the flow time problem (see Appendix~\ref{app:techniques_discussion} for additional discussion). A major difference is that our algorithm consists of two distinct phases.

\paragraph{Predictions cannot be completely trusted.} We also note that a first natural approach is to assume that the predictions are exactly correct and aim for
% \hwcomment{an algorithm that matches the lower bound $\gamma_{\text{on}}_{\text{off}}$?} 
a $1$-consistent algorithm.  For the problem with deadlines, \citet{BamasMRS20} showed that there is no $1$-consistent algorithm with bounded robustness. In Appendix~\ref{app:no-1-cons}, we show that this approach would also fail for the flow time problem because the algorithm might start by processing jobs too fast and consume too much energy when trusting the predictions. More generally, in
Appendix~\ref{app:lowerbound-tradeoff}, we show that there is a necessary trade-off between
consistency and robustness for the flow time problem by proving that any $1+ \lambda$-consistent algorithm must be $O(\sqrt{1 + 1/\lambda})$-robust.

\subsection{Description of the algorithm} 

The algorithm, called \tpe,  takes as input an arbitrary quality of service cost function $F$, predictions $\predJ$, a confidence level $\lambda \in (0,1]$ in the predictions,  an offline algorithm \textsc{OfflineAlg} for  $F$,  and an online algorithm \textsc{OnlineAlg}  for  $F$ (without predictions).  We denote by $\texttt{OFF}(\J) := \text{cost}(\textsc{OfflineAlg}(\J), \J)$  the  objective value  achieved by $\textsc{OfflineAlg}$ over $\J$.
% For example, for the energy plus flow time minimization problem where $F = \sum_{j} c_j - r_j$, there exists an optimal offline algorithm \textsc{OfflineAlg} \cite{albers2007energy} and a $2$-competitive online algorithm \textsc{OnlineAlg}. 

The algorithm proceeds in two phases. In the first phase (Lines 1-5), \tpe \ ignores the predictions and runs the auxiliary online algorithm \textsc{OnlineAlg} over the true jobs $\J_{\leq t}$ that have been released by time $t$. More precisely, during the first phase of the algorithm, $s_j(t)$ is the speed according to the online algorithm \textsc{OnlineAlg} for all jobs. The first phase ends at the time $\tl$ when the cost of the offline schedule computed by running \textsc{OfflineAlg} on jobs $\J_{\leq t}$ reaches the threshold value  $\lambda  \cdot \texttt{OFF}(\predJ)$. As we will detail in the analysis section, this first phase guarantees a bounded robustness since we ensure that the offline cost for the true jobs reaches some value before starting to trust the predictions (hence, \tpe \ does not initially `burn' too much energy  compared to the  optimal offline cost, unlike the example described in Appendix~\ref{app:no-1-cons}).

\begin{algorithm}[H]
\caption{Two-Phase Energy Efficient Scheduling (\textsc{TPE})}
\label{alg-no-ass}
% \hspace*{\algorithmicindent} 
{\bf Input:} predicted and true  sets of jobs $\predJ$ and $\J$, quality of cost function $F$, offline and online algorithms (without predictions) \textsc{OfflineAlg} and  \textsc{OnlineAlg} for problem $F$, confidence level $\lambda \in (0, 1]$.
\begin{algorithmic}[1]
\STATE \textbf{for} {$t \geq 0$} \textbf{do} 
\STATE \quad \textbf{if}  $\texttt{OFF}(\J_{\leq t}) >\lambda  \cdot \texttt{OFF}(\predJ)$ \textbf{then}
\STATE \quad \quad $\tl \leftarrow t$
\STATE \quad \quad \textbf{break}
\STATE \quad  $\{s_j(t)\}_{j\in \J_{\leq t}} \leftarrow \textsc{OnlineAlg}(\J_{\leq t})(t)$
\STATE $\{\hat{s}_j(t)\}_{t\geq \tl, j\in \predJ_{\geq\tl}}\leftarrow \textsc{OfflineAlg}(\predJ_{\geq\tl})$
\STATE \textbf{for} {$t \geq \tl$} \textbf{do}
% \State $\ebedit{\hat{s}, \{\hat{\alpha}_j\}_{j} \leftarrow \hat{s}(t), \{\hat{\alpha}_j(t)\}_{j}}$ 
\STATE \quad $\{s_j(t)\}_{j\in \J_{\leq t}\setminus \predJ_{\geq \tl}} \hspace{-.12cm} \leftarrow \hspace{-.03cm} \textsc{OnlineAlg}(\J_{\leq t}\setminus \predJ_{\geq \tl})(t)$
\STATE \quad $\{s_j(t)\}_{j\in \J_{[\tl, t]} \cap \predJ_{\geq\tl} }  \leftarrow \hspace{-.03cm} \{\hat{s}_j(t)\}_{j\in \J_{[\tl, t]} \cap \predJ_{\geq\tl} }$
% \EndFor
% \EndWhile
\STATE \textbf{return} $\{s_j(t)\}_{t\geq 0, j\in \J}$
\end{algorithmic}
\end{algorithm}

%\npcomment{Remove? Repeats the 'our techniques' section} The algorithm we propose is a two-phase algorithm, which first follows the optimal online algorithm, then starts using the predictions. At a high level, the algorithm must incorporate in the second phase all jobs that were not predicted, and all jobs that were not correctly scheduled in the first phase and still have work remaining at the switch point. Before giving the formal description of the algorithm, we give in the next paragraph some intuition about why any algorithm with bounded robustness cannot start by following blindly the predictions. 

%As mentioned above, our algorithm OnlineAlgWithPred proceeds in two phases: in the first one (lines 1-4), the algorithm follows the online algorithm \textsc{OnlineAlg}, until the time $\tl$ where the optimal offline objective for all jobs arrived so far reaches some threshold value. 

In the second phase (Lines 6-9), \tpe \ starts leveraging the predictions. More precisely, \tpe \ needs to set the speeds for three different types of jobs: (1) the remaining jobs that were correctly predicted (i.e., $\J_{\geq \tl}\cap \predJ_{\geq \tl}$) (2) the remaining jobs that were not predicted (i.e., $ \J_{\geq \tl}\setminus \predJ_{\geq \tl}$) (3) the jobs that were not correctly scheduled in the first phase and still have work remaining at the switch point $\tl$ (which are a subset of $\J_{< \tl}$). To schedule these jobs, \tpe \ combines two different schedules. The first one is the offline schedule $\hat{S}:= \textsc{OfflineAlg}(\predJ_{\geq\tl})$ for the jobs $\predJ_{\geq \tl}$ that are predicted to arrive in the second phase. Each future job in the true set that was correctly predicted (i.e., $j\in \J_{[\tl, t]} \cap \predJ_{\geq\tl} $ on Line 9) will then be scheduled by following $\hat{S}$. The second schedule is an online schedule for the set of jobs $ \J\setminus\predJ_{\geq \tl} =\J_{<\tl} \cup  \J_{\geq \tl} \setminus \predJ_{\geq\tl}$, which includes all jobs that have not been completed during the first phase ($\subseteq \J_{<\tl}$) and the incorrectly predicted jobs that are released during the second phase ($\J_{[\tl, t]} \setminus \predJ_{\geq\tl}$). This online schedule is constructed by running $\textsc{OnlineAlg}$ on the set $\J\setminus\predJ_{\geq \tl}$ (Line 8). Note that the total speed of the machine at each time step is the sum of the speeds of these two online and offline schedules.

\subsection{Analysis of the algorithm}

 We analyze the competitive ratio of \tpe \ as a function of the prediction error $\eta$, from which the consistency and robustness bounds follow. Missing proofs are provided in Appendix~\ref{app:algorithm}. We separately bound the cost of the algorithm due to jobs in $\J_{< \tl}$, $\J_{\geq \tl}\setminus \predJ_{\geq \tl}$ and $\J_{\geq \tl} \cap \predJ_{\geq \tl}$. We do this by analyzing  the costs of  schedules $S^{on}:= \textsc{OnlineAlg}(\J\setminus\predJ_{\geq \tl} )$ and $\hat{S}:=\textsc{OfflineAlg}(\predJ_{\geq\tl})$.  In the next lemma, we first analyze the cost of combining, i.e., summing,  two arbitrary schedules.

\begin{restatable}{lem}{lemunionoftwoschedule}
\label{lem:union-of-two-schedule}
Let $\J_1$ be a set of jobs and $S_1$ be a feasible schedule for $\J_1$,
let $\J_2$ be a set of jobs and $S_2$ be a feasible schedule for $\J_2$. Consider the schedule $S:= S_1+S_2$ for $ \J_1 \cup \J_2$ which, at each time $t$, runs the machine at total speed $s(t) = {s_1}(t)+ {s_2}(t)$ and processes each job $j\in \J_1$ at speed $s_{1,j}(t)$ and  each job $j\in \J_2$ at speed $s_{2,j}(t)$. 
Then, $\cst(S, \J_1\cup\J_2) \leq \left(\cst(S_1, \J_1)^{\frac{1}{\alpha}} + \cst(S_2, \J_2)^{\frac{1}{\alpha}}\right)^{\alpha}.$
\end{restatable}

We next upper bound the cost of the schedule output by \tpe \ as a function of the prediction error $\eta$, which we decompose into  $\eta_1 = \frac{\texttt{OPT}(\J\setminus \predJ)}{\predT}$ and $ \eta_2 = \frac{\texttt{OPT}(\predJ \setminus \J)}{\predT}$. The proof uses the previous lemma repeatedly, first to analyze the cost of the schedule $S^{on}:= \textsc{OnlineAlg}(\J\setminus\predJ_{\geq \tl} )$ for the set of jobs $\J\setminus\predJ_{\geq \tl}  = (\J_{< \tl})\cup (\J_{\geq \tl}\setminus \predJ_{\geq \tl}$), then to analyze the cost of the final schedule, which combines $S^{on}$ and $\hat{S}:=\textsc{OfflineAlg}(\predJ_{\geq\tl})$.

\begin{restatable}{lem}{lemupperboundschedule}
\label{lem:upper_bound_schedule}
Assume that \textsc{OfflineAlg} is $\gamma_{\text{off}}$-competitive and  that \textsc{OnlineAlg} is $\gamma_{\text{on}}$-competitive. Then, for all $\lambda \in (0,1]$, the schedule $S$ output by \tpe \ run with confidence parameter $\lambda$ satisfies $\cst(S,\J)\leq \predT\left(\gamma_{\text{off}}^{\frac{1}{\alpha}} +
\gamma_{\text{on}}^{\frac{1}{\alpha}}((\lambda\gamma_{\text{off}})^{\frac{1}{\alpha}} + \eta_1^{\frac{1}{\alpha}})\right)^{\alpha}.$
\end{restatable}
\begin{proof}
We start by upper bounding $\cst(S^{on}, \J\setminus\predJ_{\geq \tl})$. First, by the algorithm, we have that $\texttt{OFF}(\J_{< \tl})\leq \lambda \cdot \texttt{OFF}(\predJ)$. Since \textsc{OfflineAlg} is $\gamma_{\text{off}}$-competitive, we get
\begin{equation}
    \texttt{OPT}(\J_{< \tl})\leq \texttt{OFF}(\J_{< \tl})\nonumber \leq \lambda \cdot \texttt{OFF}(\predJ) \leq \lambda \gamma_{\text{off}} \cdot \predT.
\label{eq:OPT_OFF}
\end{equation}

We also have that $\texttt{OPT}(\J_{\geq \tl}\setminus \predJ_{\geq \tl})\leq \texttt{OPT}(\J\setminus \predJ) \leq \eta_1\predT$ where the first inequality is since $\J_{\geq \tl}\setminus \predJ_{\geq \tl}\subseteq \J\setminus \predJ$ and the second is by definition of $\eta_1$. Recall that $S^*(.)$ denotes the optimal offline schedule for the problem and consider the schedule $S' = S^{\star}(\J_{< \tl}) + S^{\star}(\J_{\geq \tl}\setminus \predJ_{\geq \tl})$ for $\J\setminus\predJ_{\geq \tl} = \J_{< \tl}\cup (\J_{\geq \tl}\setminus \predJ_{\geq \tl})$. We obtain that
\begin{align*}
   \texttt{OPT}( \J\setminus\predJ_{\geq \tl})  & \leq  
\cst(S', \J\setminus\predJ_{\geq \tl})  \leq  \left(\texttt{OPT}(\J_{< \tl})^{\frac{1}{\alpha}} + \texttt{OPT}(\J_{\geq \tl}\setminus \predJ_{\geq \tl})^{\frac{1}{\alpha}}\right)^{\alpha}  \\
 % \leq \ & \left((\lambda \gamma_{\text{off}}\predT)^{\frac{1}{\alpha}} + \texttt{OPT}(\J_{\geq \tl}\setminus \predJ_{\geq \tl})^{\frac{1}{\alpha}}\right)^{\alpha} \\
  & \leq  \left((\lambda \gamma_{\text{off}}\predT)^{\frac{1}{\alpha}} + (\eta_1 \predT)^{\frac{1}{\alpha}})\right)^{\alpha}\\
 &  = \predT\left((\lambda \gamma_{\text{off}})^{\frac{1}{\alpha}} + \eta_1^{\frac{1}{\alpha}}\right)^{\alpha},
\end{align*}
where the second inequality is by Lemma~\ref{lem:union-of-two-schedule}. Since we assumed that \textsc{OnlineAlg} is $\gamma_{\text{on}}$-competitive, 
\begin{align*}
   \cst(S^{on}, \J\setminus\predJ_{\geq \tl}) 
   \leq  \gamma_{\text{on}} \cdot \texttt{OPT}( \J\setminus\predJ_{\geq \tl}) 
   \leq  \gamma_{\text{on}}\cdot \predT\left((\lambda \gamma_{\text{off}})^{\frac{1}{\alpha}} + \eta_1^{\frac{1}{\alpha}}\right)^{\alpha}. 
\end{align*}

We now bound the cost of schedule $S$. First, note that $
    \cst(\hat{S}, \predJ_{\geq \tl}) = \textsc{OfflineAlg}(\predJ_{\geq\tl}) 
    \leq 
\gamma_{\text{off}}\cdot\texttt{OPT}(\predJ_{\geq \tl}) \leq  
    \gamma_{\text{off}}\cdot\predT,$
where the first inequality is since \textsc{OfflineAlg} is $\gamma_{\text{off}}$-competitive
% $\gamma_{\text{on}}_{\text{off}}$-competitive 
and the last one since $\predJ_{\geq \tl}\subseteq \predJ$.
Therefore, by applying again Lemma~\ref{lem:union-of-two-schedule}, we get:
\begin{align*}
   \cst(S, \J) &\leq \left(\cst(\hat{S}, \predJ_{\geq \tl})^{\frac{1}{\alpha}} + \cst(S^{on}, \J\setminus\predJ_{\geq \tl})^{\frac{1}{\alpha}}\right)^{\alpha}\\
    &\leq \Big((
    % \gamma_{\text{on}}_{\text{off}}
    \gamma_{\text{off}}\cdot\predT)^{\frac{1}{\alpha}}  + \left(\gamma_{\text{on}}\cdot \predT\left((\lambda \gamma_{\text{off}})^{\frac{1}{\alpha}} + \eta_1^{\frac{1}{\alpha}}\right)^{\alpha}\right)^{\frac{1}{\alpha}}\Big)^{\alpha}\\
    &= \predT\left( \gamma_{\text{off}}^{\frac{1}{\alpha}} +
    \gamma_{\text{on}}^{\frac{1}{\alpha}}((\lambda \gamma_{\text{off}})^{\frac{1}{\alpha}} + \eta_1^{\frac{1}{\alpha}})\right)^{\alpha}.\qedhere
\end{align*}

\end{proof}

We next state a simple corollary of Lemma~\ref{lem:union-of-two-schedule}. 

\begin{restatable}{cor}{corJinterHat}
\label{cor:JinterHatJ-LB}
$\texttt{OPT}(\mathcal{J}\cap \predJ)\geq \left(1-\eta_2^{\frac{1}{\alpha}}\right)^{\alpha}\predT$, and, assuming that \textsc{OfflineAlg} is $\gamma_{\text{off}}$-competitive, we have: if $\texttt{OFF}(\J)\leq \lambda \texttt{OFF}(\predJ)$, then $\eta_2 \geq \left(1-(\lambda \gamma_{\text{off}})^{\frac{1}{\alpha}}\right)^{\alpha}$.
\end{restatable}

% \begin{restatable}{cor}{corlbeta}
% \label{cor:lb_eta2}
% If $\texttt{OPT}(\J)\leq \lambda\predT$, then $\eta_2 \geq \left(1-\lambda^{\frac{1}{\alpha}}\right)^{\alpha}$.
% \end{restatable}

We are ready to state the main result of this section, which is our upper bound on the competitive ratio of \tpe. 

% To show this upper bound, we consider two different cases. \npedit{In the case where $\texttt{OPT}(\mathcal{J})\leq \lambda \predT$, the result follows immediately since the algorithm never transitions to the second phase and keeps following \textsc{OnlineAlg} all along. In the case where $\texttt{OPT}(\mathcal{J})>\lambda \predT$, we obtain the upper bound in the numerator by using Lemma~\ref{lem:upper_bound_schedule}. To get the lower bound in the denominator, we study more precisely the optimal offline cost as a function of the prediction error by using Corollary~\ref{cor:JinterHatJ-LB}.}

\begin{restatable}{theorem}{maintheorem}
\label{thm:competitive_ratio}
For any $\lambda\in (0,1]$, \tpe \  with  a $\gamma_{\text{on}}$-competitive algorithm \textsc{OnlineAlg} and a $\gamma_{\text{off}}$-competitive offline
% a $\gamma_{\text{on}}_{\text{off}}$-competitive 
algorithm \textsc{OfflineAlg} achieves a competitive ratio of  
$$
\begin{cases}
    \gamma_{\text{on}} &  \text{if } \texttt{OFF}(\mathcal{J})\leq \lambda \texttt{OFF}(\predJ) \vspace{0.2cm} \\
     \frac{\left(\gamma_{\text{off}}^{\frac{1}{\alpha}}
     % \gamma_{\text{on}}_{\text{off}}^{\frac{1}{\alpha}} 
     + \gamma_{\text{on}}^{\frac{1}{\alpha}}((\lambda \gamma_{\text{off}})^{\frac{1}{\alpha}} + \eta_1^{\frac{1}{\alpha}})\right)^{\alpha}}{\max\left\{\frac{\lambda}{\gamma_{\text{off}}}, \eta_1+ \left(1-\eta_2^{\frac{1}{\alpha}}\right)^{\alpha}\right\}} & \text{otherwise}.
     \end{cases}
$$
\end{restatable}

The consistency and robustness immediately follow (for simplicity, we present the results in the case where \textsc{OfflineAlg} is optimal). Additional discussion on this competitive ratio is provided in Appendix~\ref{app:cr}.

\begin{restatable}{cor}{corconsistency}
\label{cor:consistency}
For any $\lambda\in (0,1)$, \tpe \  with  a $\gamma_{\text{on}}$-competitive algorithm \textsc{OnlineAlg} and an optimal offline
% a $\gamma_{\text{on}}_{\text{off}}$-competitive 
algorithm \textsc{OfflineAlg}\ is $ 1+ 
\gamma_{\text{on}} 2^{\alpha}\lambda^{\frac{1}{\alpha}}$ competitive if $\eta_1 = \eta_2 = 0$ (consistency) and 
 $\max\{\gamma_{\text{on}},\frac{1+\gamma_{\text{on}}2^{2\alpha}\lambda^{\frac{1}{\alpha}}}{\lambda}\}$-competitive for all  $\eta_1, \eta_2$ (robustness). In particular, for any constant $\epsilon > 0$, with  $\lambda = (\frac{\epsilon}{\gamma_{\text{on}} 2^{\alpha}})^{\alpha}$, \tpe \ is $1+\epsilon$-consistent and $O(1)$-robust.
\end{restatable}

\subsection{Results for well-studied GES problems}
\label{sec:various_objectives}
We apply the general framework detailed in Section~\ref{sec:framework} to derive smooth, consistent and robust algorithms for a few classically studied objective functions.
% \hwcomment{optimal offline only exists for the case of unit-work jobs or do we assume we have optimal offline here?}

\paragraph{Energy plus flow time minimization.} Recall that $c_S^j$ denote the completion time of job $j$. The quality cost function is defined as: $F(S,\J) = \sum_{j\in \J} (c_S^j - r_j)$, with total objective $\cst(S, \J) = F(S,\J) + E(S)$. By a direct application of Corollary~\ref{cor:consistency}, we get that for all $\lambda\in (0,1]$, Algorithm~\ref{alg-no-ass} run with the 2-competitive online algorithm from \cite{andrew2009optimal} and confidence parameter $\lambda$ is $(1
+2^{\frac{1}{\alpha}}\lambda^{\frac{1}{\alpha}})^{\alpha}$-consistent and $\frac{1
% \gamma_{\text{off}}
+2\cdot 2^{2\alpha}\lambda^{\frac{1}{\alpha}}}{\lambda}$-robust.

\paragraph{Energy plus fractional weighted flow time minimization.} In this setting, each job has a weight $v_j$.  The quality cost is $F(S,\J) = \sum_{j\in \J} v_j\int_{t\geq r_j}w_j^S(t)\text{dt}$. We can use as \textsc{OnlineAlg} the 2-competitive algorithm from \cite{bansal2013speed}.

\paragraph{Energy plus integral weighted flow time minimization.} In this setting, each job has a weight $v_j$. The quality cost function is defined as: $F(S,\J) = \sum_{j\in \J} v_j (c_S^j - r_j)$. We can use as \textsc{OnlineAlg} the $O((\alpha/\log \alpha)^2)$-competitive algorithm from \cite{bansal2010speed}.

\paragraph{Energy minimization with deadlines.} In this setting, there is also a deadline $d_j$ for the completion of each job. By writing the quality cost as $F(S, \J) = \sum_{j\in \J} \mathbf{\delta}_{c_S^j> d_j}$, where $\delta_{c_S^j> d_j} = +\infty$ if $c_S^j> d_j$ and $0$ otherwise,  the total objective can be written as $\cst(S,\J) = E(S) + F(S,\J)$. We can use as \textsc{OnlineAlg} the \textsc{Average Rate} heuristic \cite{yao1995scheduling} (which is $2^{\alpha}$-competitive for uniform deadlines \cite{BamasMRS20}). In particular, for uniform deadlines, and for all $\epsilon\in (0,1]$, by setting $\lambda = (\frac{\epsilon}{C2^{\alpha}})^{\alpha}$, we obtain a consistency of $(1+\epsilon)$ for a robustness factor of $O(4^{\alpha^2} / \epsilon^{\alpha-1})$.

\subsection{Discussion on the competitive ratio} 
\label{app:cr}

We assume in this section that \textsc{OfflineAlg} is optimal. Note that for small $\eta_1$ and $\eta_2$, the competitive ratio is upper bounded as $\frac{\left( 1
% \gamma_{\text{off}}^{\frac{1}{\alpha}} 
+ \gamma_{\text{on}}^{\frac{1}{\alpha}}(\lambda^{\frac{1}{\alpha}} + \eta_1^{\frac{1}{\alpha}})\right)^{\alpha}}{\eta_1+ \left(1-\eta_2^{\frac{1}{\alpha}}\right)^{\alpha}}$, which smoothly goes to $(1
% \gamma_{\text{on}}_{\text{off}}
+\gamma_{\text{on}}^{\frac{1}{\alpha}}\lambda^{\frac{1}{\alpha}})^{\alpha}$ (consistency case) when $\eta_1, \eta_2$ go to $0$. Moreover, our upper bound distinguishes the effect of two possible sources of errors on the algorithm: (1) when removing jobs from the prediction ($\eta_1=0$ and $\eta_2$ goes to $1$), the upper bound degrades monotonically to $O(\frac{1}{\lambda})$. (2) when adding jobs to the prediction ($\eta_2=0$ and $\eta_1$ goes to $+\infty$), the upper bound first degrades, then improves again, with an optimal asymptotic rate of $\gamma_{\text{on}}$. This is since our algorithm mostly follows the 
% \hwcomment{optimal is not for the all case} 
online algorithm when the cost of the additional jobs dominates.

\section{The Extension to Small Deviations}
\label{sec:extension}

Note that in the definition of the prediction error $\eta$, a job $j$ is considered to be correctly predicted only if $r_j = \hat{r}_j$ and $p_j = \hat{p}_j$.
In this extension, we consider that  a job is correctly predicted even if its release time and processing time are shifted by a small amount.  We also allow each job to have some weight $v_j>0$, that can be shifted as well. Assuming an additional smoothness condition on the quality cost function $F(., .)$, which is satisfied for the energy plus flow time minimization problem and its variants, we propose and analyze an algorithm that generalizes the algorithm from the previous section. 

The algorithm, called \tpes \ and formally described in Appendix~\ref{sec:app_extension}, takes the same input parameters as Algorithm \tpe, with some additional shift tolerance parameter $\eta^{\text{shift}}\in [0,1)$ that is chosen by the algorithm designer. 
Two main ideas are to artificially increase
the predicted processing time $\hat{p}_j$ of each job $j$ (because the
true processing time $p_j$ of job $j$ could be shifted and be
slightly larger than $\hat{p}_j$) and to introduce small delays for
the job speeds (because the true release time $r_j$ of some
jobs j could be shifted and be slightly later than $\hat{r}_j$). Details can be found in Appendix \ref{sec:app_extension}.

\section{Experiments}
\label{sec:experiments}

We empirically evaluate the performance of Algorithm \tpes \ on both synthetic and real datasets. Specifically, we consider the energy plus flow time minimization problem where $F(S, \J) = \sum_{j \in \J} c_j - r_j$ and consider unit-work jobs (i.e., $p_j=1$ for all $j$) and fix $\alpha = 3$.

\subsection{Experimental setting}

\paragraph{Benchmarks.} \textbf{\tpes} is Algorithm~\ref{alg-no-ass-shift} with the default setting $\lambda = 0.02$, $\eta^{\text{shift}} = 1$ and $\sigma = 0.4$, where $\sigma$ is a parameter that controls the level of prediction error, that we call the error parameter. \textbf{\textsc{2-competitive}} is the $2$-competitive online algorithm from \cite{andrew2009optimal} that sets the speed at each time $t$ to $n(t)^{\frac{1}{\alpha}}$, where $n(t)$ is the number of jobs with $r_j\leq t$ unfinished at time $t$, and uses the Shortest Remaining Processing Time rule. 

\paragraph{Data sets.} We consider two synthetic datasets and a real dataset. For the synthetic data, we first generate a predicted set of jobs $\predJ$ and we fix the value of the error parameter $\sigma>0$. To create the true set of jobs $\J$, we generate, for each job $j\in \predJ$, some error $err(j)$ sampled i.i.d. from $\N(0, \sigma)$. The true set of jobs is then defined as $\J= \{(j,\hat{r}_j + err(j)): j \in \predJ\}$, which is the set of all predicted jobs, shifted according to $\{err(j)\}$. Note that for all $j\in \J$, $j\in \J^{\text{shift}}$ only if $|\hat{r}_j - r_j| = |err(j)| <\sft$. Hence, a larger $\sigma>0$ corresponds to a larger prediction error $\eta^g$. For the first synthetic dataset, called the periodic dataset, the prediction is a set of $n = 300$ jobs, with $i^{th}$ job's arrival $r_i = i / \alpha$. For the second synthetic dataset, we generate the prediction by using a power-law distribution. More precisely, for each time $t\in \{1,\ldots, T\}$, where we fix $T=75$, the number of jobs' arrivals at time $t$ is set to $M(1-p(a))$, where $p(a)$ is sampled from a power law distribution of parameter $a$, and $M$ is some scaling parameter. In all experiments, we use the values $a = 100$, $M = 500$. 
 
We also evaluate the two algorithms on the College Message dataset from the SNAP database \cite{panzarasa2009patterns}, where the scheduler must process messages that arrive over $9$ days, each with between 300 and 500 messages. We first fix the error parameter $\sigma>0$, then, for each day, we define the true set $\J$ as the arrivals for this day, and we create the predictions $\predJ$ by adding some error $err(i)$ to the release time of each job $i$, where $err(i)$ is sampled i.i.d. from $\N(0, \sigma$).
 
 % We also study the performance of \tpes \ under a large standard deviation of the shift under the setting of the periodic arrival.

 % we use the same instances as the first set of experiments and 
 % fixed the standard deviation of the shift to be $\N(0, 0.4)$ to evaluate the performance of choosing different values of the input parameters $\eta^{\text{shift}}$ and $\lambda$ for \tpes in periodic arrival. For each figure, the competitive ratio achieved by the different algorithms is averaged over $10$ instances generated i.i.d. as described in the experiment set 1. 
 % Additional details of the experiment setup are provided in Appendix~\ref{sec:appexperiments}.  
 
\subsection{Experiment results}

 \begin{figure*}[t!]
 	\centering
  \setlength{\belowcaptionskip}{-8pt}
    \subfigure{\includegraphics[width=0.24\textwidth]{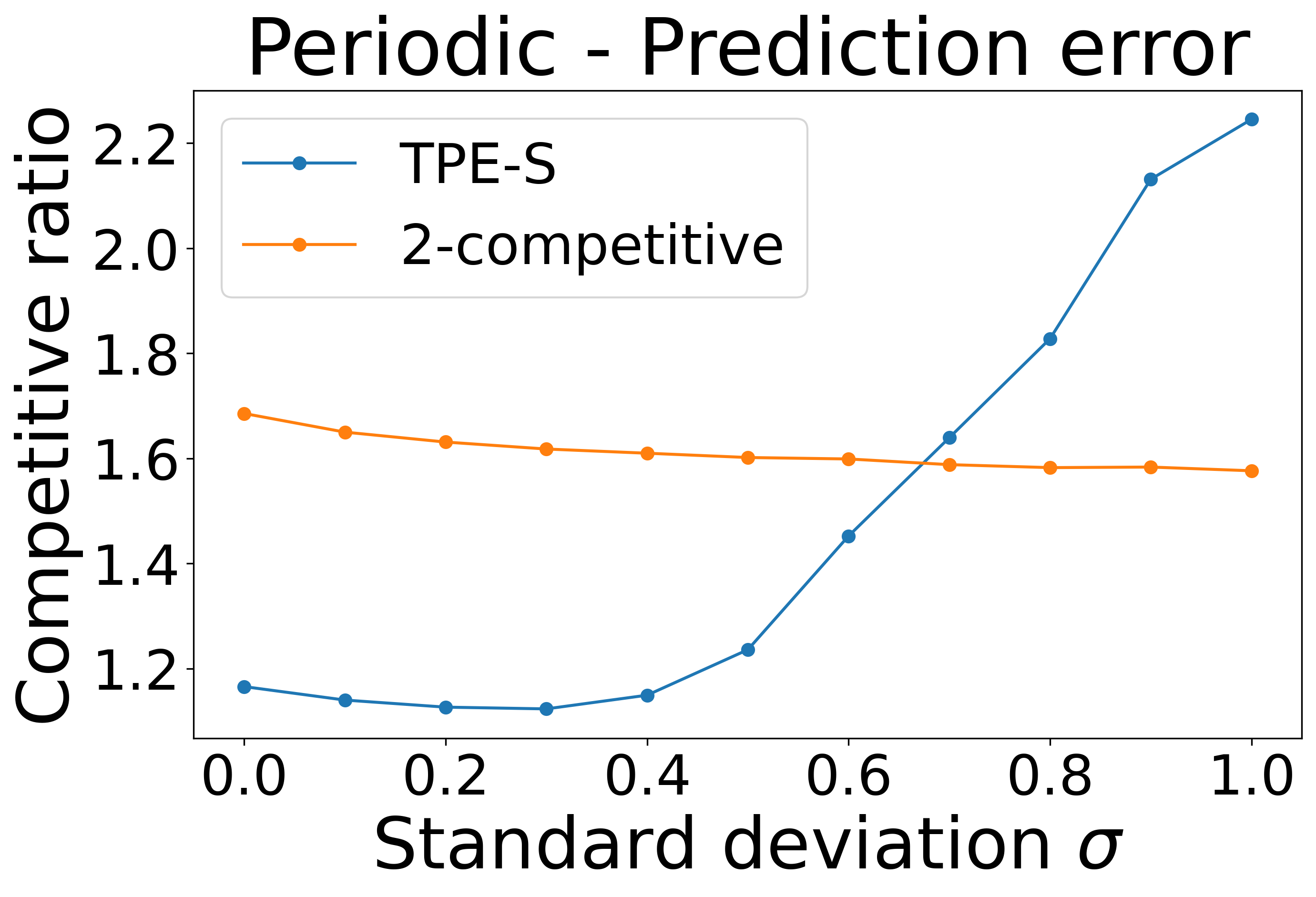}} 
    \subfigure{\includegraphics[width=0.24\textwidth]{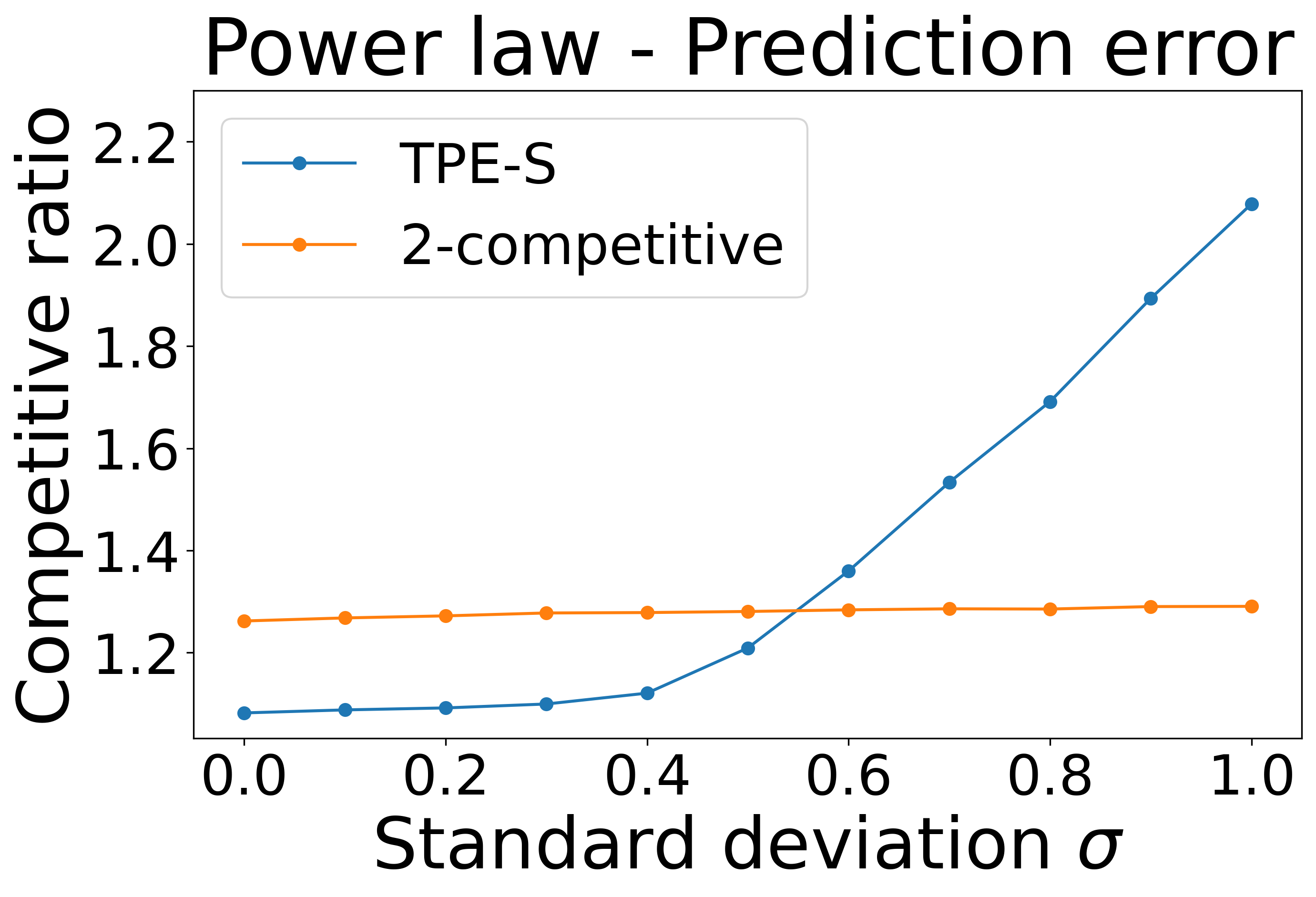}} \subfigure{\includegraphics[width=0.24\textwidth]{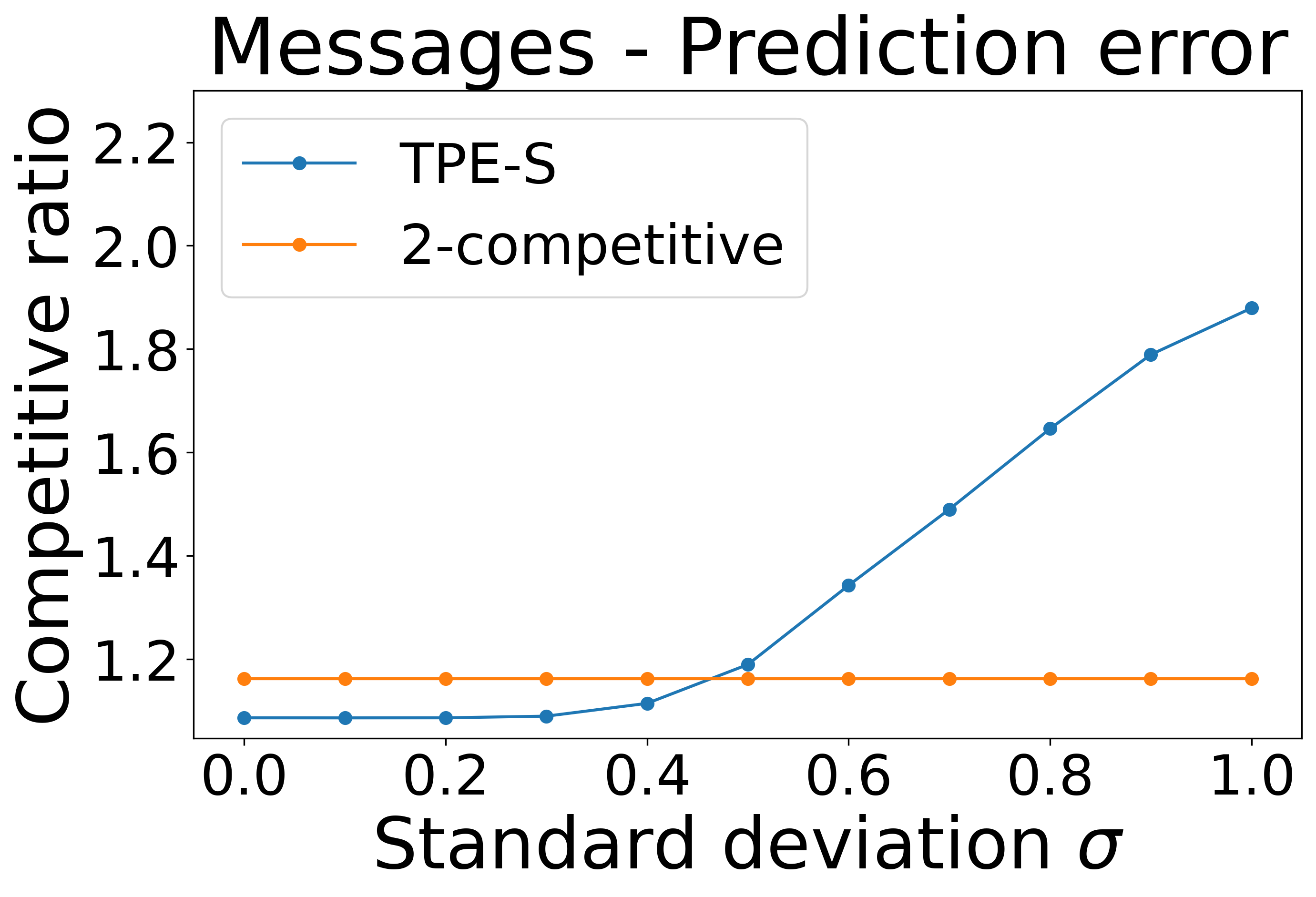}}
    % \caption{The competitive ratio achieved by our algorithm, \tpes, and the benchmark algorithm, as a function of the error parameter $\sigma$.}
    \subfigure{\includegraphics[width=0.23\textwidth]{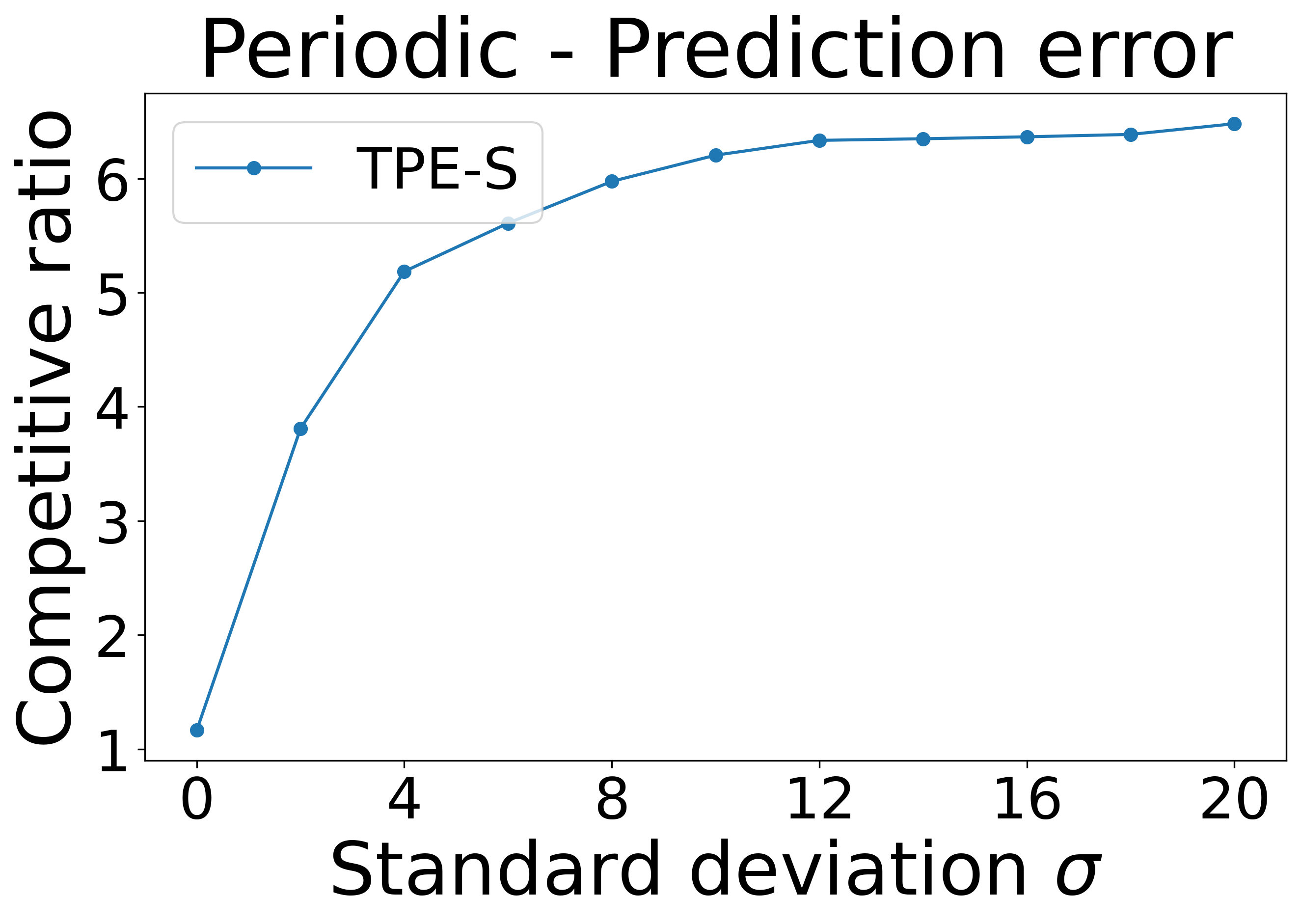} }
    \caption{\footnotesize{The competitive ratio achieved by our algorithm, \tpes, 
    and the benchmark algorithm, as a function of the error parameter $\sigma$ (from left-most to the second from the right), and the competitive ratio of \tpes \ for a larger range of $\sigma$, as a function of $\sigma$ (right-most).}}
    \label{fig:exp:1}  
 \end{figure*}

\begin{figure}[t!]
\centering
\setlength{\belowcaptionskip}{-8pt}
\subfigure{\includegraphics[width=0.23\textwidth]{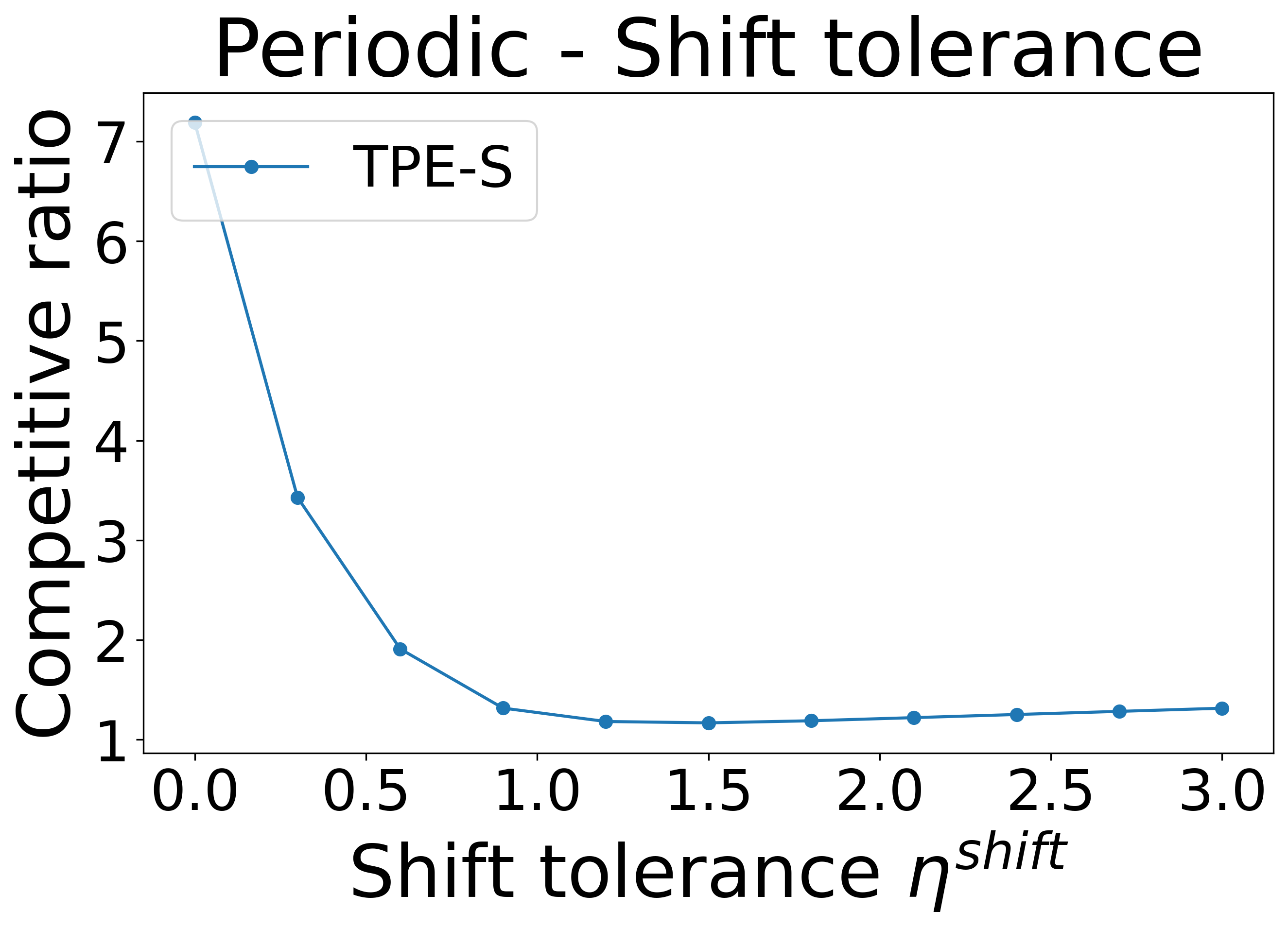}}
\subfigure{\includegraphics[width=0.26\textwidth]{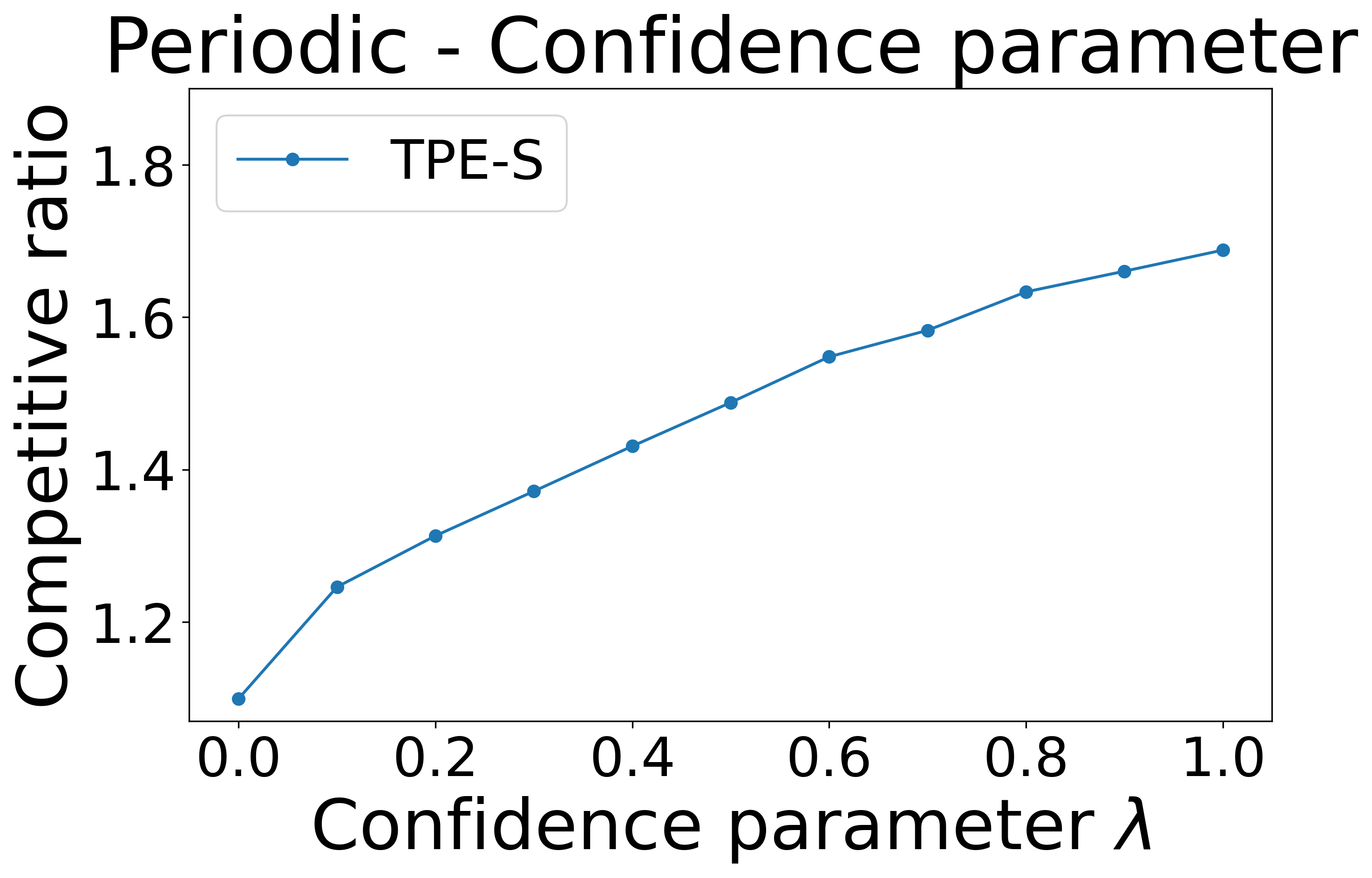}}
\caption{\footnotesize{The competitive ratio achieved by our algorithm, \tpes,  
% and the benchmark algorithm, 
as a function of the shift tolerance parameter $\eta^{\text{shift}}$ (left) and as a function of the confidence parameter $\lambda$ (right).}}\label{fig:exp:2}
\end{figure}

 For each of the synthetic datasets, the competitive ratio achieved by the different algorithms is averaged over $10$ instances generated i.i.d., and for the real dataset, it is averaged over the arrivals for each of the 9 days.
 
\paragraph{Experiment set 1.} We first evaluate the performance of the algorithms as a function of the error parameter $\sigma$. In Figure~\ref{fig:exp:1}, we observe that \tpes \ outperforms \textsc{2-competitive} \ when the error parameter is small. In the right-most figure of Figure~\ref{fig:exp:1}, the competitive ratio of \tpes \ plateaus when the value of $\sigma$ increases, which is consistent with our bounded robustness guarantee.

% Note that, since \textsc{2-competitive} does not use the predictions, its performance remains constant as a function of $\sigma$.

 \paragraph{Experiment set 2.}
In the second set of experiments, we study the impact of the parameters $\eta^{\text{shift}}$ and $\lambda$ of the algorithm for the periodic dataset (results for the other datasets can be found in Appendix~\ref{sec:appexperiments})  and fix $\sigma = 0.4$. In the left plot of Figure~\ref{fig:exp:2}, we observe the importance of allowing some shift in the predictions as the performance of our algorithms first rapidly improves as a function of $\eta^{\text{shift}}$ and then slowly deteriorates. The rapid improvement is because an increasing number of jobs are treated by the algorithm as being correctly predicted when $\eta^{\text{shift}}$  increases. Next, in the right plot, we observe that the competitive ratio deteriorates as a function of $\lambda$, which implies that the algorithm can completely skip the first phase that ignores the predictions and run the second phase that combines the offline and online schedules when the prediction error is not too large. Note, however, that a larger value of $\lambda$ leads to a better competitive ratio when the predictions are incorrect. Hence, there is a general trade-off here.

\section*{Acknowledgements} Eric Balkanski was supported by NSF grants CCF-2210502 and IIS-2147361.
Clifford Stein was supported in part by NSF grant  CCF-2218677,  ONR grant ONR-13533312, and by the Wai T. Chang Chair in Industrial Engineering and Operations Research.
Hao-Ting Wei was supported in part by NSF grant  CCF-2218677 and ONR grant ONR-13533312.

\newpage
% \nocite{*}%print all references

\bibliographystyle{plainnat}
\bibliography{reference}

\begin{thebibliography}{29}
\providecommand{\natexlab}[1]{#1}
\providecommand{\url}[1]{\texttt{#1}}
\expandafter\ifx\csname urlstyle\endcsname\relax
  \providecommand{\doi}[1]{doi: #1}\else
  \providecommand{\doi}{doi: \begingroup \urlstyle{rm}\Url}\fi

\bibitem[Albers(2010)]{albers2010energy}
Susanne Albers.
\newblock Energy-efficient algorithms.
\newblock \emph{Communications of the ACM}, 53\penalty0 (5):\penalty0 86--96, 2010.

\bibitem[Albers and Fujiwara(2007)]{albers2007energy}
Susanne Albers and Hiroshi Fujiwara.
\newblock Energy-efficient algorithms for flow time minimization.
\newblock \emph{ACM Transactions on Algorithms (TALG)}, 3\penalty0 (4):\penalty0 49--es, 2007.

\bibitem[Andrew et~al.(2009)Andrew, Wierman, and Tang]{andrew2009optimal}
Lachlan~LH Andrew, Adam Wierman, and Ao~Tang.
\newblock Optimal speed scaling under arbitrary power functions.
\newblock \emph{ACM SIGMETRICS Performance Evaluation Review}, 37\penalty0 (2):\penalty0 39--41, 2009.

\bibitem[Antoniadis et~al.(2022)Antoniadis, Ganje, and Shahkarami]{antoniadis2021novel}
Antonios Antoniadis, Peyman~Jabbarzade Ganje, and Golnoosh Shahkarami.
\newblock A novel prediction setup for online speed-scaling.
\newblock In Artur Czumaj and Qin Xin, editors, \emph{18th Scandinavian Symposium and Workshops on Algorithm Theory, {SWAT} 2022, June 27-29, 2022, T{\'{o}}rshavn, Faroe Islands}, volume 227 of \emph{LIPIcs}, pages 9:1--9:20, 2022.
\newblock \doi{10.4230/LIPIcs.SWAT.2022.9}.
\newblock URL \url{https://doi.org/10.4230/LIPIcs.SWAT.2022.9}.

\bibitem[Balkanski et~al.(2022{\natexlab{a}})Balkanski, Gkatzelis, and Tan]{balkanski2022strategyproof}
Eric Balkanski, Vasilis Gkatzelis, and Xizhi Tan.
\newblock Strategyproof scheduling with predictions.
\newblock \emph{arXiv preprint arXiv:2209.04058}, 2022{\natexlab{a}}.

\bibitem[Balkanski et~al.(2022{\natexlab{b}})Balkanski, Ou, Stein, and Wei]{balkanski2022scheduling}
Eric Balkanski, Tingting Ou, Clifford Stein, and Hao-Ting Wei.
\newblock Scheduling with speed predictions.
\newblock \emph{arXiv preprint arXiv:2205.01247}, 2022{\natexlab{b}}.

\bibitem[Bamas et~al.(2020)Bamas, Maggiori, Rohwedder, and Svensson]{BamasMRS20}
{\'{E}}tienne Bamas, Andreas Maggiori, Lars Rohwedder, and Ola Svensson.
\newblock Learning augmented energy minimization via speed scaling.
\newblock In Hugo Larochelle, Marc'Aurelio Ranzato, Raia Hadsell, Maria{-}Florina Balcan, and Hsuan{-}Tien Lin, editors, \emph{Advances in Neural Information Processing Systems 33: Annual Conference on Neural Information Processing Systems 2020, NeurIPS 2020, December 6-12, 2020, virtual}, 2020.
\newblock URL \url{https://proceedings.neurips.cc/paper/2020/hash/af94ed0d6f5acc95f97170e3685f16c0-Abstract.html}.

\bibitem[Bansal and Chan(2009)]{bansal2009weighted}
Nikhil Bansal and Ho-Leung Chan.
\newblock Weighted flow time does not admit o (1)-competitive algorithms.
\newblock In \emph{Proceedings of the twentieth annual ACM-SIAM symposium on Discrete algorithms}, pages 1238--1244. SIAM, 2009.

\bibitem[Bansal and Pruhs(2014)]{BansalP14}
Nikhil Bansal and Kirk Pruhs.
\newblock The geometry of scheduling.
\newblock \emph{{SIAM} J. Comput.}, 43\penalty0 (5):\penalty0 1684--1698, 2014.
\newblock \doi{10.1137/130911317}.
\newblock URL \url{https://doi.org/10.1137/130911317}.

\bibitem[Bansal et~al.(2007)Bansal, Kimbrel, and Pruhs]{bansal2007speed}
Nikhil Bansal, Tracy Kimbrel, and Kirk Pruhs.
\newblock Speed scaling to manage energy and temperature.
\newblock \emph{Journal of the ACM (JACM)}, 54\penalty0 (1):\penalty0 1--39, 2007.

\bibitem[Bansal et~al.(2008)Bansal, Chan, Lam, and Lee]{bansal2008scheduling}
Nikhil Bansal, Ho-Leung Chan, Tak-Wah Lam, and Lap-Kei Lee.
\newblock Scheduling for speed bounded processors.
\newblock In \emph{International Colloquium on Automata, Languages, and Programming}, pages 409--420. Springer, 2008.

\bibitem[Bansal et~al.(2010)Bansal, Pruhs, and Stein]{bansal2010speed}
Nikhil Bansal, Kirk Pruhs, and Cliff Stein.
\newblock Speed scaling for weighted flow time.
\newblock \emph{SIAM Journal on Computing}, 39\penalty0 (4):\penalty0 1294--1308, 2010.

\bibitem[Bansal et~al.(2011)Bansal, Bunde, Chan, and Pruhs]{bansal2011average}
Nikhil Bansal, David~P Bunde, Ho-Leung Chan, and Kirk Pruhs.
\newblock Average rate speed scaling.
\newblock \emph{Algorithmica}, 60\penalty0 (4):\penalty0 877--889, 2011.

\bibitem[Bansal et~al.(2013)Bansal, Chan, and Pruhs]{bansal2013speed}
Nikhil Bansal, Ho-Leung Chan, and Kirk Pruhs.
\newblock Speed scaling with an arbitrary power function.
\newblock \emph{ACM Transactions on Algorithms (TALG)}, 9\penalty0 (2):\penalty0 1--14, 2013.

\bibitem[Brooks et~al.(2000)Brooks, Bose, Schuster, Jacobson, Kudva, Buyuktosunoglu, Wellman, Zyuban, Gupta, and Cook]{brooks2000power}
David~M Brooks, Pradip Bose, Stanley~E Schuster, Hans Jacobson, Prabhakar~N Kudva, Alper Buyuktosunoglu, John Wellman, Victor Zyuban, Manish Gupta, and Peter~W Cook.
\newblock Power-aware microarchitecture: Design and modeling challenges for next-generation microprocessors.
\newblock \emph{IEEE Micro}, 20\penalty0 (6):\penalty0 26--44, 2000.

\bibitem[Cho et~al.(2022)Cho, Henderson, and Shmoys]{cho2022scheduling}
Woo-Hyung Cho, Shane Henderson, and David Shmoys.
\newblock Scheduling with predictions.
\newblock \emph{arXiv preprint arXiv:2212.10433}, 2022.

\bibitem[Im et~al.(2021)Im, Kumar, Montazer~Qaem, and Purohit]{im2021non}
Sungjin Im, Ravi Kumar, Mahshid Montazer~Qaem, and Manish Purohit.
\newblock Non-clairvoyant scheduling with predictions.
\newblock In \emph{Proceedings of the 33rd ACM Symposium on Parallelism in Algorithms and Architectures}, pages 285--294, 2021.

\bibitem[Lam et~al.(2008)Lam, Lee, To, and Wong]{lam2008speed}
Tak-Wah Lam, Lap-Kei Lee, Isaac~KK To, and Prudence~WH Wong.
\newblock Speed scaling functions for flow time scheduling based on active job count.
\newblock In \emph{European Symposium on Algorithms}, pages 647--659. Springer, 2008.

\bibitem[Lattanzi et~al.(2020)Lattanzi, Lavastida, Moseley, and Vassilvitskii]{LLMV20}
Silvio Lattanzi, Thomas Lavastida, Benjamin Moseley, and Sergei Vassilvitskii.
\newblock Online scheduling via learned weights.
\newblock In \emph{Proceedings of the 2020 ACM-SIAM Symposium on Discrete Algorithms (SODA)}, pages 1859--1877, 2020.

\bibitem[Lee et~al.(2021)Lee, Maghakian, Hajiesmaili, Li, Sitaraman, and Liu]{lee2021online}
Russell Lee, Jessica Maghakian, Mohammad Hajiesmaili, Jian Li, Ramesh Sitaraman, and Zhenhua Liu.
\newblock Online peak-aware energy scheduling with untrusted advice.
\newblock \emph{ACM SIGENERGY Energy Informatics Review}, 1\penalty0 (1):\penalty0 59--77, 2021.

\bibitem[Lindermayr and Megow(2022)]{lindermayr2022permutation}
Alexander Lindermayr and Nicole Megow.
\newblock Permutation predictions for non-clairvoyant scheduling.
\newblock In \emph{Proceedings of the 34th ACM Symposium on Parallelism in Algorithms and Architectures}, pages 357--368, 2022.

\bibitem[Lindermayr et~al.(2023)Lindermayr, Megow, and Rapp]{lindermayr2023speed}
Alexander Lindermayr, Nicole Megow, and Martin Rapp.
\newblock Speed-oblivious online scheduling: Knowing (precise) speeds is not necessary.
\newblock \emph{arXiv preprint arXiv:2302.00985}, 2023.

\bibitem[Lykouris and Vassilvitskii(2021)]{DBLP:journals/jacm/LykourisV21}
Thodoris Lykouris and Sergei Vassilvitskii.
\newblock Competitive caching with machine learned advice.
\newblock \emph{J. {ACM}}, 68\penalty0 (4):\penalty0 24:1--24:25, 2021.
\newblock \doi{10.1145/3447579}.
\newblock URL \url{https://doi.org/10.1145/3447579}.

\bibitem[Mahdian et~al.(2007)Mahdian, Nazerzadeh, and Saberi]{mahdian2007allocating}
Mohammad Mahdian, Hamid Nazerzadeh, and Amin Saberi.
\newblock Allocating online advertisement space with unreliable estimates.
\newblock In \emph{Proceedings of the 8th ACM conference on Electronic commerce}, pages 288--294, 2007.

\bibitem[Mitzenmacher(2020)]{MM20}
Michael Mitzenmacher.
\newblock {Scheduling with Predictions and the Price of Misprediction}.
\newblock In \emph{11th Innovations in Theoretical Computer Science Conference (ITCS 2020)}, volume 151 of \emph{Leibniz International Proceedings in Informatics (LIPIcs)}, pages 14:1--14:18, 2020.
\newblock ISBN 978-3-95977-134-4.

\bibitem[Mudge(2001)]{mudge2001power}
Trevor Mudge.
\newblock Power: A first-class architectural design constraint.
\newblock \emph{Computer}, 34\penalty0 (4):\penalty0 52--58, 2001.

\bibitem[Panzarasa et~al.(2009)Panzarasa, Opsahl, and Carley]{panzarasa2009patterns}
Pietro Panzarasa, Tore Opsahl, and Kathleen~M Carley.
\newblock Patterns and dynamics of users' behavior and interaction: Network analysis of an online community.
\newblock \emph{Journal of the American Society for Information Science and Technology}, 60\penalty0 (5):\penalty0 911--932, 2009.

\bibitem[Purohit et~al.(2018)Purohit, Svitkina, and Kumar]{KPZ18}
Manish Purohit, Zoya Svitkina, and Ravi Kumar.
\newblock Improving online algorithms via ml predictions.
\newblock In S.~Bengio, H.~Wallach, H.~Larochelle, K.~Grauman, N.~Cesa-Bianchi, and R.~Garnett, editors, \emph{Advances in Neural Information Processing Systems}. Curran Associates, Inc., 2018.

\bibitem[Yao et~al.(1995)Yao, Demers, and Shenker]{yao1995scheduling}
Frances Yao, Alan Demers, and Scott Shenker.
\newblock A scheduling model for reduced cpu energy.
\newblock In \emph{Proceedings of IEEE 36th annual foundations of computer science}, pages 374--382. IEEE, 1995.

\end{thebibliography}

% the recommnded bibstyle

%%%%%%%%%%%%%%%%%%%%%%%%%%%%%%%%%%%%%%%%%%%%%%%%%%%%%%%%%%%%

\newpage
\appendix
\onecolumn

\section*{Appendix}

\section{Consistent algorithms are not robust}
\label{app:consistency_robustness}

In this section, we show that any learning-augmented algorithm for the (GSSP) problem must incur some trade-off between robustness and consistency. Note that some impossibility results for general objective functions of the form $\text{cost}(S, \J) =  E(S) + F(S, \J)$ given in Section~\ref{sec:preliminaries} 
follow immediately 
from \cite{BamasMRS20}, since the problem of speed scaling with deadline constraints was studied there is a special case of (GSSP) (see Section~\ref{sec:various_objectives}). 

We prove here some impossibility results for a different family of objective functions, where the objective is to maximize the total energy plus flow time. This is one of the most widely studied objectives of the form in $\text{cost}(S, \J) =  E(S) + F(S, \J)$ given in Section~\ref{sec:preliminaries} (see for instance\cite{albers2007energy,andrew2009optimal,bansal2010speed,bansal2013speed}).
% \npcomment{citations}. 
Here, for all $j\in \J$, we let $c_S^j$ denote the completion time of job $j$ while following schedule $S$. The quality cost function studied in the remainder of this section is defined as: $F(S,\J) = \sum_{j\in \J} (c_S^j - r_j)$ and the total objective is $\cst(S, \J) = F(S,\J) + E(S)$. We recall that for this problem, the best possible online algorithm is the $2$-competitive algorithm from \cite{andrew2009optimal}.

\subsection{Warm-up: no $1$-consistent algorithm is robust}
\label{app:no-1-cons}

We show in this section that no algorithm that is perfectly consistent (i.e., achieves an optimal cost when the prediction is totally correct) can have a bounded competitive ratio in the case the prediction is incorrect. To show this property, we build an instance where a lot of jobs are predicted, but only one of them arrives. To achieve consistency, any online algorithm must `burn' a lot of energy during the first few time steps; however, in the case where only one job arrives, the algorithm ends up having wasted too much energy. This illustrates the necessity of a trade-off between robustness and consistency.\\

\begin{proposition}
    For the objective of minimizing total energy plus (non-weighted) flow time, there is no algorithm that is $1$-consistent and $o(\sqrt{n})$-robust, even if all jobs have unit-size work and if $\J\subseteq \predJ$.
\end{proposition}

\begin{proof}
Set $\alpha=2$ and consider an instance $(\predJ, \J')$ where $\predJ$ contains $n$ jobs of unit-size work such that the first job arrives at time $t=0$ and the remaining $n-1$ jobs arrive at time $t = \frac{1}{\sqrt{n}}$, and $\J'$ contains only the job that arrives at time $t=0$.

By using results from \cite{albers2007energy}, the optimal offline schedule for $\predJ$ is to schedule each job $i\in [n]$ at speed $\sqrt{n-i+1}$. Moreover, processing the first job any slower leads to a strictly worse cost. Hence, any algorithm that is $1$-consistent (i.e, achieves an optimal competitive ratio when the realization is exactly $\predJ$) must process the first job at speed $s_1(t) = \sqrt{n}$ for all $t \in [0, \frac{1}{\sqrt{n}}]$. In this case, the total objective is at least $\sqrt{n}^2 \cdot 1/\sqrt{n} + 1/\sqrt{n} = \sqrt{n} + 1/\sqrt{n}$.

However, by using results from \cite{albers2007energy}, the optimal objective for $\J'$ is 2 (with the speed of the single job arriving at time $t=0$ being set to 1). Hence, in the case where the realization is $\J'$, any algorithm that schedules the first job at speed $s_1(t) = \sqrt{n}$ has a competitive ratio at least $\frac{\sqrt{n} + 1/\sqrt{n}}{2}$.

Therefore, any $1$-consistent algorithm must have a robustness factor of at least $\frac{\sqrt{n} + 1/\sqrt{n}}{2}$.
\end{proof}

\subsection{Consistency-robustness trade-off}
\label{app:lowerbound-tradeoff}

In this section, we quantify more precisely the necessary trade-off between robustness and consistency. More precisely, we prove that there is a constant $C>0$ such that for any $\lambda$ small enough, any algorithm that is at most $(1+\lambda)$ consistent must be at least $C\sqrt{\frac{1}{\lambda}+1}$ robust. Moreover, letting $n_{\leq t}(\J) = |\{j\in \J: r_j\leq t\}|$ be the number of jobs of $\J$ that arrived before time $t$, we show that for the following natural notion of error: 
\[
\tilde{\eta}(\J,\J') = \frac{1}{\max\{|\J|, |\J'|\}}\max_{t\geq 0}\{|n_{\leq t}(\J) - n_{\leq t}(\J')|\},
\]
which mimics the probability density function of the predicted and realized jobs, 
this property remains true even if we assume a small prediction error $\tilde{\eta}(\J,\predJ)$. Hence, one cannot obtain a smooth algorithm relatively to this notion of error. This motivates the introduction of the more refined notion of error from Section~\ref{sec:preliminaries}. More specifically, we show the following lemma.

\begin{restatable}{lem}{lemlowerbound}
\label{lem:lower_bound}
For the objective of minimizing total energy plus (non-weighted) flow time, there are $\lambda'\in (0,1]$ and $C>0$ such that for any $\epsilon>0$, there is $M\in \mathbb{N}$ such that for all $\lambda \leq \lambda'$ and $n\geq M$, there is an instance $(\predJ_{n,\lambda, \epsilon}, \J_{n,\lambda, \epsilon})$ such that $|\predJ_{n,\lambda, \epsilon}| = n$ and $\tilde{\eta}(\predJ_{n,\lambda, \epsilon}, \J_{n,\lambda, \epsilon})\leq \epsilon$, and such that for any algorithm $\mathcal{A}$, 
\begin{itemize}
    \item either $\cst_{\A}(\predJ = \predJ_{n,\lambda, \epsilon}, \J = \predJ_{n,\lambda, \epsilon})> (1+\lambda)\cdot \OPT(\predJ_{n,\lambda, \epsilon})$ \textbf{(large consistency factor)}
    \item or $\cst_{\A}(\predJ = \predJ_{n,\lambda, \epsilon}, \J = \J_{n,\lambda, \epsilon})\geq C\sqrt{\frac{1}{\lambda}+1}\cdot\OPT(\J_{n,\lambda, \epsilon})$ \textbf{(large robustness factor)}.
\end{itemize}
\end{restatable}

The rest of this section is dedicated to the proof of Lemma \ref{lem:lower_bound}.

We first describe our lower bound instance. In the remainder of this section, we set $\alpha = 2$.\\

\noindent\textbf{Lower bound instance.} Let $\lambda \in (0,1], n\in \mathbb{N}$ and $\epsilon>0$. We construct an instance $(\predJ_{n,\lambda, \epsilon}, \J_{n,\lambda, \epsilon})$ where the jobs in $\predJ_{n,\lambda, \epsilon}$ can be organized in three different groups.
\begin{enumerate}
    \item Group $A_{n,\lambda, \epsilon}$ is composed of $\frac{4}{3}\lambda\epsilon n$ jobs that all arrive at time $0$.
    \item Group $B_{n,\lambda, \epsilon}$ is composed of $\epsilon n$ jobs that all arrive at time $t_A:=\frac{4\lambda\epsilon n}{3\sqrt{\epsilon n (1 +\frac{4}{3}\lambda)}}$.
    \item Group $C_{n,\lambda, \epsilon}$ consists of $n$ dummy jobs, where, for some $t'>>0$, each job $j\in [n]$ arrives at time $t' + j$.
\end{enumerate}

Next, we define $\J_{n,\lambda, \epsilon}$ as the union of jobs in $A_{n,\lambda, \epsilon}$ and $C_{n,\lambda, \epsilon}$. Note that by construction, we have $
\tilde{\eta}(\predJ_{n,\lambda, \epsilon}, \J_{n,\lambda, \epsilon})\leq \epsilon$.\\

We now state and prove a few useful lemmas. 
\begin{lemma}
\label{lem:opt_cost_bounds}
Let $\K$ be a set of $n$ jobs that all arrive at some time $t\geq 0$. Then, we have
\[
\frac{4}{3}n^{3/2}\leq \OPT(\K) \leq \frac{4}{3}n^{3/2} + o(n^{3/2}).
\]
\end{lemma}

\begin{proof}

By \cite{albers2007energy}, the optimal schedule is to run each job $i$ at speed $s_i = \sqrt{n-i+1}$.  The total cost is as follows:

\begin{align*}
    \cst(S^*(\K)) &= F(S^*(\K)) + E(S^*(\K))\\
    & = \sum_{i=1}^{n}\sum_{j=1}^{i}\frac{1}{\sqrt{n-j+1}} + \sum_{i=1}^{n}\frac{1}{\sqrt{n-i+1}}{\sqrt{n-i+1}}^2\\
    & = \sum_{j=1}^{n}\frac{1}{\sqrt{n-j+1}}\sum_{i = j}^{n}1 + \sum_{i=1}^{n}{\sqrt{n-i+1}}\\ 
    & = 2 \sum_{i=1}^{n}{\sqrt{n-i+1}}
\end{align*}

Hence we have 
\begin{equation*}
    2\int_{0}^{n} \sqrt{x} dx \le \cst(S^*(\K)) \le 2\int_{1}^{n+1} \sqrt{x} dx 
\end{equation*}
\begin{equation*}
    \Rightarrow \frac{4}{3}n^{3/2} \le \cst(S^*(\K)) \le \frac{4}{3}[(n+1)^{3/2}-1] = \frac{4}{3}[n^{3/2}+o(n^{3/2})]
\end{equation*}

\end{proof}

\begin{lemma}
\label{lem:opt_cost_dummy_jobs}
Let $\K = \{j_1,\ldots, j_{|\K|}\}$ be a set of $|\K|$ jobs such that for all $i \in [|\K|]$, $|r_{i+1}-r_i|\geq 1$. Then, we have
\[
 \OPT(\K) = 2|\K|.
\]
\end{lemma}

\begin{proof}
By using \cite{albers2007energy}, the optimal solution is to run each job at speed $1$. The result follows immediately.
\end{proof}

\begin{lemma}
\label{lem:upper_bound_opt_Jpredn}
Let $\lambda \in (0,1], n\in \mathbb{N}$. Then, the optimal cost for jobs in $\predJ_{n,\lambda, \epsilon}$ is upper bounded as follows:
\begin{align*}
   & \OPT(\predJ_{n,\lambda, \epsilon})\leq \frac{4}{3}(\epsilon n)^{3/2} (1 + \lambda + o(\lambda) + o(1)) \\
&\qquad\qquad\qquad\qquad\qquad\qquad\qquad\;\text{as $\lambda \rightarrow 0 \text{ (independently of $\epsilon$)},\; n \rightarrow +\infty \text{ (for a fixed $\epsilon$)}$.}
\end{align*}

\end{lemma}

\begin{proof}

Consider the schedule $S$ which runs jobs in $A_{n,\lambda, \epsilon}$ at speed $\sqrt{\epsilon n (1+\frac{4}{3}\lambda)}$, and jobs in $B_{n,\lambda, \epsilon}$ at the optimal speeds for $\epsilon n$ jobs arriving at the same time, and jobs in $C_{n,\lambda, \epsilon}$ at speed 1.
Consider the cost of $S$ for all jobs in $A_{n,\lambda, \epsilon}$. Note that all the jobs in $A_{n,\lambda, \epsilon}$ are finished by time  $t_A=\frac{4\lambda\epsilon n}{3\sqrt{\epsilon n (1 +\frac{4}{3}\lambda)}}$.
Hence, we have 
\begin{align*}
    \cst(S_{[0,t_A]}, A_{n,\lambda, \epsilon}) &\le F(S_{[0,t_A]},A_{n,\lambda, \epsilon}) + E(S_{[0,t_A]})\\
    & \le t_A\frac{4}{3}\lambda\epsilon n + t_A\left(\sqrt{\epsilon n (1+\frac{4}{3}\lambda)}\right)^2\\
    & \le \frac{\frac{4}{3}\lambda\epsilon n\cdot\frac{4}{3}\lambda\epsilon n}{\sqrt{\epsilon n (1 +\frac{4}{3}\lambda)}} + \frac{4}{3}\lambda\epsilon n \sqrt{\epsilon n (1+\frac{4}{3}\lambda)} \\
    & = (\epsilon n)^{3/2} o(\lambda) + (\epsilon n)^{3/2}(\frac{4}{3}\lambda \sqrt{(1+\frac{4}{3}\lambda)}) \\
    & = (\epsilon n)^{3/2} [\frac{4}{3}\lambda+o(\lambda)].
\end{align*}

Let $t_B$ be the time at which $S$ finishes all jobs in $B_{n,\lambda, \epsilon}$. Recall that all $n$ jobs in $B_{n,\lambda, \epsilon}$ arrive at time $t_A$. By Lemma~\ref{lem:opt_cost_bounds}, we have $$\cst(S_{[t_A,t_B]},B_{n,\lambda, \epsilon}) \le \frac{4}{3}(\epsilon n)^{3/2}+o((\epsilon n)^{3/2}),$$

Since all jobs in $C_{n,\lambda, \epsilon}$ arrive at time $t'>> t_B$, we have $$\cst(S_{\geq t_B},C_{n,\lambda, \epsilon}) = 2n = o((\epsilon n)^{3/2}).$$

Therefore,  when $\lambda$ goes to $0$  and $n$ goes to  $+\infty$, the total cost of $S$ is upper bounded as follows:
\begin{align*}
    \cst(S,\predJ_{n,\lambda, \epsilon}) &=  \cst(S_{[0,t_A]}, A_{n,\lambda, \epsilon})+ \cst(S_{[t_A,t_B]},B_{n,\lambda, \epsilon}) + \cst(S_{\geq t_B},C_{n,\lambda, \epsilon}) \\
    &\leq \frac{4}{3}(\epsilon n)^{3/2} [1 + \lambda + o(\lambda) + o(1)].
\end{align*}

\end{proof}

\begin{lemma}
\label{lem:lower_bound_ratio}
There is $\lambda'\in (0,1]$ such that for any $\epsilon>0$, there is $M\in \mathbb{N}$ such that for all $\lambda \leq \lambda'$ and $n\geq M$, and for any schedule $S$ for $\predJ_{n,\lambda, \epsilon}$ which has at least $\lambda \epsilon n$ units of jobs from $A_{n,\lambda, \epsilon}$ remaining at time $t_A$, we have

\[
\frac{\cst(S, \predJ_{n,\lambda, \epsilon})}{\OPT(\predJ_{n,\lambda, \epsilon})} > \left(1+\frac{1}{4}\lambda\right).
\]
\end{lemma}

\begin{proof}
Let $\lambda\in (0,1]$ and $\epsilon>0$, and let $S$ be a schedule for $\predJ_{n,\lambda, \epsilon}$ which has at least $\lambda \epsilon n$ units of jobs from $A_{n,\lambda, \epsilon}$ remaining at time $t_A$.

Note that the cost of $S$ for times $t\geq t_A$ is at least the cost of an optimal schedule for the remaining $\lambda \epsilon n$ units of jobs from $A_{n,\lambda, \epsilon}$ and the $\epsilon n$ units of job from $B_{n,\lambda, \epsilon}$. By Lemma~\ref{lem:opt_cost_bounds}, we thus get that:

\[
\cst(S, \predJ_{n,\lambda, \epsilon})\geq \cst(S, A_{n,\lambda, \epsilon}\cup B_{n,\lambda, \epsilon}) \geq \frac{4}{3}((1+\lambda)\epsilon n)^{3/2}.
\]

Now, by Lemma~\ref{lem:upper_bound_opt_Jpredn}, we get that when $\lambda$ goes to $0$ (independently of $\epsilon$)  and $n$ goes to  $+\infty$,
\[
\cst(S^*(\predJ_{n,\lambda, \epsilon}))\leq \frac{4}{3}(\epsilon n)^{3/2} (1 + \lambda + o(\lambda) + o(1)).
\]

Hence,
\begin{align*}
    \frac{\cst(S, \predJ_{n,\lambda, \epsilon})}{\cst(S^*(\predJ_{n,\lambda, \epsilon}))} &\geq 
    \frac{\frac{4}{3}((1+\lambda)\epsilon n)^{3/2}}{\frac{4}{3}(\epsilon n)^{3/2} (1 + \lambda + o(\lambda) + o(1))}\\
    &= \left(1 +\frac{3}{2}\lambda + o(\lambda)\right)\cdot (1 - \lambda - o(\lambda) - o(1))\\
    &= 1 +\frac{1}{2}\lambda - o(\lambda) - o(1).
\end{align*}

Hence, there is $\lambda'\in (0,1]$ and $M\in \mathbb{N}$ (note that $\lambda'\in (0,1]$ is independent of $\epsilon$ while $M$ depends on it) such that if $\lambda\leq \lambda'$ and $n\geq M$, then 
\[
\frac{\cst(S, \predJ_{n,\lambda, \epsilon})}{\cst(S^*(\predJ_{n,\lambda, \epsilon}))}>\left(1+\frac{1}{4}\lambda\right).
\]

\end{proof}

\begin{lemma}
\label{lem:cost_first_third_A}
Let $\lambda \in (0,1], n\in \mathbb{N}$. Assume that $S$ schedules at least $\frac{1}{3}\lambda \epsilon n$ units of jobs from $A_{n,\lambda, \epsilon}$ from time $0$ to $t_A$. Then, there is a constant $C>0$ such that
\[
\cst(S, \J_{n,\lambda, \epsilon})\geq C(\epsilon n)^{3/2} \lambda \sqrt{1+\lambda}.
\]
\end{lemma}

\begin{proof}
For convenience of exposition, assume that $S$ schedules exactly $\frac{1}{3}\lambda \epsilon n$ units of jobs from $A_{n,\lambda, \epsilon}$ from time $0$ to $t$ (note that if $S$ schedules more work from $A_{n,\lambda, \epsilon}$, then the cost can only be higher). By using \cite{albers2007energy}, the optimal solution is to schedule each job $i$ at speed $s_i = \sqrt{\frac{\epsilon\lambda n}{3}-i+c+1}$, where $c$ is the unique constant such that $$\sum_{i = 1}^{\frac{\epsilon\lambda n}{3}} \frac{1}{\sqrt{\frac{\epsilon\lambda n}{3}-i+c+1}} = t_A .$$
To lower bound $c$, note that we then have
\begin{align*}
    t_A &\ge \sum_{i = \frac{1}{2}\frac{\epsilon\lambda n}{3}}^{\frac{\epsilon\lambda n}{3}} \frac{1}{\sqrt{\frac{\epsilon\lambda n}{3}- \frac{1}{2}\frac{\epsilon\lambda n}{3}+c+1}}\\
    & = \frac{\epsilon\lambda n}{6} \frac{1}{\sqrt{\frac{\epsilon\lambda n}{6}+c+1}}. \\
\end{align*}

By definition of $t_A$, we get 
\begin{align*}
    &\frac{4\lambda\epsilon n}{3\sqrt{\epsilon n (1 +\frac{4}{3}\lambda)}} \ge \frac{\epsilon\lambda n}{6} \frac{1}{\sqrt{\frac{\epsilon\lambda n}{6}+c+1}}\\
    &\iff (\frac{4}{3}6)^2(\frac{\epsilon\lambda n}{6}+c+1) \ge (1+\frac{4}{3}\lambda)\epsilon n \\
    &\iff c \ge c_2 \epsilon n -  c_3\epsilon\lambda n - 1. \tag*{with $c_2 = \frac{1}{64}, c_3 = -\frac{1}{48} + \frac{1}{6}$}
\end{align*}
And the corresponding energy consumption is:

\begin{align*}
    & \sum_{i = 1}^{\frac{\epsilon\lambda n}{3}}\sqrt{{\frac{\epsilon\lambda n}{3}}-i + c + 1}   \\
    & \ge \sum_{i = 1}^{\frac{\epsilon\lambda n}{6}}\sqrt{{\frac{\epsilon\lambda n}{3}}-\frac{\epsilon\lambda n}{6} + c_2 \epsilon n -  c_3\epsilon\lambda n } \\
    & \ge \frac{\epsilon\lambda n}{6} \sqrt{{\frac{\epsilon\lambda n}{3}}-\frac{\epsilon\lambda n}{6} + c_2 \epsilon n +(\frac{1}{48} -\frac{1}{6})\epsilon\lambda n } \\
    & = \frac{\epsilon\lambda n}{6} \sqrt{c_2 \epsilon n +\frac{1}{48}\epsilon\lambda n } \\
    & \ge C(\epsilon n)^{3/2} \lambda \sqrt{1+\lambda}.\tag*{ (\text{for some constant $C >0$})} 
\end{align*}

Therefore, any schedule $S$ that completes $\lambda\epsilon n$ jobs before time $t$ has cost lower bounded as:
$$\cst(S, \J_{n,\lambda, \epsilon}) \ge C(\epsilon n)^{3/2} \lambda \sqrt{1+\lambda}.$$

\end{proof}

We are now ready to present the proof of Lemma \ref{lem:lower_bound}.\\

\noindent\textbf{Proof of Lemma \ref{lem:lower_bound}.}
By Lemma~\ref{lem:lower_bound_ratio}, we have that there is a constant $\lambda'\in (0,1]$ such that for all $\epsilon>0$, there is $M\in \mathbb{N}$ such that for any algorithm $\A$ and $n\geq M$, and when running $\mathcal{A}$ with predictions $\predJ = \predJ_{n,\lambda, \epsilon}$ and realization $\J\in \{\predJ_{n,\lambda, \epsilon}, \J_{n,\lambda, \epsilon}\}$, then either $\A$ schedules at least $\frac{1}{3}\lambda \epsilon n $ units of jobs of $A_{n,\lambda, \epsilon}$ before time $t$, or the schedule $S$ output by $\A$ satisfies:
\[
\frac{\cst(S, \predJ_{n,\lambda, \epsilon})}{\cst(S^*(\predJ_{n,\lambda, \epsilon}))}>\left(1+\frac{1}{4}\lambda\right).
\]

Hence, if $\A$ achieves a consistency of at most $\left(1+\frac{1}{4}\lambda\right)$, $\A$ must schedule at least $\frac{1}{3}\lambda \epsilon n $ units of jobs of $A_{n,\lambda, \epsilon}$ before time $t$. However, we then have, by Lemma~\ref{lem:cost_first_third_A}, that for some constant $C>0$,
\[
\cst_{\A}(\predJ = \predJ_{n,\lambda, \epsilon}, \J = \J_{n,\lambda, \epsilon})\geq C(\epsilon n)^{3/2} \lambda \sqrt{1+\lambda}.
\]

On the other hand, assuming that $\J = \J_{n,\lambda, \epsilon}$, we get by Lemma~\ref{lem:opt_cost_bounds} and Lemma~\ref{lem:opt_cost_dummy_jobs} that
\[
\OPT(\J_{n,\lambda, \epsilon})\leq \frac{4}{3}(\lambda \epsilon n)^{3/2} + o((\epsilon n)^{3/2}) + 2n.
\]

Hence, we get that for some constant $C''>0$ and $n$ large enough, 
\begin{align*}
    \frac{\cst_{\A}(\predJ = \predJ_{n,\lambda, \epsilon}, \J = \J_{n,\lambda, \epsilon})}{\OPT(\J_{n,\lambda, \epsilon})} &\geq \frac{ C(\epsilon n)^{3/2} \lambda \sqrt{1+\lambda}}{\frac{4}{3}(\lambda \epsilon n)^{3/2} + o((\epsilon n)^{3/2})+ 2n} \geq C''\sqrt{\frac{1}{\lambda} +1 }.
\end{align*}
\qed
\section{Missing analysis from Section~\ref{sec:framework}}
\label{app:algorithm}

\lemunionoftwoschedule*

\begin{proof}

We first upper bound the quality cost $F(S, \J_1\cup \J_2)$ of the proposed schedule $S$. In each infinitesimal time interval $[t,t+dt]$ and for all $j\in \J_1$, $S$ processes $s_{1,j}(t)dt$ units of work of job $j$, and for each $j\in  \J_2$, $S$ processes $s_{2,j}(t)dt$ units of work of job $j$. Hence $S$ processes exactly the same amount of work for each job $j\in \J_1$ (resp. $j\in \J_2$) as $S_1$ (resp. $S_2$). We thus get that for all $t\geq 0$, \begin{equation}
\label{eq:work}
    w^S_j(t) = w^{S_1}_j(t) \text{ for all } j\in \J_1 \qquad \text{ and } \qquad w^S_j(t) =  w^{S_2}_j(t) \text{ for all }  j\in \J_2.
\end{equation} Therefore, 
\begin{align*}
    F(S, \J_1\cup\J_2)&\leq F(S, \J_1) + F(S, \J_2)\tag*{(\text{$F$} is sub-additive)}\\
    &= f\left(\{(W^S_j,j)\}_{j\in \J_1}\right)+f\left( \{(W^S_j,j)\}_{j\in \J_2}\right) \\
    &= f\left(\{(W^{S_1}_j,j)\}_{j\in \J_1}\right)+f\left( \{(W^{S_2}_j,j)\}_{j\in \J_2}\right) \tag*{(by (\ref{eq:work}))}\\
&= F(S_1, \J_1) + F(S_2, \J_2).
\end{align*}
Next, we upper bound the energy consumption $E(S)$ of the proposed schedule $S$.
\begin{align*}
    E(S) &= \int({s_1}(t)+ {s_2}(t))^{\alpha}\text{dt}\\
    &=  \sum_{i=0}^{\alpha}{\alpha \choose i} \int ({s_1}(t)^{\alpha})^{\frac{i}{\alpha}} ({s_2}(t)^{\alpha})^{\frac{\alpha-i}{\alpha}}\text{dt}\\
    & \le  \sum_{i=0}^{\alpha}{\alpha \choose i} \left(\int ({s_1}(t))^{\alpha}\text{dt}\right)^{\frac{i}{\alpha}} \left(\int({s_2}(t))^{\alpha}\text{dt}\right)^{\frac{\alpha-i}{\alpha}} \ \text{(H\"{o}lder's inequality)}\\
    &= E(S_1) + E(S_2) + \sum_{i=1}^{\alpha-1}{\alpha \choose i}E(S_1)^{\frac{i}{\alpha}}E(S_2)^{\frac{\alpha-i}{\alpha}}\\
    &\le E(S_1) +  E(S_2) + \sum_{i=1}^{\alpha-1}{\alpha \choose i}\text{cost}(S_1)^{\frac{i}{\alpha}}\text{cost}(S_2)^{\frac{\alpha-i}{\alpha}}.
\end{align*}

Therefore, the total cost of schedule $S$ can be upper bounded as follows:
\begin{align*}
    \cst(S, \J_1\cup\J_2) 
    &= F(S, \J_1\cup\J_2) + E(S)\\
    &\leq F(S_1, \J_1) + E(S_1) +  F(S_2, \J_2) + E(S_2)
    + \sum_{i=1}^{\alpha-1}{\alpha \choose i}\cst(S_1)^{\frac{i}{\alpha}}\cst(S_2)^{\frac{\alpha-i}{\alpha}}\\
    &=\left(\cst(S_1, \J_1)^{\frac{1}{\alpha}} + \cst(S_2, \J_2)^{\frac{1}{\alpha}}\right)^{\alpha}.
\end{align*}

\end{proof}

% \lemupperboundschedule*

\corJinterHat*

\begin{proof}
We prove the first part of the Corollary by contradiction. Assume that $\texttt{OPT}(\mathcal{J}\cap \predJ)< \left(1-\eta_2^{\frac{1}{\alpha}}\right)^{\alpha}\predT$. Next, by definition of the error $\eta_2$, we have $\texttt{OPT}(\predJ\setminus \J) =  \eta_2\cdot\predT$. Hence, by Lemma~\ref{lem:union-of-two-schedule}, there exists a schedule $S$ for $(\J \cap \predJ) \cup (\predJ\setminus \J) = \predJ$ such that
\begin{align*}
   \cst(S, \predJ) 
    &\le \left(\texttt{OPT}(\predJ\setminus \J)^{\frac{1}{\alpha}} + \texttt{OPT}(\mathcal{J}\cap \predJ)^{\frac{1}{\alpha}}\right)^{\alpha}\\
     &< \left((\eta_2\predT)^{\frac{1}{\alpha}} + \left(\left(1-\eta_2^{\frac{1}{\alpha}}\right)^{\alpha}\predT\right)^{\frac{1}{\alpha}}\right)^{\alpha}\\
    & = \predT,
\end{align*}

    which contradicts the definition of $\predT$ and ends the proof of the first result.
% \end{proof}

% % \corlbeta*

% \begin{proof}

We now show the second part of the Corollary. 

 Assume that $\texttt{OFF}(\J)\leq \lambda\texttt{OFF}(\predJ)$. Then, since $\textsc{OfflineAlg}$ is $\gamma_{\text{off}}$-competitive, we have
\[
\texttt{OPT}(\J)\leq
\texttt{OFF}(\J)\leq \lambda\texttt{OFF}(\predJ)\leq \lambda \gamma_{\text{off}}\texttt{OPT}(\predJ).
\]
 
 In particular, $\texttt{OPT}(\J\cap \predJ)\leq \lambda \gamma_{\text{off}}\predT$. Next, assume by contradiction, that $\eta_2 < \left(1-(\lambda \gamma_{\text{off}})^{\frac{1}{\alpha}}\right)^{\alpha}$, which implies that $\texttt{OPT}(\predJ\setminus\J)< \left(1-(\lambda \gamma_{\text{off}})^{\frac{1}{\alpha}}\right)^{\alpha}\predT$. Then, by Lemma~\ref{lem:union-of-two-schedule}, there exists a schedule $S$ for $(\J \cap \predJ) \cup (\predJ\setminus \J) = \predJ$ such that
\begin{align*}
   \cst(S, \predJ) 
    &\le \left(\texttt{OPT}(\predJ\setminus \J)^{\frac{1}{\alpha}} + \texttt{OPT}(\mathcal{J}\cap \predJ)^{\frac{1}{\alpha}}\right)^{\alpha}\\
     &< \left(((\lambda \gamma_{\text{off}})\predT)^{\frac{1}{\alpha}} + \left(\left(1-(\lambda \gamma_{\text{off}})^{\frac{1}{\alpha}}\right)^{\alpha}\predT\right)^{\frac{1}{\alpha}}\right)^{\alpha}\\
    & = \predT,
\end{align*}

    which contradicts the definition of $\predT$. Hence, $\eta_2 \geq \left(1-(\lambda \gamma_{\text{off}})^{\frac{1}{\alpha}}\right)^{\alpha}$.
\end{proof}

\maintheorem*

\begin{proof}
First, assume that that for all $t\geq 0$,  $\texttt{OFF}(\mathcal{J}_{\leq t})\leq \lambda \texttt{OFF}(\predJ)$ (i.e., \tpe \ never goes through lines 6-10). Then, the schedule $S$ returned by the algorithm is obtained by running the $\gamma_{\text{on}}$-competitive algorithm \textsc{OnlineAlg} on $\J$, hence 
\begin{equation}
\label{eq:ub_cost1}
    \cst(S,\J)\leq \gamma_{\text{on}} \cdot\trueT.
\end{equation} 
% Note that by Corollary~\ref{cor:lb_eta2}, this only happens if $\eta_2 \geq \left(1-\lambda^{\frac{1}{\alpha}}\right)^{\alpha}$.

Next, assume that there is $\tl\geq 0$ such that $\texttt{OFF}(\mathcal{J}_{\leq \tl})>\lambda \texttt{OFF}(\predJ)$. Since $\textsc{OfflineALG}$ is $\gamma_{\text{off}}$-competitive, we immediately get:
\[
\trueT \geq \texttt{OPT}(\mathcal{J}_{\leq \tl})\geq \frac{\texttt{OFF}(\mathcal{J}_{\leq \tl})}{\gamma_{\text{off}}}>\frac{\lambda}{\gamma_{\text{off}}} \texttt{OFF}(\predJ)\geq \frac{\lambda}{\gamma_{\text{off}}} \predT.
\]

By Corollary~\ref{cor:JinterHatJ-LB} and by definition of the error $\eta_1$, we also get the following lower bound on the optimal schedule:
    \begin{align*}
        \trueT &\geq \texttt{OPT}(\mathcal{J}\setminus \predJ) + \texttt{OPT}(\mathcal{J}\cap \predJ)\geq \eta_1\predT + \left(1-\eta_2^{\frac{1}{\alpha}}\right)^{\alpha}\predT.
    \end{align*}
    
Therefore, 
\begin{equation}
    \label{eq:lb_opt}
    \trueT\geq \max\left\{\frac{\lambda}{\gamma_{\text{off}}}, \eta_1+ \left(1-\eta_2^{\frac{1}{\alpha}}\right)^{\alpha}\right\}\predT.
\end{equation}

Now, by Lemma~\ref{lem:upper_bound_schedule}, the cost of the schedule S output by \tpe \ is always upper bounded as follows:
\begin{equation}
\label{eq:ub_cost}
    \cst(S,\J)\leq \predT\left(\gamma_{\text{off}}^{\frac{1}{\alpha}}
    % {\gamma_{\text{on}}_{off}}^{\frac{1}{\alpha}} 
    + {\gamma_{\text{on}}}^{\frac{1}{\alpha}}((\lambda \gamma_{\text{off}})^{\frac{1}{\alpha}} + \eta_1^{\frac{1}{\alpha}})\right)^{\alpha}.
\end{equation}

    Hence, we get the following upper bound on the competitive ratio of \tpe:
    \begin{align*}
    &\;\;\;\frac{\cst(S, \J)}{\trueT}\\
    &= \mathbf{1}_{\texttt{OFF}(\mathcal{J})\leq \lambda \texttt{OFF}(\predJ)}\frac{\cst(S, \J)}{\trueT} + \mathbf{1}_{\texttt{OFF}(\mathcal{J})> \lambda \texttt{OFF}(\predJ)}\frac{\cst(S, \J)}{\trueT}
    \\
    &\le \mathbf{1}_{\texttt{OFF}(\mathcal{J})\leq \lambda \texttt{OFF}(\predJ)}\gamma_{\text{on}} +  \mathbf{1}_{\texttt{OFF}(\mathcal{J})> \lambda \texttt{OFF}(\predJ)}\frac{\left( \gamma_{\text{off}}^{\frac{1}{\alpha}}
    % {\gamma_{\text{on}}_{off}}^{\frac{1}{\alpha}} 
    + {\gamma_{\text{on}}}^{\frac{1}{\alpha}}((\lambda \gamma_{\text{off}})^{\frac{1}{\alpha}} + \eta_1^{\frac{1}{\alpha}})\right)^{\alpha}}{\max\left\{\frac{\lambda}{\gamma_{\text{off}}}, \eta_1+ \left(1-\eta_2^{\frac{1}{\alpha}}\right)^{\alpha}\right\}}\tag*{by (\ref{eq:ub_cost1}),(\ref{eq:lb_opt}),(\ref{eq:ub_cost})}.
    \end{align*}
\end{proof}

\corconsistency*

\begin{proof}
We start by the consistency. Since we assumed that \textsc{OfflineAlg} is optimal, by Corollary~\ref{cor:JinterHatJ-LB}, we have that for all $\lambda \in (0,1)$, $\texttt{OFF}(\J)>\lambda\texttt{OFF}(\predJ)$, when $\eta_2 =0$. The result follows by an immediate upper bound on the competitive ratio in this case.

We now show the robustness. First note that if $\texttt{OFF}(\J)\leq\lambda\texttt{OFF}(\predJ)$, then $ \frac{\cst(S, \J)}{\trueT}\leq \gamma_{\text{on}}\leq \max\{\gamma_{\text{on}},\frac{1+\gamma_{\text{on}}2^{2\alpha}\lambda^{\frac{1}{\alpha}}}{\lambda} \}$ for any $\lambda \in [0,1]$. Now, assume that $\texttt{OFF}(\J)>\lambda\texttt{OFF}(\predJ)$. We then have
\begin{align*}
    \frac{\cst(S, \J)}{\trueT}&
    \leq \frac{\left( 1
    + {\gamma_{\text{on}}}^{\frac{1}{\alpha}}(\lambda^{\frac{1}{\alpha}} + \eta_1^{\frac{1}{\alpha}})\right)^{\alpha}}{\max\left\{\lambda, \eta_1+ \left(1-\eta_2^{\frac{1}{\alpha}}\right)^{\alpha}\right\}}\leq\frac{\left(1 + {\gamma_{\text{on}}}^{\frac{1}{\alpha}}(\lambda^{\frac{1}{\alpha}} + \eta_1^{\frac{1}{\alpha}})\right)^{\alpha}}{\max\{\lambda, \eta_1\}}.
\end{align*}

If $\eta_1\leq \lambda$, we get:
\[
\frac{\cst(S, \J)}{\trueT}\leq\frac{\left(1 + {\gamma_{\text{on}}}^{\frac{1}{\alpha}}(2\lambda^{\frac{1}{\alpha}})\right)^{\alpha}}{\lambda}\leq \frac{1+\gamma_{\text{on}}2^{2\alpha}\lambda^{\frac{1}{\alpha}}}{\lambda} ,
\]
and if $\eta_1\geq \lambda$, we get:
\[\frac{\cst(S, \J)}{\trueT}\leq\frac{\left(1 + {\gamma_{\text{on}}}^{\frac{1}{\alpha}}(2\eta_1^{\frac{1}{\alpha}})\right)^{\alpha}}{\eta_1}\leq \frac{1+\gamma_{\text{on}} 2^{2\alpha}\max\{\eta_1, \eta_1^{\frac{1}{\alpha}}\}}{\eta_1}.
\]
Since the above function reaches its maximum value over $[\lambda, +\infty[$ for $\eta_1 = \lambda$, this immediately yields the result.
\end{proof}
\section{The Extension with Job Shifts (Full version)}
\label{sec:app_extension}

Note that in the definition of the prediction error $\eta$, a job $j$ is considered to be correctly predicted only if $r_j = \hat{r}_j$ and $p_j = \hat{p}_j$.
In this section, we consider an extension where a job is considered to be correctly predicted even if the release time and processing time are shifted by a small amount.  In this extension, we also allow each job to have some weight $v_j>0$, that can be shifted as well. We propose and analyze an algorithm that generalizes the algorithm from the previous section.
% To overcome this limitation, we introduce a third dimension $\etan$ to the error measure, which we call the \textit{noise tolerance}. More precisely, we consider in this section the specific case where the cost function $F(S,\J)$ satisfies some smoothness condition and we show how to generalise Algorithm~\npcomment{} to an algorithm whose competitive ratio also gracefully degrades with $\eta_{noise}$.\\

\noindent\textbf{Motivating example.} Consider the objective of minimizing energy plus flow time with $\alpha = 2$. Let $(\predJ, \J)$ be an instance  where $\J$ has $n$ jobs with weight $w = 1.01$ and processing time $p = 0.99$, all released at time $r =0.1$, and $\predJ$ has $n$ jobs with weight $w= 1$ and processing time $p =1$, all released at time $r = 0$. Since $\predJ\setminus \J = \predJ$, we have here that $\eta_1 = \OPT(\predJ\setminus\J) =  \OPT(\predJ) = \Omega(n^{3/2})$ (by Lemma~\ref{lem:opt_cost_bounds}), whereas it seems reasonable to say that $\predJ$ was a 'good' prediction for instance $\J$, since it accurately represents the pattern of the jobs in $\J$.

In this section, we assume that the quality of cost function $F$ is such that the total cost function $E + F$ satisfies a smoothness condition, which we next define.

\noindent\textbf{Smooth objective function.} Let $\mathbb{J}$ denote the collection of all sets of jobs. We say that a function $\beta: \mathbb{J}\longrightarrow \mathbb{R}$ is smooth if for all $\J\in \mathbb{J}$, $\{r_j'\}_{j\in \J}\geq 0$ and $\{\eta_j\}_{j\in \J}\geq 0$, we have $\beta(\J')\leq (1+\max_j\eta_j)\beta(\J)$, where $\J'=\{(j,r_j', p_j(1+\eta_j), v_j(1+\eta_j)\}$. $\beta$ is monotone if for all $\J''\subseteq \J$, we have $\beta(\J'')\leq \beta(\J)$.

We say that a cost function $\cst(.,.)$ is $\beta$-smooth if there is a smooth monotone function $\beta(.)\geq 1$ such that for all $\eta, \eta'\in [0,1]$, $\J_1,\J_2$ with $|\J_1| = |\J_2|$ and bijection $\pi:\J_1\longrightarrow\J_2$, and for all $S_1$ and $S_2$ feasible for $\J_1$ and $\J_2$: 
\begin{itemize}
    \item \textbf{(smoothness of optimal cost).} If for all $j\in \J_1$, $|r_j -r_{\pi(j)}|\leq \eta'$, $p_j\leq p_{\pi(j)}(1+ \eta)$ and $v_j\leq v_{\pi(j)}(1+ \eta)$, then \[\OPT(\J_1)\leq (1+\beta(\J_1)\eta)\OPT(\J_2) + \beta(\J_1)|\J_1|\eta'.\]
    \item \textbf{(shifted work profile for dominated schedule.)} If for all $j\in \J$, $p_j \leq p_{\pi(j)}$, $v_j \leq v_{\pi(j)}$, $r_j \geq r_{\pi(j)} - \eta'$ and for all $t\geq r_{\pi(j)}+\eta'$, $w_{S_1}^j(t) \leq w_{S_2}^{\pi(j)}(t- \eta')$, where $w_{S_1}^j(t)$ (resp. $w_{S_2}^j(t)$) denotes the remaining amount of work for $j$ a time $t$ for $S_1$ (resp. $S_2$), then \[F(S_1, \J_1)\leq  F(S_2,\J_2)+  \beta(\J_1)|\J_1|\eta'.\]
\end{itemize}
 
In other words, if $\J_1$ and $\J_2$ are close to each other, then the optimal costs for $\J_1$ and $\J_2$ are close, and if schedules $S_1$, $S_2$ induce similar but slightly shifted work profiles for $\J_1$ and $\J_2$, then the quality costs for $S_1$ and $S_2$ are close.

We show in Appendix~\ref{app:extension} that for the classically studied energy plus weighted flow time minimization problem with $\alpha\geq 1$, the cost function is $\max(4\max_{j} v_j,2^{\alpha}-1)$-smooth. Note that for energy minimization with deadlines, the objective introduced in Section~\ref{sec:various_objectives} is not smooth for any bounded function $\beta(.)$, since a small shift in the work profiles can induce a large increase in the objective function (in the case we miss a job's hard deadline). However, \cite{BamasMRS20} (Section F.2) shows that it is also possible to transform any prediction-augmented algorithm for the energy plus deadline problem into a shift-tolerant algorithm. \\

\noindent\textbf{Shift tolerance and error definition.} In this extension, we allow each job in the prediction to be perturbed by a small amount.
% of at most $\eta$ times the average cost per job (the intuition here is that for most  objective functions, the average cost per job is at least the average completion time per job). 
Past this tolerance threshold, the perturbed job is treated as a distinct job.  
We assume here that when a job arrives, it is always possible to identify which job of the prediction (if any) it corresponds to. 
More specifically, for each job $j$, we write $(j,r_j, p_j, v_j)$ for the real values of the parameters associated with $j$ and $(j, \hat{r}_j, \hat{p}_j, \hat{v}_j)$ for their predicted values (with the convention that $(j,r_j, p_j, v_j) = \emptyset$  if the job didn't arrive  and $(j, \hat{r}_j, \hat{p}_j, \hat{v}_j) = \emptyset$ if the job was not predicted).

Next, we let $\eta^{\text{shift}}\in [0,1)$ be a shift tolerance parameter, that is initially set by the decision-maker, and we assume that the objective function is $\beta$-smooth for some smooth monotone function $\beta(.)\geq 1$. We now define the set of 'correctly predicted' jobs as $\J^{\text{shift}} = \{(j, r_j, p_j, v_j): |r_j - \hat{r}_j|\leq \sft, |p_j -\hat{p}_j|\leq \frac{\eta^{\text{shift}}}{\beta(\predJ)}\cdot \hat{p}_j, |v_j -\hat{v}_j|\leq \frac{\eta^{\text{shift}}}{\beta(\predJ)} \cdot \hat{v}_j\}$,   which is the set of jobs whose release time, weights and processing times have only been slightly shifted as compared to their predicted values. The amount of shift we tolerate depends on the smoothness function $\beta(.)$ of the objective function $F$ and on the predicted instance $\predJ$. In addition, note that the allowed shift in release time is proportional to the average cost per job (the intuition here is that for most objective functions, the average cost per job is at least the average completion time per job). We underscore the fact that $\J\setminus \J^{\text{shift}}$ contains both the predicted jobs that have past the shift tolerance and additional jobs in the realization.  Finally, we let $\predJ^{\text{shift}} = \{(j, \hat{r}_j, \hat{p}_j, \hat{v}_j): (j, r_j, p_j, v_j)\in \J^{\text{shift}}\}$.  The error $
\eta^g = \frac{1}{\texttt{OPT}(\predJ)}\cdot\max\{\texttt{OPT}(\J\setminus \J^{\text{shift}}), \texttt{OPT}(\predJ \setminus \predJ^{\text{shift}})\}$ is now defined as the optimal cost for both the additional and missing jobs (similarly as in the previous sections) and the jobs that have past the shift tolerance, normalized by the optimal cost for the prediction.

\noindent\textbf{Algorithm description.}
% Note that we cannot directly apply Algorithm~\ref{alg-no-ass} by replacing $\predJ_{\geq \tl}$ by $\J^{\text{shift}}_{\geq \tl}$ on line 6, since: (1) we don't know yet the complete sequence $\J^{\text{shift}}_{\geq \tl}$ at time $\tl$ (2) the schedule $\hat{S} = \textsc{OfflineAlg}(\predJ_{\geq\tl})$ may not be feasible for $\J^{\text{shift}}_{\geq \tl}$.
The Algorithm, called \tpes \ and formally described in Algorithm~\ref{alg-no-ass-shift}, takes the same input parameters as Algorithm \tpe, with some additional shift tolerance parameter $\eta^{\text{shift}}\in [0,1)$, from which we compute the maximum allowed shift in release time $\bar{\eta}$ (Line 8). 

\tpes \ globally follows the structure of \tpe, with a few differences, that we now detail. 
First, we start by slightly increasing the predicted weight and processing time of each job to obtain the set of jobs $\predJ^{\text{up}}: =\{(j, \hat{r}_j, \hat{p}_j(1+\frac{\eta^{\text{shift}}}{\beta(\predJ)}),\hat{v}_j(1+\frac{\eta^{\text{shift}}}{\beta(\predJ)}))\}$ (Line 1). Note that by the first smoothness condition, the optimal schedule for  $\predJ^{\text{up}}$ has only a slightly higher cost than $\predT$. 

Then, similarly as \tpe, \ \tpes \ starts with a first phase where it follows the online algorithm \textsc{OnlineALg} until the time $\tl$ where the optimal offline has reached some threshold value $\lambda \OPT(\predJ^{\text{up}})$ (Lines 3-7). In the second phase (Lines 9-13), it again combines two schedules, this time, for (1) the jobs in $\J^{\text{shift}}_{\geq \tl}$ that are within the shift tolerance (2) the jobs in $\J\setminus \J^{\text{shift}}_{\geq \tl}$, which include the remaining jobs from phase 1 and the non-predicted jobs (or jobs that have past the shift-tolerance) that are released after time $\tl$. To schedule the jobs in $\J^{\text{shift}}_{\geq \tl}$, we first compute an offline schedule $\hat{S}$ for $\predJ^{\text{up}}$ (Line 9). 
% Note that we consider here all predicted jobs released after time $\tl-\bar{\eta}$, since some jobs arriving in the time interval $[\tl-\bar{\eta}, \tl]$ may have been delayed. 
One small difference with \tpe \ is that we will delay the schedule $\hat{S}$ by $\bar{\eta}$ time steps backwards when we schedule jobs in $\J^{\text{shift}}_{\geq \tl}$. More precisely, each job $(j, r_j, p_j, v_j)\in \J^{\text{shift}}_{\geq \tl}$ is scheduled on the same way the job with the same identifier $j$ in $\predJ^{\text{up}}$ is scheduled by $\hat{S}$ at time $t-\bar{\eta}$ (Line 12). The intuition here is that we need to wait a small delay of $\bar{\eta}$ in order to identify which jobs of the predictions indeed arrived.  Finally, and similarly as \tpe, \ the speeds for jobs in $\J\setminus \J^{\text{shift}}_{\geq \tl}$ are set by running $\textsc{OnlineAlg}$ on the set $\J\setminus \J^{\text{shift}}_{\geq \tl}$ (Line 13).

% Second, the two 
% The key is to notice that the schedule $\Ss$ obtained by scheduling each job $(r_j, p_j, v_j)$ on the exact same way $\hat{S}$ schedules $(\hat{r}_j, \hat{p}_j(1+\eta^{\text{shift}}),\hat{v}_j(1+\eta^{\text{shift}}))$ and by shifting   $\sft$ steps backwards the resulting schedule is always feasible for $\J^{\text{shift}}_{\geq \tl}$. Moreover, there is only a small increase in cost as compared to $\hat{S}$.

\begin{algorithm}[H]
\caption{Two-Phase Energy Efficient Scheduling with Shift Tolerance (\textsc{TPE-S})}
\label{alg-no-ass-shift}
% \hspace*{\algorithmicindent} 
{\bf Input:}
predicted and true  sets of jobs $\predJ$ and $\J$, quality of cost function $F$, offline and online algorithms (without predictions) \textsc{OfflineAlg} and  \textsc{OnlineAlg} for problem $F$, confidence level $\lambda \in (0, 1]$,  shift tolerance $\eta^{\text{shift}}>0$.
\begin{algorithmic}[1]
\STATE $\predJ^{\text{up}}\leftarrow \{(j, \hat{r}_j, \hat{p}_j(1+\frac{\eta^{\text{shift}}}{\beta(\predJ)}),\hat{v}_j(1+\frac{\eta^{\text{shift}}}{\beta(\predJ)}))\}$
\STATE $\hat{\texttt{OPT}}
\leftarrow \texttt{OPT}(\predJ^{\text{up}})$
% \STATE \textbf{while} {$\texttt{OPT}(\J_{\leq t}) \leq \lambda  \cdot \hat{\texttt{OPT}}$} \textbf{do}
%\While
% \STATE $\{s_j(t)\}_{j} \leftarrow \textsc{OnlineAlg}(\J_{\leq t})(t)$
%\EndWhile
% \STATE $\tl\leftarrow t$
% \STATE $\tl \leftarrow \max\{t : \texttt{OPT}(\J_{\leq t}) \leq \lambda  \cdot \hat{\texttt{OPT}}\}$
% (\predJ^{\text{up}})\}$
\STATE \textbf{for} {$t \geq 0$} \textbf{do} 
\STATE \quad \textbf{if}  $\texttt{OPT}(\J_{\leq t}) > \lambda  \cdot  \hat{\texttt{OPT}}$ \textbf{then}
\STATE \quad \quad $\tl \leftarrow t$
\STATE \quad \quad \textbf{break}
\STATE \quad  $\{s_j(t)\}_{j\in \J_{\leq t}} \leftarrow \textsc{OnlineAlg}(\J_{\leq t})(t)$

% \State $\J_{\leq t} \leftarrow \{j \in \J : r_j \leq t\}$ 

% \State $\{\hat{s}_j(t)\}_{t\in [0, \tl], j\in \predJ_{\geq\tl}}\leftarrow \{0\}_{t\in [0, \tl], j\in \predJ_{\geq\tl}}$

% \STATE $t' \leftarrow \tl - \sft $
\STATE $\bar{\eta} \leftarrow  \sft $
\STATE $\{\hat{s}_{j}(t)\}_{t\geq 0, j\in \predJ^{\text{up}}}\leftarrow \textsc{OfflineAlg}(\predJ^{\text{up}})$
\STATE \textbf{for} {$t \geq \tl$} \textbf{do}
% \State $\ebedit{\hat{s}, \{\hat{\alpha}_j\}_{j} \leftarrow \hat{s}(t), \{\hat{\alpha}_j(t)\}_{j}}$ 
\STATE \quad\textbf{for} {$j: (j, r_j, p_j, v_j)\in \J^{\text{shift}}_{\geq \tl}$} \textbf{do}
\STATE \quad \quad $s_{j}(t)\leftarrow \hat{s}_j(t-\bar{\eta})$
% \STATE \quad $\{s_{j}(t)\}_{j\in \J^{\text{shift}}_{\geq \tl}}\leftarrow \{\hat{s}_j(t-\bar{\eta})\}$
\STATE \quad $\{s_j(t)\}_{j\in \J_{\leq t}\setminus \J^{\text{shift}}_{\geq \tl}} \hspace{-.12cm} \leftarrow \hspace{-.03cm} \textsc{OnlineAlg}(\J_{\leq t}\setminus \J^{\text{shift}}_{\geq \tl})(t)$

% \STATE \textbf{while} {$\texttt{OPT}(\J_{\leq t}) > \lambda  \cdot \hat{\texttt{OPT}}$} \textbf{do}
% \State $\ebedit{\hat{s}, \{\hat{\alpha}_j\}_{j} \leftarrow \hat{s}(t), \{\hat{\alpha}_j(t)\}_{j}}$
% \STATE \quad $\{s^{on}_j(t)\}_{j\in \J_{\leq t}\setminus \J^{\text{shift}}_{\geq \tl}} \leftarrow \textsc{OnlineAlg}(\J_{\leq t}\setminus \J^{\text{shift}}_{\geq \tl})(t)$
% \STATE \quad $s(t) \leftarrow \hat{s}(t - \sft) + s^{on}(t)$ 
% \STATE \quad \textbf{for all} {$j\in \J_{\leq t}\setminus \J^{\text{shift}}_{\geq \tl}$} \textbf{do}
% \ForAll
% \STATE \quad \quad $s_j(t) = s^{on}_j(t)$
% \EndFor
% \ForAll
% \STATE \quad \textbf{for all} {$j\in \J^{\text{shift}}_{[\tl,t]}$} \textbf{do}
% \STATE \quad \quad $s_j(t) = \hat{s}_{j}\left(t- \sft\right)$
% \EndFor
% \EndWhile
\STATE \textbf{return} $\{s_j(t)\}_{t\geq 0, j\in \J}$
\end{algorithmic}
\end{algorithm}

\noindent\textbf{Analysis.} We now present the analysis of \tpes. All missing proofs are provided in Appendix~\ref{app:extension}. In the following lemma, we start by upper bounding the cost of the schedule output by \tpes \ for the jobs that were released in the second phase and were correctly predicted (i.e., within the shift tolerance). The proof mainly exploits the two smoothness conditions of the cost function.

% We first show that the schedule $\Ss = \left\{\hat{s}_{j}\left( t- \eta\right) \right\}_{t\geq \tl, j\in \Js$  is feasible for $\Js$ and we upper bound its cost.
\begin{restatable}{lem}{lemsshift}
\label{lem:sshift}
Assume that $\cst(.,.)$ is $\beta$-smooth. Consider the schedule $\Ss$, which, for all $t\geq \tl$ and $ (j,r_j, p_j, v_j)\in \Js$, processes job $j$ at speed \[s_j(t) = \hat{s}_{j}\left( t- \sft\right).\] Then,
    \[
\cst(\Ss, \Js)\leq (1+2\eta^{\text{shift}}(1+\eta^{\text{shift}}))\predT.
    \]
\end{restatable}

 We now show some slightly modified version of Corollary~\ref{cor:JinterHatJ-LB}. Similarly as before, we write $\eta_1 = \frac{\texttt{OPT}(\J\setminus \J^{\text{shift}})}{\texttt{OPT}(\predJ)}$ to denote the error corresponding to additional jobs in the prediction, and $\eta_2 = \frac{\texttt{OPT}(\predJ \setminus \predJ^{\text{shift}})}{\texttt{OPT}(\predJ)}$ for the error corresponding to missing jobs.

\begin{restatable}{cor}{corJshift}
\label{cor:JinterHatJ-LB-shift}
Assume that $\cst(.,.)$ is $\beta$-smooth, then
\[\texttt{OPT}(\J^{\text{shift}})\geq \Big[\left(1-\eta_2^{\frac{1}{\alpha}}\right)^{\alpha}-\eta^{\text{shift}}\Big]\Big/(1+ \eta^{\text{shift}})\predT.\]
\end{restatable}

\begin{restatable}{cor}{corlbetashift}
\label{cor:lb_eta2-shift} 
Assume that $\cst(.,.)$ is $\beta$-smooth. If $\texttt{OPT}(\J)\leq \lambda\predT$, then $$\eta_2 \geq \left(1-(  \lambda(1+\eta^{\text{shift}})+ \eta^{\text{shift}})^{\frac{1}{\alpha}}\right)^{\alpha}.$$
\end{restatable}

We now state the main result of this section, which is our upper bound on the competitive ratio of the shift-tolerant Algorithm \tpes.

\begin{restatable}{theorem}{thmUBshift}
    \label{thm:UBshiftm}
Assume that $\cst(.,.)$ is $\beta$-smooth. Then, for any $\lambda\in (0,1], \eta^{\text{shift}}\in [0,1)$, the competitive ratio of 
% Algorithm~\ref{alg-no-ass-shift} 
\tpes \ run with trust parameter $\lambda$,  a $\gamma$-competitive algorithm \textsc{OnlineAlg}, an optimal offline
% a $\gamma_{\text{off}}$-competitive 
algorithm \textsc{OfflineAlg}, and shift tolerance $\eta^{\text{shift}}$ is at most 
$$\begin{cases}
\gamma & \hspace{-1.85cm} \text{if } \texttt{OPT}(\mathcal{J})\leq \lambda \predT \vspace{.2cm} \\
 \frac{\left((1+2\eta^{\text{shift}}(1+\eta^{\text{shift}}))^{\frac{1}{\alpha}} + \gamma^{\frac{1}{\alpha}}\left(\lambda^{\frac{1}{\alpha}} + \eta_1^{\frac{1}{\alpha}}\right)\right)^{\alpha}}{\max\left\{\lambda, \eta_1 + \left(\left(1-\eta_2^{\frac{1}{\alpha}}\right)^{\alpha}-\eta^{\text{shift}}\right)(1+ \eta^{\text{shift}})^{-1}\right\}} & \text{otherwise.}
 \end{cases}$$
\end{restatable}

% \begin{restatable}{theorem}{thmUBshift}
%     \label{thm:UBshift}
% Assume that $\cst(.,.)$ is $\beta$-smooth. Then, for any $\lambda\in (0,1]$, the competitive ratio of \tpes \ run with trust parameter $\lambda$,  a $\gamma_{\text{on}}$-competitive algorithm \textsc{OnlineAlg}, an optimal offline
% % a $\gamma_{\text{off}}$-competitive 
% algorithm \textsc{OfflineAlg}, and shift tolerance $\eta^{\text{shift}}$ is at most 
% $$\begin{cases}
% \gamma & \hspace{-1.85cm} \text{if } \texttt{OPT}(\mathcal{J})\leq \lambda \predT \vspace{.2cm} \\
%  \frac{\left((1+2\eta^{\text{shift}})^{\frac{1}{\alpha}} + \gamma^{\frac{1}{\alpha}}\left(\lambda^{\frac{1}{\alpha}} + \eta_1^{\frac{1}{\alpha}}\right)\right)^{\alpha}}{\max\left\{\lambda, \eta_1 + \left(\left(1-\eta_2^{\frac{1}{\alpha}}\right)^{\alpha}-\eta^{\text{shift}}\right)(1+ \eta^{\text{shift}})^{-1}\right\}} & \text{otherwise.}
%  \end{cases}$$
% \end{restatable}

In particular, we deduce the following consistency and robustness guarantees.

\begin{restatable}{cor}{corconsistencyshift}(consistency)
\label{cor:consistency_s}
For any $\lambda\in (0,1]$ and $\eta^{\text{shift}}\in [0,1)$, if $\eta_1 = \eta_2 = 0$ (all jobs are within the shift tolerance and there is no extra or missing jobs), then the competitive ratio of \tpes \ run with trust parameter $\lambda$ and shift tolerance parameter $\eta^{\text{shift}}$ is upper bounded by $\min(\frac{1}{\lambda},\frac{(1+ \eta^{\text{shift}})}{(1- \eta^{\text{shift}})})\cdot\left((1+2\eta^{\text{shift}}(1+\eta^{\text{shift}}))^{\frac{1}{\alpha}} + \gamma^{\frac{1}{\alpha}}\lambda^{\frac{1}{\alpha}}\right)^{\alpha}\leq \frac{(1+2\eta^{\text{shift}}(1+\eta^{\text{shift}}))^2}{(1- \eta^{\text{shift}})}\cdot (1+\gamma 2^{\alpha}\lambda^{\frac{1}{\alpha}})$.
\end{restatable}

\begin{restatable}{cor}{corrobustnessshift}(robustness)
\label{cor:robustness_s}
For any $\lambda\in (0,1]$ and $\eta^{\text{shift}}\in [0,1)$, the competitive ratio of Algorithm~\ref{alg-no-ass} run with trust parameter $\lambda$ and shift tolerance parameter $\eta^{\text{shift}}$ is upper bounded by $ \frac{(1+2\eta^{\text{shift}}(1+\eta^{\text{shift}}))(1+\gamma 2^{2\alpha}\lambda^{\frac{1}{\alpha}})}{\lambda}$.
\end{restatable}

\section{Missing analysis from Appendix~\ref{sec:app_extension}}
\label{app:extension}

\begin{lemma}
\label{lem:smooth}
    For the objective of minimizing total integral weighted flow time plus energy with $\alpha\geq 1$, the cost function is $\max(4\cdot\max_{j} v_j,2^{\alpha}-1)$-smooth.
\end{lemma}

\begin{proof}
   Let $\J_1,\J_2$ with $|\J_1| = |\J_2|$, a bijection $\pi:\J_1\longrightarrow\J_2$ and $S_1$ and $S_2$ feasible for $\J_1$ and $\J_2$. 
   
   We start with the first smoothness condition. Assume that for some $\eta, \eta'\in [0,1]$ and for all $j\in \J_1$, $|r_j -r_{\pi(j)}|\leq \eta'$, $p_j\leq p_{\pi(j)}(1+ \eta)$ and $v_j\leq v_{\pi(j)}(1+ \eta)$. Let $S^*$ be an optimal schedule for $\J_2$ and consider the schedule $S = \{s_j(t):= (1+\eta)\cdot s^*_{\pi(j)}(t-\eta')\}_{j\in \J_1}$ for $\J_1$.

   Note that for all $j\in \J_1$ and $t\geq 0$, $s^*_{\pi(j)}(t-\eta')>0$ only if $t-\eta'>r_{\pi(j)}$. Since we assumed $|r_j -r_{\pi(j)}|\leq \eta'$, we get that $s_j(t)>0$ only if $t\geq r_j$, hence $S$ is feasible for $\J_1$. Next, note that 
   \[
E(S) = E(S^*)(1+\eta)^{\alpha}\leq E(S^*)(1+(2^{\alpha} -1 )\eta).
   \]
   Recall that $c_{j}^S$ denotes the completion time of $j$ by $S$. Since $p_j\leq (1+\eta)p_{\pi(j)}$, and since we assumed that $|r_j -r_{\pi(j)}|\leq \eta'$, we have,  by definition of $S$, that for all $j$, $c_j^S\leq \eta' + c_{\pi(j)}^{S^*}$.
   
   Hence, 
   \begin{align*}
       F(S,\J_1) &= \sum_{j\in \J_1} v_j(c_j^S - r_j)\\
       &\leq \sum_{j\in \J_1} v_{\pi(j)}(1+\eta)(c_{\pi(j)}^{S^*} +\eta' - (r_{\pi(j)} - \eta')) \\
       &= (1+\eta)F(S^*, \J_2) + 2\max_j v_j|\J_1|(1+\eta)\eta'\\
       &\leq (1+\eta)F(S^*, \J_2) + 4\max_j v_j|\J_1|\eta' &\eta\in [0,1]\\
       &\leq (1+(2^{\alpha} -1)\eta)F(S^*, \J_2) + 4\max_j v_j|\J_1|\eta' &\alpha\geq 1.
   \end{align*}

Therefore,
\begin{align*}
    \OPT(\J_1)&\leq \cst(S, \J_1) \\
    &= E(S) + F(S,\J_1) \\
    &\leq (E(S^*) + F(S^*,\J_2))\cdot(1+(2^{\alpha} -1 )\eta) + 4\max_j v_j|\J_1|\eta'\\
    &= \cst(S^*, \J_2)(1+(2^{\alpha} -1 )\eta)+ 4\max_j v_j|\J_1|\eta'\\
    &= \OPT(\J_2)(1+(2^{\alpha} -1 )\eta)+ 4\max_j v_j|\J_1|\eta'.
\end{align*}

We now show the second smoothness condition. Assume that for all $j\in \J_1$, $p_j \leq p_{\pi(j)}$, $v_j \leq v_{\pi(j)}$, $r_j \geq r_{\pi(j)} - \eta'$ and that for all $t\geq r_{\pi(j)}+\eta'$, $w_{S_1}^j(t) \leq w_{S_2}^{\pi(j)}(t- \eta')$. Then, in particular $w_{S_1}^j(c_{\pi(j)}^{S_2}+\eta') \leq w_{S_2}^{\pi(j)}(c_{\pi(j)}^{S_2}) = 0$, hence  $c_j^{S_1} = \min\{t\geq r_j: w_{S_1}^j(t) = 0\} \leq c_{\pi(j)}^{S_2}+\eta'$. By a similar argument as above, we conclude that:
\[
F(S,\J_1) \leq F(S_2, \J_2) + 4\max_j v_j|\J_1|\eta'.
\]

\end{proof}

\lemsshift*

\begin{proof}

For simplifying the exposition, in the remainder of the proof, we write $\bar{\eta}$ instead of $\sft$.

We first analyse the energy cost.

\allowdisplaybreaks{
\begin{align}
    E(\Ss) = \int_{t\geq \tl} \left(\sum_{j\in \J^{\text{shift}}_{\geq \tl}}\hat{s}_{j}\left( t- \bar{\eta}\right)\right)^{\alpha}\text{dt}
    &= \int_{t\geq \tl-\bar{\eta}} \left(\sum_{j\in \J^{\text{shift}}_{\geq \tl}}\hat{s}_{j}\left( t\right)\right)^{\alpha}\text{dt} \nonumber\\
    &\leq \int_{t} \left(\sum_{j\in \predJ^{\text{up}}}\hat{s}_{j}\left( t\right)\right)^{\alpha}\text{dt}
     = E(\hat{S}).\label{eq:shiftE}
\end{align}}

Next, we analyze the quality cost. Note that by definition of $\Ss$, we have that for all $ j \in \Js$ and $t\geq \hat{r}_j + \bar{\eta}$ , the same amount of work for $(j, r_j, p_j,v_j)$ has been processed by $\Ss$ at time $t$ as the amount of work for $(j, \hat{r}_j, \hat{p}_j\big(1+\frac{\eta^{\text{shift}}}{\beta(\predJ)}\big),\hat{v}_j\big(1+\frac{\eta^{\text{shift}}}{\beta(\predJ)}\big))$ processed by $\hat{S}$ at time $t- \bar{\eta}$.
Hence, for all $t\geq \hat{r}_j + \bar{\eta}$,
\[
w_{\Ss}^{(j, r_j, p_j,v_j)}(t) = w_{\hat{S}}^{(j, \hat{r}_j, \hat{p}_j\big(1+\frac{\eta^{\text{shift}}}{\beta(\predJ)}\big),\hat{v}_j\big(1+\frac{\eta^{\text{shift}}}{\beta(\predJ)}\big))}\left(t- \bar{\eta}\right).
\]

By definition of $\Js$, we also have that for all $(j,r_j,p_j, v_j)\in \Js$, $|r_j-\hat{r}_j|\leq \bar{\eta}$, $p_j \leq \hat{p}_j\big(1+\frac{\eta^{\text{shift}}}{\beta(\predJ)}\big)$ and $v_j \leq \hat{v}_j\big(1+\frac{\eta^{\text{shift}}}{\beta(\predJ)}\big)$. Hence, we can apply the second smoothness condition with $\J_1 = \Js$ and $\J_2 = \{(j, \hat{r}_j, \hat{p}_j\big(1+\frac{\eta^{\text{shift}}}{\beta(\predJ)}\big),\hat{v}_j\big(1+\frac{\eta^{\text{shift}}}{\beta(\predJ)}\big)): (j, r_j,p_j,v_j)\in \Js\}\subseteq \predJ^{\text{up}}$. This gives:

\begin{align}
F(\Ss, \Js) &\leq F(\hat{S}, \J_2) + \bar{\eta}\beta(\Js)|\Js|\nonumber\\
    &\leq F(\hat{S}, \predJ^{\text{up}}) + \bar{\eta}\beta(\Js)|\Js|\nonumber\\
    &= F(\hat{S}, \predJ^{\text{up}}) + \sft \cdot \beta(\Js)|\Js|\nonumber\\
     &\leq  F(\hat{S}, \predJ^{\text{up}}) + \sft \cdot \beta(\predJ)\left(1+\frac{\eta^{\text{shift}}}{\beta(\predJ)}\right)|\predJ|\nonumber\\
    &= F(\hat{S}, \predJ^{\text{up}}) + \eta^{\text{shift}}\left(1+\frac{\eta^{\text{shift}}}{\beta(\predJ)}\right)\OPT(\predJ)\nonumber\\
    &\leq F(\hat{S}, \predJ^{\text{up}}) + \eta^{\text{shift}}(1+\eta^{\text{shift}})\OPT(\predJ),\label{eq:shiftF}
\end{align}

where the first inequality is by the second smoothness condition, the second one is is by monotonicity of $F$, the equality is by definition of $\bar{\eta}$, and the third inequality is by the smoothness and monotonicity of $\beta$. The last inequality is since $\beta\geq 1$.

Therefore, we get
\begin{align*}
    \cst(\Ss, \Js) &= E(\Ss)+ F(\Ss, \Js)\nonumber\\
    &\leq E(\hat{S}) + F(\hat{S}, \predJ^{\text{up}}) + \eta^{\text{shift}}(1+\eta^{\text{shift}})\OPT(\predJ)\\
    &=  \OPT(\predJ^{\text{up}}) + \eta^{\text{shift}}(1+\eta^{\text{shift}})\predT\\
    &\leq  \left(1+ \beta(\predJ^{\text{up}})\frac{\eta^{\text{shift}}}{\beta(\predJ)}\right)\OPT(\predJ) + \eta^{\text{shift}}(1+\eta^{\text{shift}})\predT\\
&\leq  \left(1+ \beta(\predJ)\left(1+\frac{\eta^{\text{shift}}}{\beta(\predJ)}\right)\frac{\eta^{\text{shift}}}{\beta(\predJ)}\right)\OPT(\predJ) + \eta^{\text{shift}}(1+\eta^{\text{shift}})\predT\\
    &\leq (1 + \eta^{\text{shift}}(1+\eta^{\text{shift}}))\predT +  \eta^{\text{shift}}(1+\eta^{\text{shift}})\predT\\
    &=(1+2\eta^{\text{shift}}(1+\eta^{\text{shift}}))\predT,
\end{align*}

where the first inequality is by (\ref{eq:shiftE}) and (\ref{eq:shiftF}), the second inequality is by the first smoothness condition with $\J_1 = \predJ^{\text{up}}$ and $\J_2 = \predJ$, the third inequality is by smoothness of $\beta$ and the last inequality since $\beta\geq 1$.

\end{proof}

\corJshift*

\begin{proof}
We prove the result by contradiction. Assume that $\texttt{OPT}(\J^{\text{shift}})< \Big[\left(1-\eta_2^{\frac{1}{\alpha}}\right)^{\alpha}-\eta^{\text{shift}}\Big]\Big/(1+ \eta^{\text{shift}})\predT$. Then, we have:
\begin{align*}
    \texttt{OPT}(\predJ^{\text{shift}})&\leq \left(1+ \beta(\predJ^{\text{shift}})\cdot\frac{\eta^{\text{shift}}}{\beta(\predJ)}\right)\texttt{OPT}(\J^{\text{shift}}) + \beta(\predJ^{\text{shift}})|\predJ^{\text{shift}}|\cdot \sft\\
    &\leq \left(1+ \beta(\predJ)\cdot\frac{\eta^{\text{shift}}}{\beta(\predJ)}\right)\texttt{OPT}(\J^{\text{shift}}) + \beta(\predJ)|\predJ^{\text{shift}}|\cdot \sft\\
    &\leq (1+\eta^{\text{shift}})\texttt{OPT}(\J^{\text{shift}}) + \eta^{\text{shift}}\predT\\
    &< \left(1-\eta_2^{\frac{1}{\alpha}}\right)^{\alpha}\predT.
\end{align*}
where the first inequality is by the first smoothness condition with $\J_1 = \predJ^{\text{shift}}$, $\J_2 = \J^{\text{shift}}$, and the second one is by monotonicity of $\beta$.

Next, by definition of the error $\eta_2$, we have $\texttt{OPT}(\predJ\setminus \predJ^{\text{shift}}) =  \eta_2\cdot\predT$. Hence, by Lemma~\ref{lem:union-of-two-schedule}, there exists a schedule $S$ for $\predJ^{\text{shift}} \cup (\predJ\setminus\predJ^{\text{shift}}) = \predJ$ such that
\begin{align*}
   \cst(S, \predJ) 
    &\le \left(\texttt{OPT}(\predJ\setminus \predJ^{\text{shift}})^{\frac{1}{\alpha}} + \texttt{OPT}(\predJ^{\text{shift}})^{\frac{1}{\alpha}}\right)^{\alpha}\\
     &< \left((\eta_2\predT)^{\frac{1}{\alpha}} + \left(\left(1-\eta_2^{\frac{1}{\alpha}}\right)^{\alpha}\predT\right)^{\frac{1}{\alpha}}\right)^{\alpha}\\
    & = \predT,
\end{align*}

    which contradicts the definition of $\predT$.
\end{proof}

\corlbetashift*
\begin{proof}
 Assume that $\texttt{OPT}(\J)\leq \lambda\predT$. Since $\J^{\text{shift}}\subseteq \J$, we get $\texttt{OPT}(\J^{\text{shift}})\leq \lambda\predT$. Hence, we have:
\begin{align*}
    \texttt{OPT}(\predJ^{\text{shift}})&\leq \left(1+ \beta(\predJ^{\text{shift}})\cdot\frac{\eta^{\text{shift}}}{\beta(\predJ)}\right)\texttt{OPT}(\J^{\text{shift}}) + \beta(\predJ^{\text{shift}})|\predJ^{\text{shift}}|\cdot \sft\\
    &\leq \left(1+ \beta(\predJ)\cdot\frac{\eta^{\text{shift}}}{\beta(\predJ)}\right)\texttt{OPT}(\J^{\text{shift}}) + \beta(\predJ)|\predJ^{\text{shift}}|\cdot \sft\\
    &\leq (1+\eta^{\text{shift}})\texttt{OPT}(\J^{\text{shift}}) + \eta^{\text{shift}}\predT\\
    &\leq (  \lambda(1+\eta^{\text{shift}})+ \eta^{\text{shift}})\predT.
\end{align*}
where the first inequality is by the first smoothness condition with $\J_1 = \predJ^{\text{shift}}$, $\J_2 = \J^{\text{shift}}$, and the second one is by monotonicity of $\beta$.

 Next, assume by contradiction, that $\eta_2 <  \left(1-(  \lambda(1+\eta^{\text{shift}})+ \eta^{\text{shift}})^{\frac{1}{\alpha}}\right)^{\alpha}$, which implies that $\texttt{OPT}(\predJ\setminus \predJ^{\text{shift}})<  \left(1-(  \lambda(1+\eta^{\text{shift}})+ \eta^{\text{shift}})^{\frac{1}{\alpha}}\right)^{\alpha}\predT$. 
 
 Then, by Lemma~\ref{lem:union-of-two-schedule}, there exists a schedule $S$ for $\predJ^{\text{shift}} \cup (\predJ\setminus \predJ^{\text{shift}}) = \predJ$ such that
\begin{align*}
   &\cst(S, \predJ) \\
    &\le \left(\texttt{OPT}(\predJ\setminus \predJ^{\text{shift}})^{\frac{1}{\alpha}} + \texttt{OPT}(\predJ^{\text{shift}})^{\frac{1}{\alpha}}\right)^{\alpha}\\
     &< \left(((  \lambda(1+\eta^{\text{shift}})+ \eta^{\text{shift}})\predT)^{\frac{1}{\alpha}} + \left(\left(1-(  \lambda(1+\eta^{\text{shift}})+ \eta^{\text{shift}})^{\frac{1}{\alpha}}\right)^{\alpha}\predT\right)^{\frac{1}{\alpha}}\right)^{\alpha}\\
    & = \predT,
\end{align*}

    which contradicts the definition of $\predT$. Hence, $\eta_2 \geq \left(1-(  \lambda(1+\eta^{\text{shift}})+ \eta^{\text{shift}})^{\frac{1}{\alpha}}\right)^{\alpha}$.
\end{proof}

\thmUBshift*

\begin{proof}

Similarly as in the proof of Theorem~\ref{thm:competitive_ratio}, we have that if $\texttt{OPT}(\mathcal{J})\leq \lambda \predT$, then 
\begin{equation}
\label{eq:ub_cost1-shift}
    \cst(S,\J)\leq \gamma\cdot\trueT.
\end{equation} 
% and, by Corollary~\ref{cor:lb_eta2-shift}, that this only happens if $\eta_2 \geq \left(1-(  \lambda(1+\eta^{\text{shift}})+ \eta^{\text{shift}})^{\frac{1}{\alpha}}\right)^{\alpha}$.

Next, we assume that there is $\tl\geq 0$ such that $\texttt{OPT}(\mathcal{J}_{\leq \tl})>\lambda \predT$. Hence we immediately have:
\[
\trueT \geq \texttt{OPT}(\mathcal{J}_{\leq \tl})> \lambda \predT.
\]

By Corollary~\ref{cor:JinterHatJ-LB-shift} and by definition of the error $\eta_1$, we also get the following lower bound on the optimal schedule:
    \begin{align*}
        \trueT &\geq \texttt{OPT}(\mathcal{J}\setminus \J^{\text{shift}}) + \texttt{OPT}(\J^{\text{shift}})\\
        &\geq \eta_1\predT + \Big[\left(1-\eta_2^{\frac{1}{\alpha}}\right)^{\alpha}-\eta^{\text{shift}}\Big]\Big/(1+ \eta^{\text{shift}})\predT.
    \end{align*}
    
Therefore, 
\begin{equation}
    \label{eq:lb_opt-shift}
    \trueT\geq \max\left\{\lambda, \eta_1 + \Big[\left(1-\eta_2^{\frac{1}{\alpha}}\right)^{\alpha}-\eta^{\text{shift}}\Big]\Big/(1+ \eta^{\text{shift}})\right\}\predT.
\end{equation}

We now upper bound the cost of the schedule output by our algorithm. By the same argument as in the proof of Lemma~\ref{lem:upper_bound_schedule}, we get:
\[
\cst(S^{on}, \J\setminus\Js)  \leq  \gamma\cdot \predT\left(\lambda^{\frac{1}{\alpha}} + \eta_1^{\frac{1}{\alpha}}\right)^{\alpha}.
\]

Now, from Lemma~\ref{lem:sshift}, we have 

\[
\cst(\Ss, \Js)\leq (1+2\eta^{\text{shift}}(1+\eta^{\text{shift}}))\predT.
\]

Therefore, by applying Lemma~\ref{lem:union-of-two-schedule}, we get:
\begin{align}
    \cst(S, \J) &\leq \left(\cst(\Ss, \Js)^{\frac{1}{\alpha}} + \cst(S^{on}, \J\setminus\Ss)^{\frac{1}{\alpha}}\right)^{\alpha}\nonumber\\
    &\leq \predT\left((1+2\eta^{\text{shift}}(1+\eta^{\text{shift}}))^{\frac{1}{\alpha}} + \gamma^{\frac{1}{\alpha}}(\lambda^{\frac{1}{\alpha}} + \eta_1^{\frac{1}{\alpha}})\right)^{\alpha}\label{eq:ub_cost-shift}.
\end{align}

    Hence, we get the following upper bound on the competitive ratio of Algorithm~\ref{alg-no-ass-shift}:
    \begin{align*}
    &\;\;\;\frac{\cst(S, \J)}{\trueT}\\
    &= \mathbf{1}_{\texttt{OPT}(\mathcal{J})\leq \lambda \predT}\frac{\cst(S, \J)}{\trueT} + \mathbf{1}_{\texttt{OPT}(\mathcal{J})> \lambda \predT}\frac{\cst(S, \J)}{\trueT}
    \\
    % &= \mathbf{1}_{\texttt{OPT}(\mathcal{J})\leq \lambda \predT}\mathbf{1}_{\eta_2 \geq \left(1-(  \lambda(1+\eta^{\text{shift}})+ \eta^{\text{shift}})^{\frac{1}{\alpha}}\right)^{\alpha}}\frac{\cst(S, \J)}{\trueT} + \mathbf{1}_{\texttt{OPT}(\mathcal{J})> \lambda \predT}\frac{\cst(S, \J)}{\trueT}\\
    &\le \mathbf{1}_{\texttt{OPT}(\mathcal{J})\leq \lambda \predT}\cdot \gamma +  \mathbf{1}_{\texttt{OPT}(\mathcal{J})> \lambda \predT}\frac{\left((1+2\eta^{\text{shift}}(1+\eta^{\text{shift}}))^{\frac{1}{\alpha}} + \gamma^{\frac{1}{\alpha}}(\lambda^{\frac{1}{\alpha}} + \eta_1^{\frac{1}{\alpha}})\right)^{\alpha}}{\max\left\{\lambda, \eta_1 + \Big[\left(1-\eta_2^{\frac{1}{\alpha}}\right)^{\alpha}-\eta^{\text{shift}}\Big]\Big/(1+ \eta^{\text{shift}})\right\}}.
    \end{align*}

\end{proof}
%%%%%%%%%%%%%%%%%%
\section{Additional Experiments}
\label{sec:appexperiments}

  \begin{figure*}[t!]
 	\centering
    \subfigure{\includegraphics[width=0.23\textwidth]{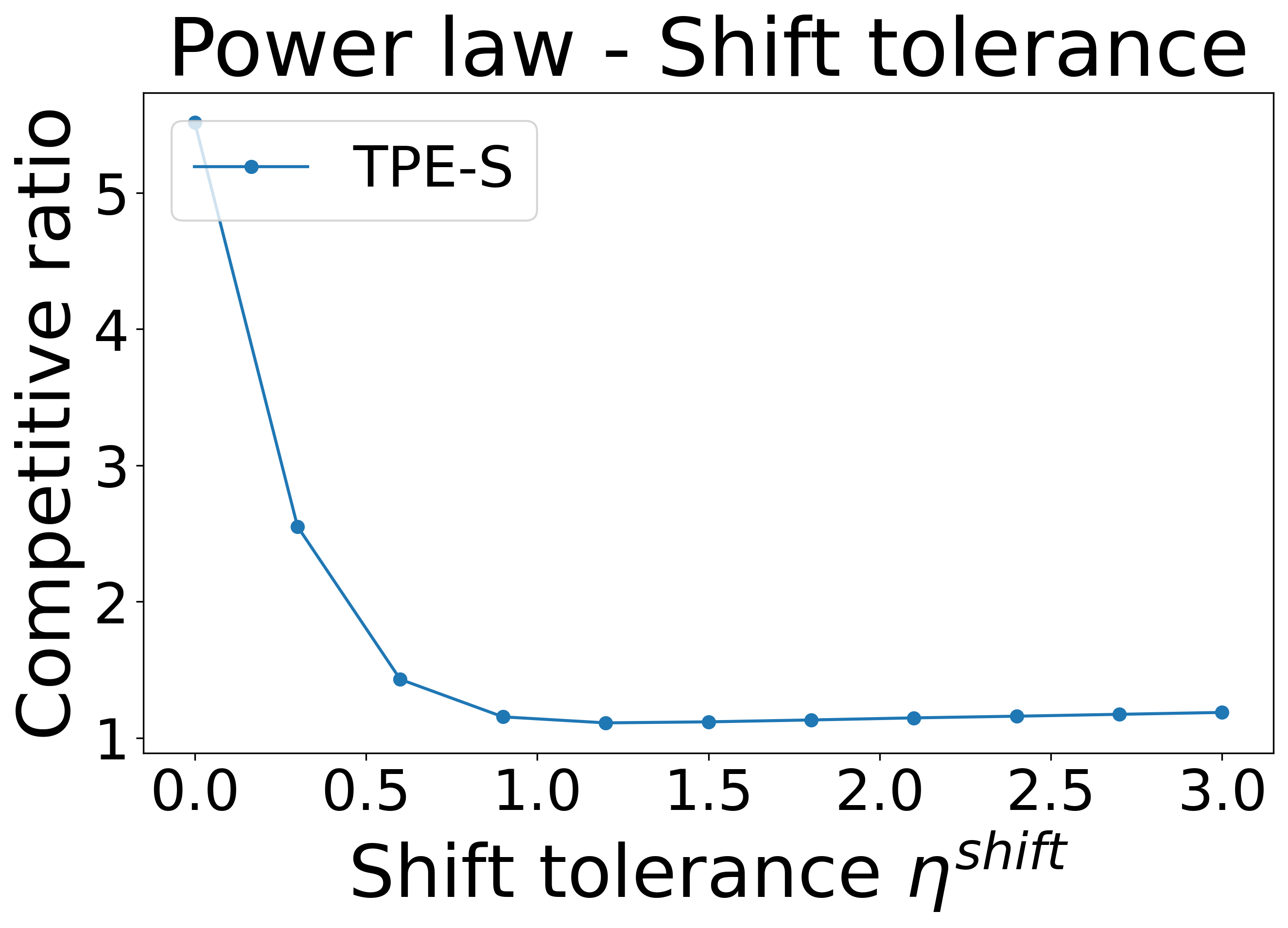}} 
    \subfigure{\includegraphics[width=0.23\textwidth]{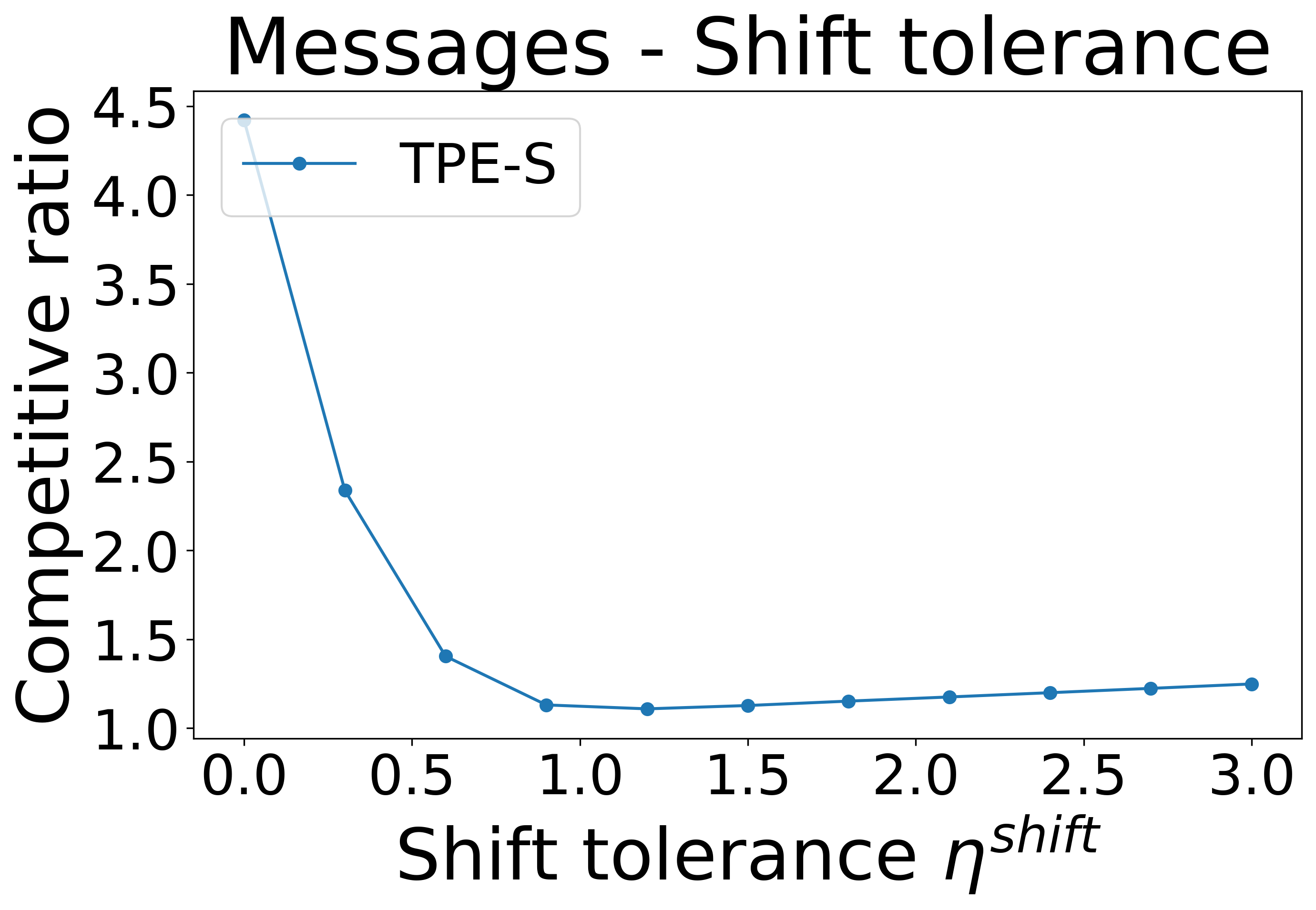}}
    \subfigure{\includegraphics[width=0.23\textwidth]{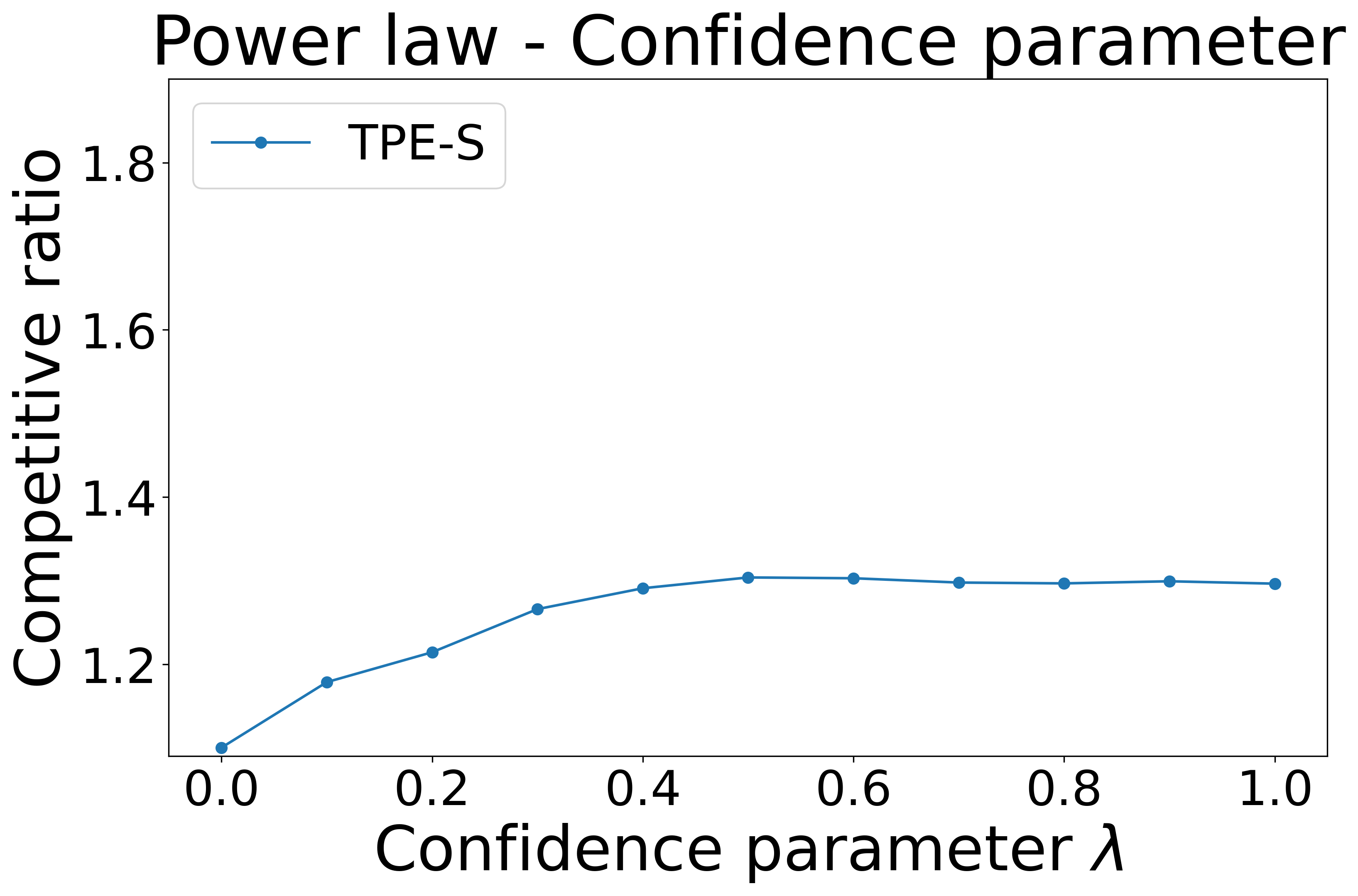}} 
    \subfigure{\includegraphics[width=0.23\textwidth]{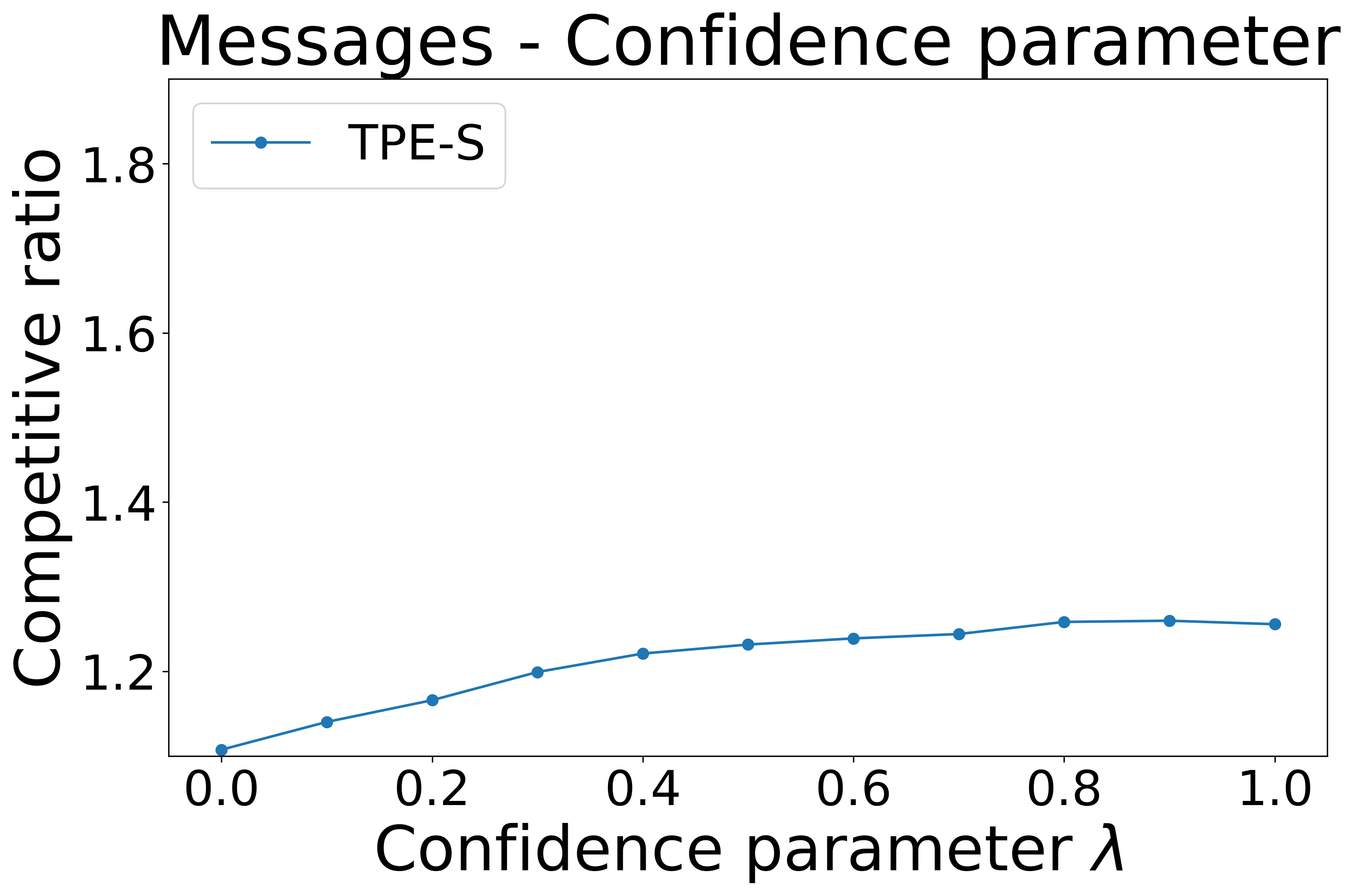}}
 	\caption{The competitive ratio achieved by our algorithm, \tpes \, and the benchmark algorithm as a function of the shift tolerance $\eta^{shift}$ (row 1) and as a function of the confidence parameter $\lambda$ (row 2).}
 	\label{fig:exp:add}
 \end{figure*}
Here, we also evaluate the impact of setting the parameters $\eta^{shift}$ and $\lambda$ on the two other datasets (power law and real datasets). The result are presented in Figure~\ref{fig:exp:add}. We observe similar behaviors as for the periodic dataset.

\section{Comparison with \cite{BamasMRS20, antoniadis2021novel}}
\label{appendix_discussion}

In \cite{BamasMRS20, antoniadis2021novel}, the authors consider the energy minimization problem with deadlines, which, as detailed in Section~\ref{sec:various_objectives}, is a special case of our general framework. For this problem, they propose two different learning-augmented algorithms. We present here some elements of comparison with our algorithm for (GESP). We first show in Section~\ref{appe:discussion_model} and~\ref{app:CR_comparison} that our prediction model and results generalize the ones in \cite{BamasMRS20}: they are similar  in the case of \textit{uniform} deadlines and generalize the ones in \cite{BamasMRS20} for \textit{general} deadlines. We also note that they are incomparable to those in \cite{antoniadis2021novel}. In Section~\ref{app:techniques_discussion}, we then discuss the algorithmic differences with \cite{BamasMRS20} for the special case of energy with uniform deadlines.

\subsection{Discussion about the prediction and error model in comparison to \cite{BamasMRS20, antoniadis2021novel}}
\label{appe:discussion_model}

% In this section, we show that our prediction model directly generalizes the one in \cite{BamasMRS20} for the energy with deadlines problem, and is incomparable to the one in \cite{antoniadis2021novel}.

At a high level, the prediction models considered in \cite{BamasMRS20} and \cite{antoniadis2021novel} are qualified by \cite{antoniadis2021novel} as 'orthogonal'. In \cite{antoniadis2021novel}, the number of jobs is known in advance, as well as the exact processing time for each job, however, the release time and deadlines are only revealed when a job arrives, and the error is proportional to the maximal shift in these values. On the contrary, in \cite{BamasMRS20}, the release times and deadlines are known in advance, and the prediction regards the total workload at each time step. The error is then defined as a function of the total variation of workload, which is the analog of additional and missing jobs in our setting. Note that in the model in \cite{antoniadis2021novel}, the predicted and true set of jobs need to contain exactly the same number of jobs, whereas the model in \cite{BamasMRS20} and our model allow for extra or missing jobs.

% In particular, we show the following:
% \begin{itemize}
%     \item  we propose a prediction model and error metrics that generalize and design an algorithm that is both robust, consistent and smooth, whereas \cite{BamasMRS20} do not introduce any error
%     \item For the special case of energy with \textit{uniform} deadlines problem, our prediction model, error metric and results are comparable with the ones in \cite{BamasMRS20}.
%     \item \npedit{We argue that for some other objective functions...}
% \end{itemize}
% We then discuss the technical differences with \cite{BamasMRS20}.

\paragraph{Comparison with the prediction model and the error metrics in \cite{BamasMRS20} for energy minimization with uniform deadlines.} 

Note that the prediction model used in  \cite{BamasMRS20} for the energy minimization with deadlines problem is slightly different than ours: the prediction is the total workload $w_i^{\text{pred}}$ that arrives at each time step $i$ and needs to be scheduled before time $i+D$, and the error metric is defined as 
\begin{equation}
\label{eq:bamas_metric}
   \text{err}(w^{\text{real}}, w^{\text{pred}}) := \sum_i ||w_i^{\text{real}} - w_i^{\text{pred}}||^{\alpha},
\end{equation}

where $w^{\text{real}}$ denotes the real workload at each time step.

However, in the specific case of energy minimization under \textit{uniform} deadline constraints, our prediction model and error metric and the ones from \cite{BamasMRS20} are comparable: a workload $w^{\text{real}}_i$ that arrives at time $i$ is equivalent in our setting to receiving $w^{\text{real}}_i$ unit jobs with release time $r=i$ and a common deadline $d = i + D$. Moreover, we prove the following lemma, which shows that a small error in the sense of \cite{BamasMRS20} induces a small error $\eta(\J, \predJ)$ in the sense defined in Section~\ref{sec:preliminaries}.

\begin{lem}
For any constant $D>0$, and any instance $(\J, \predJ)$, where at each time $i$, $\J$ is composed of $w_i^{\text{real}}$ jobs of one time unit with deadline $i+D$ and $\predJ$ is composed of $w_i^{\text{pred}}$ jobs of one time unit with deadline $i+D$, we have:
\[
\eta(\J, \predJ)\cdot \texttt{OPT}(\predJ)=  \max\{\texttt{OPT}(\J\setminus \predJ), \texttt{OPT}(\predJ \setminus \J)\}\leq D \cdot\text{err}(w^{\text{real}}, w^{\text{pred}}).
\]
\end{lem}

\begin{proof}
For convenience, we write $\Delta_i = ||w_i^{\text{real}} - w_i^{\text{pred}}||$. We can then write $(\J\setminus \predJ)\cup (\predJ \setminus \J)$ as the instance which, at each time step $i$, is composed of $\Delta_i$ unit size jobs with a common deadline $i+D$.

We now upper bound the optimal cost for $(\J\setminus \predJ)\cup (\predJ \setminus \J)$ by the cost obtained by the Average Rate heuristic (AVR) (first introduced in \cite{yao1995scheduling}). For each $j$, The AVR algorithm schedules uniformly the $\Delta_j$ units of work arriving at time $j$ over the next $D$ time steps. This is equivalent to setting the speed $s_j(t)$ for each workload $\Delta_j$ at time $t \in [ j,\ldots, j+D]$ to $\frac{\Delta_j}{D}$ and set $s_j(t) = 0$ everywhere else. For all $t\geq 0$, the machine then runs at total speed $\sum_{j} s_j(t)$.

Letting $E_{\text{AVR}}$ denote the total cost of the AVR heuristic, we get:

\allowdisplaybreaks{
\begin{align*}
     \max\{\texttt{OPT}(\J\setminus \predJ), \texttt{OPT}(\predJ \setminus \J)\}&
     \leq \OPT((\J\setminus \predJ)\cup (\predJ \setminus \J)) \\
     &\leq E_{\text{AVR}}((\J\setminus \predJ)\cup (\predJ \setminus \J))\\
     &= \sum_{t=1}^\infty \left(\sum_{j} \mathbf{1}_{s_j(t)\neq 0}s_j(t)\right)^{\alpha}\\
     & \leq \sum_{t=1}^\infty |\{j:s_j(t)\neq 0\}|^{\alpha}\left(\max_{j:s_j(t)\neq 0 }s_j(t)\right)^{\alpha}\\
      & \leq \sum_{t=1}^\infty D^{\alpha}\left(\max_{j:s_j(t)\neq 0 }s_j(t)\right)^{\alpha}\\
      & \leq \sum_{t=1}^\infty D^{\alpha}\sum_{j:s_j(t)\neq 0 }s_j(t)^{\alpha}\\
      & = \sum_{j} D^{\alpha}s_j(t)^{\alpha}\sum_t \mathbf{1}_{s_j(t)\neq 0}\\
      & \leq \sum_{j} D^{\alpha}s_j(t)^{\alpha}D\\
      &= \sum_{j}\Delta_j^{\alpha} D\\
      &= D \cdot\text{err}(w^{\text{real}}, w^{\text{pred}}),
\end{align*}
}

where the fourth and sixth inequalities are since by definition of the AVR algorithm, each workload $\Delta_j>0$ has only positive speed on time steps $[j,\ldots, j+D]$.
\end{proof}

\paragraph{Comparison of the error metrics for general objective functions.}  We illustrate here that   for a more general GESP problem, the error metric we define can be tighter than the one in \cite{BamasMRS20} (in the sense that there are instances $(\J, \predJ)$ and quality cost functions $F$ such that $\eta(\J, \predJ) << \text{err}(w^{\text{real}}, w^{\text{pred}})$) and that it may better adapt to the specific cost function under consideration. 

To illustrate this point, consider an instance where the prediction is the realization plus an additional workload of $k$ jobs that all arrive at time $0$, and consider the objective of minimizing total energy plus flow time. In this case, the error computed in (\ref{eq:bamas_metric}) is $k^{\alpha}$, whereas the error $\eta(\J, \predJ)$ we define is the optimal cost for the $k$ extra jobs. By using results from \cite{andrew2009optimal}, this is equal to $k^{\frac{2\alpha -1}{\alpha}}$ ($<< k^{\alpha}$ when $\alpha$ grows large). Hence our error metric is tighter in this case.

\subsection{Comparison with the theoretical guarantees in \cite{BamasMRS20}}
\label{app:CR_comparison}

 % Although our algorithm (Algorithm~\ref{alg-no-ass}) employs different  techniques than the algorithm LAS proposed in \cite{BamasMRS20}, we show below that we achieve similar 
 We compare below the theoretical guarantees in Theorem~\ref{thm:competitive_ratio}  and the ones shown in \cite{BamasMRS20} for the specific problem of energy minimization with \textit{uniform} deadlines. We note that we also generalize these results to the case of \textit{general} deadlines to obtain the first guarantee that smoothly degrades as a function of the prediction error in that setting. Note that for general deadlines, \cite{BamasMRS20} only obtain consistency and robustness, but not smoothness. 

\paragraph{Comparison in the case of uniform deadlines.} For convenience of the reader, we first recall below the guarantee proven in \cite{BamasMRS20}.

\begin{theorem}[Theorem 8 in \cite{BamasMRS20}]
For any given $\epsilon > 0$, algorithm LAS constructs, for the energy minimization with deadlines problem, a schedule of cost at most $\min\{(1 + \epsilon) \OPT + O((\frac{\alpha}{\epsilon})^{\alpha})\text{err}, O((\frac{\alpha}{\epsilon})^{\alpha})\OPT\}$, where 
\[
\text{err}(w^{\text{real}}, w^{\text{pred}}) := \sum_i ||w_i^{\text{real}} - w_i^{\text{pred}}||^{\alpha},
\]
\end{theorem}

which is a similar dependency in $\epsilon$ as the one proved in Theorem~\ref{thm:competitive_ratio}. In particular, for all $\epsilon>0$, Algorithm LAS achieves a consistency of $(1+\epsilon)$ for a robustness factor of $O((\frac{\alpha}{\epsilon})^{\alpha})$. On the other hand, when running Algorithm~\ref{alg-no-ass} with parameter $\lambda = (\frac{\epsilon}{C2^{\alpha}})^{\alpha}$ and the \textsc{Average Rate} heuristic \cite{yao1995scheduling} as \textsc{OnlineAlG} (which was proven to have a $2^{\alpha}$ competitive ratio in \cite{BamasMRS20}), we obtain, by plugging $\lambda = (\frac{\epsilon}{C2^{\alpha}})^{\alpha}$ in the bounds provided in Corollary~\ref{cor:consistency}, a consistency of $(1+\epsilon)$ for a robustness factor of $O(\frac{4^{\alpha^2}}{\epsilon^{\alpha-1}})$.

\subsection{Comparison with the algorithm (LAS) in \cite{BamasMRS20}}
\label{app:techniques_discussion}

In this section, we discuss the technical differences with the algorithm (LAS) proposed in \cite{BamasMRS20} for the energy with deadlines problem. We first note that \cite{BamasMRS20} only shows smoothness, consistency and robustness in the uniform deadline case, where all jobs must be completed within $D$ time steps from their release time. For the general deadline case, \cite{BamasMRS20} presents a more complicated algorithm and only show consistency and robustness. The authors note that "one can also define smooth algorithms for general deadlines as [they] did in the
uniform case. However, the prediction model and the measure of error quickly get complex and
notation heavy". On the contrary, Algorithm~\ref{alg-no-ass} remains simple, captures the general deadline case, and is also endowed with smoothness guarantees. We now discuss more specifically the technical differences. 

\paragraph{Robustification technique.} \cite{BamasMRS20} uses a convolution technique for the uniform deadline case, and a more complicated procedure that separates each interval into a base part and an auxiliary part for the general deadline case. On the other hand, our robustification technique is based on a simpler two-phase algorithm.

We now give some intuition about why a direct generalisation of the techniques in \cite{BamasMRS20}  to general objective functions does not seem straightforward. The main technical difficulty is that in \cite{BamasMRS20}, each job $j$ must be completed before its deadline $d_j$, which is revealed to  the decision maker at the time the job arrives and is used by the algorithm. For a general objective function, we do not have a deadline ; however, one could think about using the total completion time $c_j$ of each job instead. The issue is that $c_j$ may depend on all future job arrivals and is not known at the time the job arrives, hence it cannot be used directly by the algorithm. 

To illustrate this point, consider the objective of minimizing total energy plus flow time with $\alpha =2$. Consider two instances, where the first one has $1$ job arriving at time $0$ and the second one has $1$ job arriving at time $0$ and $n-1$ jobs arriving at time $\frac{1}{\sqrt{n}}$. In the first case, the optimal is to complete the first job in $1$ unit of time, whereas it is completed in $\frac{1}{\sqrt{n}}$ unit of time in the second case. Furthermore, this can be only deduced  after the $n-1$ other jobs have arrived. Hence, the completion times for the first job significantly differ in the two cases. Since it is not immediate how to generalize the technique in \cite{BamasMRS20} without knowing $c_j$ at the time each job $j$ arrives, this motivated our choice of a different robustification technique. 

\paragraph{Smoothness and consistency technique.} To obtain smoothness and consistency guarantees, we use a similar technique as in \cite{BamasMRS20} (summing the speeds obtained by computing an offline schedule for the predicted jobs and an online schedule for the extra jobs), with two main differences: \\
(1) In \cite{BamasMRS20}, the extra jobs arriving at each time $i$ are scheduled uniformly over the next $D$ time units. 
% In particular, the decision for these jobs is taken at the time they arrive. 
On the other hand, our algorithm computes the speeds for all extra jobs by following an auxiliary online algorithm given as an input to the decision maker. 
% The decision for the extra jobs arriving at some time step may be finalized at a later time step. 
In fact, the technique from \cite{BamasMRS20} can be interpreted as a special case of our algorithm, where the auxiliary algorithm is the \textsc{AVERAGE RATE} heuristic \cite{yao1995scheduling}. \\
(2) The offline schedule we compute is conceptually identical to the one used  in \cite{BamasMRS20}, however, our online schedule differs, as it needs to integrate two different types of extra jobs: (1) the extra jobs that arrive during the second phase of the algorithm ($t\geq \tl$), and (2) the jobs that were not finished during the first phase of the algorithm. \cite{BamasMRS20} only needs to handle the first type of extra jobs. This results in a different analysis.

\end{document}